\UseRawInputEncoding
\documentclass[11pt,a4paper]{article}
\usepackage{amsmath,amsfonts,amssymb,amsthm,mathtools}
\usepackage{authblk}
\usepackage{latexsym,graphicx}
\usepackage{dsfont}
\usepackage{amscd}
\usepackage{xcolor}
\usepackage{extarrows} 
\usepackage{mathtools}
\usepackage{yfonts}
\usepackage{todonotes}
\usepackage{comment}

\numberwithin{equation}{section}

\newcommand{\tet}{\theta}

\newcommand{\vy}{\boldsymbol{y}}
\newcommand{\tx}{\tilde{x}}
\newcommand{\vtx}{\boldsymbol{\tilde{x}}}

\newcommand{\vty}{\boldsymbol{\tilde{y}}}

\newlength{\intwidth}
\DeclareRobustCommand{\fpint}[2]
   {\mathop{%
      \text{%
        \settowidth{\intwidth}{$\int$}%
        \makebox[0pt][l]{\makebox[\intwidth]{$-$}}%
        $\int_{#1}^{#2}$}}}

\newcommand{\xxa}{\stackrel {\scriptscriptstyle \times}{\scriptscriptstyle \times} \!}
\newcommand{\xxe}{\! \stackrel {\scriptscriptstyle \times}{\scriptscriptstyle \times}}

\newcommand{\nna}{:\! }
\newcommand{\nne}{\!:}

\newcommand{\F}{\mathrm{F}}

\newcommand{\hrho}{ \rho}

\newcommand{\sign}{{\rm sgn}}
\newcommand{\ee}{{\rm e}}
\newcommand{\ii}{{\rm i}}
\newcommand{\dd}{{\rm d}}

\newcommand{\id}{I}

\newcommand{\ttet}{\tilde{\tet}}

\newcommand{\eps}{\epsilon}

\newcommand{\R}{{\mathbb R}}
\newcommand{\C}{{\mathbb C}}
\newcommand{\Z}{{\mathbb Z}}

\newcommand{\cH}{{\mathcal  H}}

\newcommand{\cC}{{\mathcal  C}}
\newcommand{\cF}{{\mathcal  F}}

\newcommand{\cD}{\mathcal{D}} 
 
\newcommand{\cA}{\mathcal{A}}

\newcommand{\cU}{\mathcal{U}} 

\newcommand{\im}{\mathrm{Im}}
\newcommand{\re}{\mathrm{Re}}

\newcommand{\cR}{\mathcal{R}}

\newcommand{\pdag}{^{\phantom\dag}}

\newtheorem{theorem}{Theorem}[section]
\newtheorem{lemma}[theorem]{Lemma}
\newtheorem{proposition}[theorem]{Proposition}

\theoremstyle{definition}
\newtheorem{definition}[theorem]{Definition}
\theoremstyle{remark}
\newtheorem{remark}[theorem]{Remark}

\newcommand{\half}{\mbox{$\frac12$}}

\newcommand{\sgn}{{\rm sgn}} 

\newcommand{\vx}{{\boldsymbol{x}}}

\newcommand{\vr}{{\boldsymbol{r}}}
\newcommand{\vm}{{\boldsymbol{m}}}

\usepackage[colorlinks=true]{hyperref}
\hypersetup{urlcolor=blue, citecolor=red, linkcolor=blue}

\title{\Large{\bf{Conformal field theory, solitons, and \\ elliptic Calogero--Sutherland models}}}
\date{\vspace{-0.5cm}\small\vspace{-0.6cm}}

\author[1]{Bjorn K. Berntson\footnote{Current address: Riverlane Research, Cambridge, CB2 3BZ, United Kingdom}}
\author[1,2]{Edwin Langmann}
\author[3]{Jonatan Lenells}

\affil[1]{Department of Physics, KTH Royal Institute of Technology, 10691 Stockholm, Sweden}
\affil[2]{Nordita, KTH Royal Institute of Technology and Stockholm University, 10691 Stockholm, Sweden}
\affil[3]{Department of Mathematics, KTH Royal Institute of Technology, 10044 Stockholm, Sweden}

\begin{document}
\maketitle

\let\oldthefootnote\thefootnote
\renewcommand{\thefootnote}{\fnsymbol{footnote}}
\let\thefootnote\oldthefootnote

\begin{abstract}
We construct a non-chiral conformal field theory (CFT) on the torus that accommodates a second quantization of the elliptic Calogero-Sutherland (eCS) model. 
We show that the CFT operator that provides this second quantization defines, at the same time, a quantum version of a soliton equation called the non-chiral intermediate long-wave (ncILW) equation. 
We also show that this CFT operator is a second quantization of a generalized eCS model which can describe arbitrary numbers of four different kinds of particles; 
we propose that these particles can be identified with solitons of the quantum ncILW equation. 
\end{abstract} 

\newcommand{\hu}{\rho}
\newcommand{\hv}{\sigma} 


\section{Introduction} 
\label{sec:Intro} 
The present paper was inspired by the work of Abanov and Wiegmann revealing remarkable relations between conformal field theory, soliton equations, and quantum integrable systems of Calogero-Moser-Sutherland type \cite{abanov2005}; see also \cite{stone2007, stone2008,abanov2009}. 
In particular, we substantiate the proposal in \cite{abanov2009} that there should exist an interesting generalization of the Benjamin-Ono equation \cite{benjamin1967,ono1975} related to conformal field theory (CFT) and the elliptic Calogero-Sutherland (eCS) model\footnote{By {\em elliptic Calogero-Sutherland model} we mean the {\em quantum $A$-type elliptic Calogero-Moser-Sutherland model}.} (see \cite{olshanetsky1983} for a review of quantum Calogero-Moser-Sutherland systems).  
A natural candidate for this generalization would be the intermediate long-wave (ILW) equation \cite{joseph1977,kubota1978,kodama1981}, but this is not the equation we find. 
Instead, we obtain a two-component equation with ILW-type non-local terms;
this two-component equation was introduced in \cite{berntson2020} and called the (periodic\footnote{In the rest of this paper, {\em ncILW equation} is short for {\em periodic ncILW equation}.}) {\em non-chiral ILW} (ncILW) equation.
Our main result can be summarized as follows: {\em the same CFT operator that provides a second quantization of the eCS model \cite{langmann2000,langmann2004} defines also a quantum version of the ncILW equation}. We also show that this CFT operator is, in fact, a second quantization of a generalization of the eCS model that can describe arbitrary numbers of four different kinds of quantum particles, extending previous results in \cite{atai2017} for the trigonometric case. Some of the results in this paper were announced in \cite{berntson2020}, and the integrability properties of the classical version of the ncILW equation were established in \cite{berntson2022,berntsonlangmann2022}. 

Throughout this paper, $\ell>0$, $\delta>0$, and $g>0$ are fixed constants. As explained below, these three parameters appear both in the ncILW equation and the eCS model. 
For reasons explained in Section~\ref{sec:eCSgen}, our general results are restricted to rational values of the parameter $g$. 
We also use the abbreviations 
\begin{equation}
\label{constants}  
\kappa\coloneqq \frac{\pi}{2\ell},\quad q\coloneqq \ee^{-2\kappa \delta},\quad 
c_0\coloneqq  \frac1{3}\kappa^2-8\kappa^2\sum_{n=1}^{\infty}\frac{{n}q^{2n}}{1-q^{2n}}.
\end{equation} 
We use units such that Planck's constant, $\hbar$, and the mass of the eCS particles are set to 1. 

The ncILW equation describes the time evolution of two $2\ell$-periodic scalar functions $u$ and $v$ of position $x\in \R$ and time $t\in  \R$ (i.e., $u=u(x,t)=u(x+2\ell,t)$ and similarly for $v$)  as follows, 
\begin{equation} 
\label{ncILW} 
\begin{split} 
&u_t + 2uu_x+\frac{g}{2}( Tu_{xx} +\tilde{T}v_{xx} ) =0,\\
&v_t - 2vv_x- \frac{g}{2}( Tv_{xx} +\tilde{T}u_{xx} )=0,
\end{split}
\end{equation} 
where subscripts denote partial derivatives and $T, \tilde T$ are integral transforms defined by
\begin{equation}
\label{TT}
\begin{split} 
&(Tf)(x) \coloneqq  \frac1{\pi}\fpint{-\ell}{\ell} \zeta_1(x'-x)f(x')\,\dd{x}',\\
&(\tilde{T}f)(x) \coloneqq \frac1{\pi}\int_{-\ell}^{\ell} \zeta_1(x'-x+\ii\delta)f(x')\,\dd{x}',
\end{split} 
\end{equation}
with
\begin{equation} 
\label{zeta1} 
\zeta_1(z) \coloneqq \lim_{M\to\infty}\sum_{m=-M}^M \kappa\cot(\kappa(z-2\ii m\delta)) \qquad (z\in\C).
\end{equation}
Since the partial derivative $\frac{\partial}{\partial x}$ commutes with the integral operators $T$ and $\tilde{T}$ when acting on zero-mean functions, the term $Tu_{xx}$ in \eqref{ncILW} can be interpreted either as $(Tu)_{xx}$ or as $T(u_{xx})$, and similarly for $\tilde{T}$.
The function $\zeta_1$ defined in \eqref{zeta1} is a $2\ell$-periodic variant of the  Weierstrass  $\zeta$-function with half-periods $(\ell,\ii\delta)$ (see Appendix~\ref{app:special}).
While one can set $g=2$ without loss of generality in the classical case,\footnote{This is true since \eqref{ncILW} is invariant under $(u(x,t),v(x,t))\to (su(x,t/s),sv(x,t/s))$, $g\to sg$ ($s>0$).} $g$ is an important parameter in the quantum case; the parameter $\delta>0$ is important in both cases with the trigonometric case obtained in the limit $\delta\to\infty$. 
The parameter $\ell$ sets a length scale and, for most of what we do in the present paper, one can set $\ell=\pi$; however, since $\ell\to\infty$ is an important limiting case from a physics point of view, we keep $\ell$ in our equations. 
To ease notation, we also use the redundant parameter $\kappa$ defined in \eqref{constants}. 
Note that $\kappa=1/2$ for $\ell=\pi$, $\zeta_1(z)\to (\pi/\delta)\coth(\pi z/\delta)$ as $\kappa\to 0$, and $\zeta_1(z+\ii\delta)\to (\pi/\delta)\tanh(\pi z/\delta)$ as $\ell\to\infty$. 

The eCS model is defined by the Hamiltonian\footnote{To be more precise, the Schr\"odinger operator in \eqref{eCS} defines a symmetric operator on a suitable domain, and the eCS Hamiltonian is a particular self-adjoint extension of this symmetric operator.}
\begin{equation} 
\label{eCS} 
H_{N;g}(\vx)\coloneqq -\sum_{j=1}^N \frac12 \frac{\partial^2}{\partial x_j^2} + \sum_{1\leq j<k\leq N} g(g-1) \wp_1(x_j-x_k) 
\end{equation} 
where $\wp_1(z) \coloneqq -\zeta_1'(z)$ is equal to  the Weierstrass elliptic $\wp$-function with half-periods $(\ell,\ii\delta)$, up to an additive constant (see Appendix~\ref{app:special}), with variables $\vx=(x_1,\ldots,x_N)$ on the torus $[-\ell,\ell]^N$ for arbitrary $N\in\Z_{\geq 1}$. We show that an operator constructed in \cite{langmann2000,langmann2004} as a second quantization of the eCS Hamiltonian \eqref{eCS} allows for a natural generalization which defines a quantum version of the ncILW equation \eqref{ncILW}. We also show that this operator is, in fact, a second quantization of a generalization of the eCS model where each particle $j=1,\ldots,N$ has two labels $r_j=\pm$ and $m_j\in\{1,-1/g\}$, and where the corresponding Schr\"odinger operator can be written as
\begin{equation} 
\label{eCS2} 
H^{(\vr,\vm)}_{N;g}(\vx)\coloneqq -\sum_{j=1}^N \frac1{2m_j} \frac{\partial^2}{\partial x_j^2} + \sum_{1\leq j<k\leq N} m_jm_kg(g-1)\wp_{r_j,r_k}(x_j-x_k) 
\end{equation} 
with $\vr=(r_1,\ldots,r_N)\in\{\pm\}^N$, $\vm=(m_1,\ldots,m_N)\in\{1,-1/g\}^N$, and
\begin{equation} 
\wp_{r,r'}(x)\coloneqq \begin{cases} \wp_1(x) & (r=r') \\ \wp_1(x+\ii\delta) & (r=-r'). \end{cases} 
\end{equation} 
Clearly, the standard eCS Hamiltonian \eqref{eCS} corresponds to the special case where all particles have the same labels: $(r_j,m_j)=(+,1)$ for $j=1,\ldots,N$. 
Another interesting special case is $(x_j,r_j,m_j)=(x_j,+,1)$ for $j=1,\ldots,N_1<N$ and $(x_j,r_j,m_j)=(\tx_{j-N_1},+,-1/g)$ for $j=N_1+1,\ldots,N_1+M_1$ with $M_1=N-N_1$,  reducing the operator in \eqref{eCS2} to  (we rename $(N_1,M_1)\to (N,M)$)
\begin{equation} 
\label{deCS} 
H_{N,M;g}(\vx,\vtx) = H_{N;g}(\vx)-gH_{M;1/g}(\vtx) +  \sum_{j=1}^{N} \sum_{k=1}^{M} (1-g) \wp_1(x_j-\tilde{x}_k) .
\end{equation} 
This operator defines the deformed eCS model introduced by Sergeev and Veselov \cite{sergeev2004} which is expected to be integrable (as discussed in Appendix~\ref{app:int}, the current status of this integrability question is somewhat complicated). In general, one can write the operator in \eqref{eCS2} as 
\begin{multline} 
\label{eCSgen}
H_{N_1,M_1,N_2,M_2;g}(\vx,\vtx,\vy,\vty) = H_{N_1,M_1;g}(\vx,\vtx) + H_{N_2,M_2;g}(\vy,\vty) \\  + 
\sum_{j=1}^{N_1}\sum_{k=1}^{N_2}g(g-1)\wp_1(x_j-y_k+\ii\delta) +  \sum_{j=1}^{M_1}\sum_{k=1}^{M_2}g(g-1)\wp_1(\tilde{x}_j-\tilde{y}_k+\ii\delta) \\  + 
\sum_{j=1}^{N_1}\sum_{k=1}^{M_2}(1-g)\wp_1(\tilde{x}_j-y_k+\ii\delta) +  \sum_{j=1}^{N_1}\sum_{k=1}^{M_2}(1-g)\wp_1(x_j-\tilde{y}_k+\ii\delta),     
\end{multline} 
where $N_1,M_1,N_2,M_2\in\Z_{\geq 0}$ are the number of particles with labels $(r_j,m_j)=(+,1)$, $(-,1)$, $(+,-1/g)$, $(-,-1/g)$, respectively, and $N_1+M_1+N_2+M_2=N$. 
This defines a generalization of the deformed eCS model describing arbitrary numbers of four kinds of different particles which we expect to be integrable (arguments in support of this conjecture can be found in Appendix~\ref{app:int}). 
The label $r_j$ can be interpreted as a chirality index with $r_j=-$ and $+$ corresponding to left- and right-movers, respectively, and the label $m_j$ distinguishes two different particle types with $m_j=1$ and $-1/g$ corresponding to electrons and holes, respectively (the latter is a condensed matter physics interpretation elaborated in \cite{berntson2020}). All particles interact via two-body interactions which depend on the particle distances, $|x|$, and while the interaction potential of particles of the same chirality is proportional to $\wp_1(x)$, it is proportional to $\wp_1(x+\ii\delta)$ for particles of opposite chirality;\footnote{Note that $\wp_1(x)=\wp_1(-x)$ and $\wp_1(x+\ii\delta)=\wp_1(x-\ii\delta)$.} the former is singular as $x\to 0$ and repulsive, while the latter is non-singular and attractive (note that $\wp_1(x)\to (\pi/\delta)^2/\sinh(\pi x/\delta)^2$ and $\wp_1(x+\ii\delta)\to -(\pi/\delta)^2/\cosh(\pi x/\delta)^2$ as $\ell\to \infty$). Moreover, in the limit $\delta\to\infty$, $\wp_1(x)\to \kappa^2/\sin(\kappa x)^2$ and $\wp_1(x+\ii\delta)\to 0$; thus, interactions between particles of different chiralities vanish in this limit, and the operator in \eqref{eCSgen} reduces to a sum of two commuting deformed trigonometric Calogero-Sutherland Hamiltonians.  

Our results have applications in different areas of physics. Two condensed matter physics applications in the context of the fractional quantum Hall effect and the hydrodynamic descriptions of interacting quantum many-body systems are proposed in Ref.~\cite{berntson2020}. A third application is related to particle physics: as discussed in our conclusions in Section~\ref{sec:conclusions}, our results suggest that the non-chiral CFT constructed in this paper corresponds to a non-relativistic limit of quantum sine-Gordon theory studied extensively in the mathematical physics literature after seminal work by Coleman \cite{coleman1975} and Mandelstam \cite{mandelstam1975}; see \cite{bauerschmidt2020} for recent work on this topic. If this is true, our results provide a non-relativistic version of the Coleman correspondence between quantum sine-Gordon theory and the massive Thirring model \cite{coleman1975}.

\noindent {\bf Organization of the paper.} 
An overview of our main results is provided in Section \ref{sec:results}. 
In Section \ref{sec:prerequisites}, we introduce the Fock space, Heisenberg algebras, and vertex operators that underlie our constructions. In Section \ref{sec:anyonseCS}, we define anyons as special cases of vertex operators and introduce a CFT operator, denoted by $\cH_{3,\nu}$, which is the central object of this paper. By generalizing the construction of \cite{langmann2004}, we show that $\cH_{3,\nu}$ provides a second quantization of the eCS model.
In Section \ref{sec:ncILW}, we prove our first main result (Theorem \ref{thm:qncILW}), which states that the operator $\cH_{3,\nu}$ also defines a quantum version of the ncILW equation.
In Section \ref{sec:eCSgen}, we prove our second main result (Theorem \ref{thm:eCSgen}), which shows that the operator $\cH_{3,\nu}$ also provides a second quantization of a generalized eCS model describing arbitrary numbers of four types of particles. 
Conclusions are drawn in Section \ref{sec:conclusions}. 
In Appendix \ref{app:special}, we collect some definitions and properties of certain special functions. In Appendix \ref{app:int}, we show that the generalized eCS model is quantum integrable if and only if the deformed eCS model \eqref{deCS} is.
Appendix \ref{app:proofs} contains proofs of some results stated in the main text.
In Appendix \ref{app:fermions}, we use boson-fermion correspondence to derive a fermion representation for the operator $\cH_{3,\nu}$.

\noindent {\bf Notation.} We denote as $\Z$, $\R$, and $\C$ the sets of all integers, real numbers, and complex numbers, respectively. 
We denote as $\Z_{\neq 0}$, $\Z_{>0}$, $\Z_{\geq 0}$ the sets of non-zero, positive, and non-negative integers, respectively, and similarly for $\R$.  
We write $\ii\coloneqq \sqrt{-1}$, and $\bar z$ is the complex conjugate of $z\in\C$.
We sometimes write $\partial_x$ instead of $\frac{\partial}{\partial x}$ etc. We write $r=\pm$ short for $r=+,-$ and $\{\pm\}$ short for the set $\{+,-\}$. 
 We use the common abbreviations for commutators and anti-commutators: $[A,B]\coloneqq AB-BA$ and $\{A,B\}\coloneqq AB+BA$.

\section{Summary of results}
\label{sec:results}
We summarize our results, suppressing some technical details which we discuss later to make the results presented here mathematically precise.  

\subsection{Construction of anyons}
\label{subsec:anyons} 
Loosely speaking, by {\em anyons} we mean quantum fields $\phi_\nu(x)$ labeled by a real statistics parameter $\nu$ and a position variable $x\in[-\ell,\ell]$, which act on some Fock space $\cF$ with vacuum $\Omega$, and which obey the (formal) exchange relations
\begin{equation}\label{phinuphinu} 
\phi_{\nu}(x)\phi_{\nu'}(x') = \ee^{\mp \ii \pi \nu\nu' } \phi_{\nu'}(x') \phi_{\nu}(x)
\end{equation} 
for $x\gtrless x'$, together with the relations $\phi_\nu(x)^\dag=\phi_{-\nu}(x)$ and the condition   
\begin{equation} 
\label{normalization} 
\lim_{x\to y}2\pi (x-x')^{\nu^2}\langle\Omega,\phi_\nu(x)\phi_\nu(x')^\dag\Omega\rangle =1, 
\end{equation} 
where $\dag$ and $\langle\cdot,\cdot\rangle$ are the Hilbert space adjoint and inner product in $\cF$, respectively; see \cite{carey1999}.
In particular, for $\nu^2=1$, the anyons are (standard chiral) fermions, for $\nu^2$ odd integers $\geq 3$ they are composite fermions, and for $\nu^2$ even integers they are bosons. 

Our starting point is a known second quantization\footnote{To avoid misunderstanding, we stress that this is different from conventional second quantization discussed in textbooks on quantum-many body physics.} of the eCS model given by a Hermitian operator $\cH_{3,\nu}$ on $\cF$ and characterized by the following commutator relations with products of such anyon operators, 
\begin{equation} 
\label{2nd} 
[\cH_{3,\nu},\phi_\nu(x_1)\cdots \phi_\nu(x_N)]\Omega = \left( H_{N,\nu^2}(\vx) + \half N\nu^4 c_0\right) \phi_\nu(x_1)\cdots \phi_\nu(x_N)\Omega
\end{equation} 
for arbitrary $N\in\Z_{\geq 1}$, with $H_{N,\nu^2}(\vx)$ the eCS Hamiltonian in \eqref{eCS} for $g=\nu^2$ \cite{langmann2000,langmann2004}. 
Note that this is similar to conventional second quantization in that one operator, $\cH_{3,\nu}$, acting on a Fock space accounts for an arbitrary number, $N$, of particles in a quantum mechanical model; however,  it is more powerful since it is naturally adapted to the integrability of the underlying eCS model. 

In the trigonometric limit $\delta\to\infty$, one can construct such a second quantization using the Fock space of chiral fermions in 1+1 spacetime dimensions, $\cF=\cF_{\mathrm{c}}$ \cite{carey1999}.  
For finite $\delta$, a Fock space $\cF$ accommodating such a second quantization can be obtained as a subspace of the tensor product of two copies of this chiral Fock space, $\cF_{\mathrm{c}}\otimes\cF_{\mathrm{c}}$, and this construction has a physical interpretation as a CFT at finite temperature $1/2\delta$ \cite{langmann2004} (note that $2\delta$ here corresponds to the parameter $\beta$ in \cite{langmann2004}). 
The key to the results in the present paper is a different physical interpretation: the latter Fock space also accommodates a non-chiral CFT, i.e., the chiral fermions in  $\cF_{\mathrm{c}}\otimes\cF_{\mathrm{c}}$ can be naturally combined into non-chiral (Dirac) fermions, and using the full Fock space $\cF_{\mathrm{c}}\otimes\cF_{\mathrm{c}}$, one not only has the anyons $\phi_\nu(x)$ described above but, in addition, algebraically independent anyons which we denote as $\phi_{-,\nu}(x)$. Moreover, we realized that there is a natural extension of $\cH_{3,\nu}$ such that \eqref{2nd} remains true as it stands if the anyons $\phi_\nu(x_j)$ are replaced by $\phi_{-,\nu}(x_j)$ for $j=1,\ldots,N$. 

Thus, the conceptual change (as compared to \cite{langmann2004}) is that we work on the larger Fock space $\cF=\cF_{\mathrm{c}}\otimes\cF_{\mathrm{c}}$ and, by that, increase the degrees of freedom. As discussed in Remark \ref{rem:vacua}, this change is natural from a condensed matter physics point of view since it allows us to interpret this non-chiral CFT as a Luttinger model \cite{mattis1965}.  
To emphasize that both anyon operators on this larger space $\cF$ are on equal footing, we use the symbol $\phi_{+,\nu}(x)$ instead of $\phi_\nu(x)$ in the following, allowing us to write $\phi_{r,\nu}(x)$ for both anyons, with $r=\pm$ a chirality index. We also use the shorthand notation 
\begin{equation} 
\label{phiN} 
\phi^N_{r,\nu}(\vx)\coloneqq \phi_{r,\nu}(x_1)\cdots \phi_{r,\nu}(x_N)\quad (r=\pm, \; N\in\Z_{\geq 0}),
\end{equation} 
where $\phi^0_{r,\nu}(\vx)\coloneqq I$ is the identity operator. This allows us to write these two second quantizations of the eCS model on $\cF$ as follows, 
\begin{equation} 
\label{2nd2} 
[\cH_{3,\nu},\phi^N_{r,\nu}(\vx)]\Omega = \left( H_{N;\nu^2}(\vx)+\half N\nu^4 c_0\right) \phi^N_{r,\nu}(\vx)\Omega \quad (r=\pm). 
 \end{equation} 
 
 To introduce further notation, we recall that one can generate the fermion Fock space $\cF_{\mathrm{c}}\otimes\cF_{\mathrm{c}}$ from a vacuum $\Omega$ using chiral bosons $\rho_\pm(x)$ with the (formal) commutator relations 
\begin{equation} 
\label{CCR} 
[\rho_r(x),\rho_{r'}(x')]=-2\pi\ii r\delta_{r,r'}\partial_x\delta(x-x')\quad (r,r'=\pm) 
\end{equation} 
for $x,x'\in[-\ell,\ell]$, with $\partial_x\coloneqq \frac{\partial}{\partial x}$, $\delta_{r,r'}$ the Kronecker delta, and $\delta(x-x')$ the $2\ell$-periodic Dirac delta.\footnote{We slightly abuse terminology since, strictly speaking, the chiral bosons are $\partial_x^{-1}\rho_\pm(x)$ with the anti-derivative hiding zero modes (including so-called Klein factors) which are important to make them mathematically precise.}  Using these chiral bosons, the anyon field operators discussed above are (formally) given by 
\begin{equation}
\label{phirnu}
\phi_{r,\nu}(x) = \; \xxa \ee^{-\ii r\nu \partial_x^{-1}\rho_r(x)}\xxe\quad (r=\pm)
\end{equation} 
where, here and in the following, $\xxa\cdots\xxe$ indicates normal ordering; see Definition~\ref{def:anyons}.

\subsection{Quantum ncILW equation}
Our first main result (Theorem \ref{thm:qncILW}) shows that the CFT operator $\cH_{3,\nu}$ discussed above can be written as 
\begin{equation} 
\label{cH3nu}
\cH_{3,\nu} = \frac1{4\pi}\int_{-\ell}^{\ell}\sum_{r=\pm} \xxa \bigg[\frac{\nu}{3}\rho_r^3  + \frac{(\nu^2-1)}{2}\big(  \rho_r T\hrho_{r,x} + \rho_{-r}\tilde{T}\rho_{r,x}  \bigr)\bigg]\xxe \dd{x} 
\end{equation} 
with the integral operators $T$ and $\tilde{T}$ defined in \eqref{TT}, and $\rho_r$ and $\rho_{r,x}$ short for $\rho_r(x)$ and $\partial_x\rho_r(x)$, respectively; this formula is made precise in Theorem \ref{thm:qncILW}$(a)$. This defines a quantum version of the ncILW equation \eqref{ncILW}; to see this, we compute the Heisenberg equations of motion $\rho_{r,t}=\ii[\cH_{3,\nu},\rho_r]$ ($r=\pm$), and by changing the normalization: $\hat u \coloneqq \nu\rho_+/2$ and $\hat v \coloneqq \nu\rho_-/2$, we obtain 
\begin{equation} 
\label{q-ncILW} 
\begin{split} 
\hat{u}_t &+ 2\xxa\hat{u}\hat{u}_x\xxe + \frac12(\nu^2-1)[T\hat{u}_{xx}+\tilde{T}\hat{v}_{xx}]=0,\\
\hat{v}_t &- 2\xxa\hat{v}\hat{v}_x\xxe - \frac12(\nu^2-1)[T\hat{v}_{xx}+\tilde{T}\hat{u}_{xx}]=0, 
\end{split} 
\end{equation} 
where the hats on $u$ and $v$ distinguish the variables from their classical counterparts in \eqref{ncILW}; see Theorem~\ref{thm:qncILW}$(b)$ for a precise formulation. In our units, $\nu^2-1$ in \eqref{q-ncILW} should be interpreted as $\nu^2-\hbar$ with Planck's constant $\hbar$ and thus, in the classical limit $\hbar\to 0$, where the operators $\hat u$ and $\hat v$ become functions and normal ordering $\xxa\cdots\xxe$ can be ignored, \eqref{q-ncILW} reduces to the ncILW equation in \eqref{ncILW} with $g=\nu^2$. 
We recently proved that the classical ncILW equation in \eqref{ncILW} is an integrable soliton equation \cite{berntson2022,berntsonlangmann2022}. As discussed further below, we conjecture that the quantum version of this model defined by \eqref{cH3nu} is integrable as well. 

\subsection{Second quantization of a generalized eCS model}
Our second main result (Theorem~\ref{thm:eCSgen}) concerns a generalization of the eCS model.
To motivate this result, we note that $\cH_{3,\nu}$ in \eqref{cH3nu} obeys 
\begin{equation} 
\label{cHsymmetry} 
\cH_{3,\nu} = -\nu^2 \cH_{3,-1/\nu}. 
\end{equation} 
This suggests to replace $\nu\to -1/\nu$ in \eqref{2nd2} and then use \eqref{cHsymmetry} to replace $\cH_{3,-1/\nu}$ by $-\cH_{3,\nu}/\nu^2$  to obtain  
\begin{equation} 
\label{2nd3} 
-\frac1{\nu^2} [\cH^{3,\nu},\phi^N_{r,-1/\nu}(\vx)]\Omega= \left( H_{N;1/\nu^2}(\vx)+\half(N/\nu^{4})c_0 \right) \phi^N_{r,-1/\nu}(\vx)\Omega\quad (r=\pm). 
\end{equation} 
As we will show, this is true provided that the coupling parameter, $g>0$,  is a rational number and thus, in such a case, the operator $\cH_{3,\nu}$ provides a four-fold second quantization of the eCS model. 
In fact, we prove the following general result including all these second quantizations of the eCS model as special cases: 
\begin{multline} 
\label{2ndgen}
[\cH_{3,\nu},\phi^{N_1}_{+,\nu}(\vx)\phi^{M_1}_{+,-1/\nu}(\vtx)\phi^{N_2}_{-,\nu}(\vy)\phi^{M_2}_{-,-1/\nu}(\vty)]\Omega \\ = 
\left( H_{N_1,M_1,N_2,M_2;\nu^2}(\vx,\vtx,\vy,\vty)+c_{N_1,M_1,N_2,M_2;\nu^2}\right) \phi^{N_1}_{+,\nu}(\vx)\phi^{M_1}_{+,-1/\nu}(\vtx)\phi^{N_2}_{-,\nu}(\vy)\phi^{M_2}_{-,-1/\nu}(\vty)\Omega
\end{multline} 
for arbitrary $N_1,M_1,N_2,M_2\in\Z_{\geq 0}$, with $H_{N_1,M_1,N_2,M_2;g}(\vx,\vtx,\vy,\vty)$ the generalized eCS Hamiltonian in \eqref{eCSgen} and the constant 
\begin{equation} 
\label{cNMNM}
c_{N_1,M_1,N_2,M_2;g} \coloneqq
\half\big((N_1+N_2)g^2-(M_1+M_2)/g\big)c_0. 
\end{equation} 
It is interesting to note that by expressing the generalized eCS Hamiltonian in the form \eqref{eCS2}, we can write \eqref{2ndgen}--\eqref{cNMNM} as 
\begin{equation} 
[\cH_{3,\nu},\phi_{r_1,m_1\nu}(x_1)\cdots \phi_{r_N,m_N\nu}(x_N)]\Omega 
= \left( H^{(\vr,\vm)}_{N;\nu^2}(\vx) + c^{(\vm)}_{N;\nu}\right)\phi_{r_1,m_1\nu}(x_1)\cdots \phi_{r_N,m_N\nu}(x_N)\Omega 
\end{equation}  
with
\begin{equation} 
c^{(\vm)}_{N;\nu} = \half \nu \sum_{j=1}^N (\nu m_j)^3  c_0. 
\end{equation} 
This makes manifest that each particle in the generalized eCS Hamiltonian \eqref{eCS2} corresponds to an anyon $\phi_{r_j,m_j\nu}(x_j)$ for $r_j=\pm$ and $m_j=1,-1/g$. 
This notation is not only useful for stating the result concisely but also for proving it; see Theorem~\ref{thm:eCSgen} for the precise statement. 

\section{Prerequisites} \label{sec:prerequisites}
We follow \cite[Section~2]{langmann2004}, but with some notational changes; see Remark~\ref{rem:notation1} and Remark~\ref{rem:notation2}. 

\subsection{Fock space, Klein factors and, Heisenberg algebras}
\label{subsec:cFcA} 
We define the Fock space, $\cF$,  and the algebra of quantum field theory operators, $\cA$,  underlying our constructions.

We consider the $*$-algebra $\cA$ with identity $\id$ and star operation $\dag$, generated by operators $a_{r,n}$ and $R_r$ ($r=\pm, n\in\Z)$ and characterized by the relations 
\begin{equation} 
\label{aR1}
\begin{split} 
[a_{r,n},a_{r',m}]=n\delta_{n,-m}\delta_{r,r'}\id,\quad [a_{r,n},R_{r'}] = \delta_{n,0}\delta_{r,r'}R_{r'} , \quad R_+R_-=-R_-R_+ \end{split} 
\end{equation} 
together with 
\begin{equation} 
\label{aR2} 
\quad a_{r,n}^\dag =a_{r,-n},\quad R_r^\dag = R_r^{-1} 
\end{equation} 
for all $r,r'=\pm$ and $n,m\in\Z$. 
We assume that this algebra $\cA$ is represented on a Hilbert space $\cF$ such that the following conditions are fulfilled:  (i) for $A$ an operator on $\cF$, $A^\dag$ is the Hilbert space adjoint, (ii) the following highest weight conditions are fulfilled, 
\begin{equation} 
\label{aR3} 
a_{\pm,n}\Omega=0\quad (n\in\Z_{\geq 0})
\end{equation} 
with $\Omega\in\cF$ the vacuum, (iii) the Hilbert space product $\langle\cdot,\cdot\rangle$ on $\cF$ is such that 
\begin{equation} 
\label{aR4} 
\langle\Omega,R_+^{\mu_+}  R_-^{\mu_-}\Omega\rangle = \delta_{\mu_+,0}\delta_{\mu_-,0}\quad (\mu_\pm\in \Z).
\end{equation} 
These conditions not only fully characterize the representation of $\cA$, but at the same time provide a means to construct the Hilbert space $\cF$. 
Indeed, using \eqref{aR1}--\eqref{aR4}, one can check that the elements  
\begin{equation}
\label{eta} 
\eta= \prod_{r=\pm } \prod_{n=1}^{\infty}  \frac{a_{r,-n}^{m_{r,n}}}{\sqrt{m_{r,n}! \,n^{m_{r,n}}}} \cdot  R_+^{\mu_+}R_-^{\mu_-}
\Omega \quad (m_{r,n}\in \Z_{\geq 0},\mu_r\in \Z)
\end{equation} 
where only finitely many coefficients $m_{r,n}$ are allowed to be non-zero, are orthonormal. 
(We slightly abuse notation here by identifying the algebra $\cA$ with its representation on $\cF$). 
Denoting by $\cD$ the space of finite linear combinations of such elements $\eta$ with complex coefficients, the Hilbert space $\cF$ is obtained from $\cD$ by norm completion. 
These definitions imply that $R_+$ and $R_-$ are unitary operators on $\cF$ such that $R_\pm^\mu$ is well-defined if and only if $\mu\in\Z$; we refer to $R_\pm$ as {\em Klein factors} (see \cite{langmann2015} for further details).
We note the relations
\begin{equation} 
\label{RRRR} 
R_+^{\mu_+} R_-^{\mu_-} R_+^{\nu_+} R_-^{\nu_-} = (-1)^{\mu_-\nu_+}R_+^{\mu_++\nu_+}R_-^{\mu_-+\nu_-}= 
(-1)^{\mu_-\nu_+ - \mu_+\nu_-}R_+^{\nu_+} R_-^{\nu_-} R_+^{\mu_+} R_-^{\mu_-}\
\end{equation} 
for arbitrary $\mu_\pm,\nu_\pm\in\Z$, and $R_\pm^0=I$. 
Moreover, the elements $a_{\pm,n}$ are the generators of two commuting Heisenberg algebras. 

We introduce the operators
\begin{equation} 
Q_\pm\coloneqq \nu_0 a_{\pm,0} 
\end{equation} 
with $\nu_0>0$ a constant specified further below; see \eqref{eq:r0s0}.
Since all $\eta$ in \eqref{eta} are eigenstates of $Q_\pm$ with eigenvalues $\nu_0\mu_\pm$, $Q_\pm$ are self-adjoint operators with eigenvalues in $\nu_0\Z$.  
Thus, for all $\alpha\in\R$,  $\ee^{\ii\alpha Q_\pm}$ are well-defined unitary operators on $\cF$.
We refer to the $Q_\pm$ as {\em charges}. The Klein factors are charge-raising operators in the sense that
\begin{equation}
\label{RQ}
[Q_r,R_{r'}]=\delta_{r,r'} \nu_0 R_{r'}\quad (r,r'=\pm) 
\end{equation}
where $\delta_{r,r'}$ is the Kronecker delta. 
We also note the identity 
\begin{equation} 
\label{expQR}
\ee^{\ii\alpha Q_r}R_{r'}^{\mu}=\ee^{\ii\delta_{r,r'}\alpha\mu\nu_0}R_{r'}^{\mu}\ee^{\ii\alpha Q_r}\quad (\alpha\in\R, \, \mu\in\Z, \, r= \pm, \, r'=\pm),
\end{equation} 
which follows from the commutator relations in \eqref{RQ}.

\begin{remark} 
\label{rem:notation1}
We use the index $r=\pm$ here instead of $A=1,2$ in \cite{langmann2004} to emphasize a different physical interpretation, as discussed in Section~\ref{subsec:anyons2}.
Note that $a_{+,n}$ and $a_{-,n}$ here correspond to $\hat{\rho}_{1}(n)$ and $-\hat{\rho}_{2}(n)$ in \cite{langmann2004}, respectively; the implications of the minus sign in the latter relation are explained in Remark~\ref{rem:notation2}. 
\end{remark} 

\subsection{Normal ordering}
\label{subsec:normalordering} 
We discuss a normal ordering scheme in $\cA$ and state some related technical results we need. 

\begin{definition}
\label{def:normalorder} 
The normal ordering operation $\xxa\cdot\xxe$ is a linear map $\cA\to \cA$ such that
\begin{equation}\label{normalordering0}  
\xxa AB \xxe \; = \; \xxa BA\xxe \quad \quad (A,B\in\cA),
\end{equation} 
defined as follows:
on monomial elements $M= R_+^{\mu_+} R_-^{\mu_-} a_{r_1,n_1}\cdots a_{r_k,n_k}$, 
$\mu_\pm\in\Z$, $k\in \Z_{\geq 0}$, $n_1,\ldots,n_k\in \mathbb{Z}$ and $r_1,\ldots,r_k\in \{\pm \}$,   
it is defined inductively by the rules
\begin{equation}\label{normalordering1}  
\xxa R_+^{\mu_+} R_-^{\mu_-}\xxe \; \coloneqq R_+^{\mu_+} R_-^{\mu_-} \quad (\mu_+,\mu_-\in\Z) 
\end{equation} 
and  
\begin{equation}
\label{normalordering2}
\begin{split}
\xxa  M a_{\pm,n} \xxe \; =\; \xxa  a_{\pm,n} M \xxe  &\coloneqq \begin{cases}
\xxa M \xxe a_{\pm,n} & (n\in\Z_{>0}) \\
\frac12\big(  \xxa M \xxe a_{\pm,0}+ a_{\pm,0} \xxa M\xxe  \big) & (n=0) \\
a_{\pm,n}\xxa M \xxe  & (n\in\Z_{<0}),
\end{cases}
\end{split}
\end{equation} 
and this definition is extended to non-monomial elements of $\cA$ by linearity.
\end{definition}

This definition implies that $\xxa A B\xxe \; = \; \xxa \xxa A\xxe \xxa B\xxe \xxe $, i.e., within the normal ordering operation one can insert further normal orderings without changing the result. 
For later reference we note that \eqref{normalordering0} and $R_-R_+=-R_+R_-$ imply 
\begin{equation}\label{RRRR1} 
\xxa R_-^{\mu_-} R_+^{\mu_+}\xxe  = R_+^{\mu_+}R_-^{\mu_-}=(-1)^{\mu_+\mu_-}R_-^{\mu_-} R_+^{\mu_+} \quad(\mu_+,\mu_-\in\Z). 
\end{equation} 
Definition \ref{def:normalorder} and \eqref{aR1}--\eqref{aR4} imply that, for fixed vectors $\eta,\eta'\in\mathcal{D}$, the expression
\begin{equation*}
\langle \eta, \xxa a_{r_1,n_1}\cdots a_{r_k,n_k}\xxe \eta'\rangle
\end{equation*}
is non-zero for only finitely many choices of $k\in \Z_{\geq 0}$, $n_1,\ldots,n_k\in\Z_{\neq 0}$, and $r_1,\ldots,r_k\in \{\pm\}$. It follows that the expression
\begin{equation}\label{Sdef}
S\coloneqq s_0\, \id +\sum_{k=1}^{\infty}\sum_{r_1,\ldots,r_k=\pm}\sum_{n_1,\ldots,n_k\in \Z_{\neq 0}} s^{r_1,\ldots,r_k}_{n_1,\ldots,n_k} \xxa a_{r_1,n_1}\cdots a_{r_k,n_k}\xxe,  
\end{equation}
for arbitrary complex numbers $s_0$ and $s^{r_1,\ldots,r_k}_{n_1,\ldots,n_k}$, defines a sesquilinear form on $\mathcal{D}$ according to $(\eta,\eta')\mapsto \langle \eta,S\eta'\rangle$. 

The following result ensures that all quantum field theory operators we use are well-defined. 
\begin{lemma}
\label{lem:normalordering}
For arbitrary $\mu_\pm\in \Z$, $\alpha_\pm \in\C$, and $S=\; \xxa S\xxe$ as in \eqref{Sdef},
\begin{equation}
\label{Rmu1Rmu2ealphaQ}
\xxa R_+^{\mu_+} R_-^{\mu_-} \ee^{\alpha_+ Q_+}\ee^{\alpha_- Q_-} S\xxe \; = 
\ee^{\alpha_+ Q_+/2}\ee^{\alpha_- Q_-/2} R^{\mu_+}_+R^{\mu_-}_- \ee^{\alpha_+ Q_+/2}\ee^{\alpha_- Q_-/2} S
\end{equation}
defines a sesquilinear form on $\cD$. 
\end{lemma}

\begin{proof}
Let $W$ be such that $\xxa W \xxe \; =  W$. 
Applying the identity 
\begin{equation*}
\xxa a_{\pm,0}^{n} W \xxe
= \frac{1}{2}\big(\xxa a_{\pm,0}^{n-1} W \xxe a_{\pm,0} + a_{\pm,0}\xxa a_{\pm,0}^{n-1} W \xxe\big)
\end{equation*}
repeatedly, we obtain
\begin{equation*}
\xxa a_{\pm,0}^{n} W \xxe \; = \frac{1}{2^n} \sum_{j=0}^n \binom{n}{j} a_{\pm,0}^j W a_{\pm,0}^{n-j}
\end{equation*}
with the usual binomial coefficients $\binom{n}{j}$. Hence, for arbitrary $\alpha\in \C$, using the Taylor series of the exponential function,
\begin{align*}
\xxa \ee^{\alpha Q_\pm}  W\xxe
& =  \xxa \sum_{n=0}^\infty \frac{(\alpha \nu_0 a_{\pm,0})^n}{n!} W\xxe
=  \sum_{n=0}^\infty \frac{(\alpha  \nu_0)^n}{n!} \xxa a_{\pm,0}^n W\xxe
	\\
& = \sum_{n=0}^\infty \frac{(\alpha \nu_0)^n}{n!}  \frac{1}{2^n} \sum_{j=0}^n \binom{n}{j} a_{\pm,0}^j W a_{\pm,0}^{n-j} 
	\\
& = \sum_{n=0}^\infty \sum_{j=0}^n  \frac{(\alpha \nu_0 a_{\pm,0}/2)^j}{j!} W \frac{(\alpha \nu_0 a_{\pm,0}/2)^{n-j}}{(n-j)!} 
	\\ \label{ealphaQRmu1W}
& = \Bigg(\sum_{j=0}^\infty \frac{(\alpha \nu_0 a_{\pm,0}/2)^j}{j!} \Bigg) W \Bigg(\sum_{j=0}^\infty \frac{(\alpha \nu_0 a_{\pm,0}/2)^{j}}{j!} \Bigg)
= \ee^{\alpha Q/2} W \ee^{\alpha Q/2}.
\end{align*} 
In particular, $\xxa \ee^{\alpha_\pm Q_\pm}  R_\pm^{\mu_\pm}\xxe \; = \ee^{\alpha_\pm Q_\pm/2} R_\pm^{\mu_\pm}\ee^{\alpha_\pm Q_\pm/2}$ and, since $[R_\pm,Q_\mp]=[R_\pm,S]=[Q_\pm,S]=[Q_+,Q_-]=0$, the result follows. 
\end{proof}

\subsection{Bogoliubov transformation and vertex operators} 
We define a one-parameter family of representations of the algebra $\cA$ which reduces to the defining representation above in a limiting case. We also define and study corresponding vertex operators. 

\label{subsec:BT} 
Let
\begin{equation} 
\label{BT} 
b_{r,n}\coloneqq c_n a_{r,n} - s_n a_{-r ,-n} \quad (n\in\Z_{\neq 0}), \quad 
b_{r ,0}\coloneqq a_{r,0}\quad (r=\pm)   
\end{equation} 
with
\begin{equation} 
\label{cnsn} 
c_n\coloneqq \sqrt{\frac{1}{1-q^{2|n|}}},\quad s_n\coloneqq \sqrt{\frac{q^{2|n|}}{1-q^{2|n|}}} \quad (n\in\Z_{\neq 0}) 
\end{equation} 
and $q$ in \eqref{constants}. 
Using the relation $c_n^2-s_n^2=1$, it is easy to check that the operators $b_{r,n}$ satisfy the same relations as the operators $a_{r,n}$, i.e., 
\begin{equation} 
\label{bR1R2}
[b_{r,n},b_{r',m}]=n\delta_{n,-m}\delta_{r,r'}\id,\quad [b_{r,n},R_{r'}] = \delta_{n,0}\delta_{r,r'}R_{r'} ,\quad b_{r,n}^\dag = b_{r,-n} 
\end{equation} 
for all $r,r'=\pm$ and $n,m\in\Z$. Put differently, for each $q$ in the range $0\leq q<1$,  $\rho(a_{r,n})\coloneqq b_{r,n}$ and $\rho(R_r) \coloneqq R_r$ for $r=\pm$ and $n\in\Z$ defines a $*$-representation $\rho$ of the algebra $\cA$ on $\cF$. 

\begin{remark}\label{rem:notation2}
Our $b_{+,n}$ corresponds to $\pi(\hat\rho(n))$ in \cite[Eq.~(19)]{langmann2004}; recall the identifications $a_{+,n}=\hat\rho_1(n)$ and $a_{-,n}=-\hat\rho_2(n)$, which lead to the opposite signs of the $s_n$-terms in the Bogoliubov transformations \eqref{BT} and  \cite[Eq.~(19)]{langmann2004}.    
\end{remark}

\begin{remark}\label{rem:vacua} 
It is interesting to note the following condensed matter physics interpretation of the construction above: there exists a Luttinger model Hamiltonian with particular interactions such that this Hamiltonian is diagonalized by the Bogoliubov transformation \eqref{BT}--\eqref{cnsn}; see \cite[Section~III.A]{berntson2020}. 
By using well-known results about the Luttinger model \cite[Section~IV]{mattis1965}, one can construct a unitary operator, $\cU$, on $\cF$ with the following properties: (i) $b_{r,n}=\cU^\dag a_{r,n}\cU$ ($r=\pm$, $n\in\Z$) and (ii) the ground state of this Luttinger model Hamiltonian is $\tilde\Omega\coloneqq \cU^\dag\Omega$, where $\tilde\Omega$ satisfies the highest weight conditions 
\begin{equation} 
\label{bR3} 
b_{\pm,n}\tilde\Omega=0\quad (n\in\Z_{\geq 0}).
\end{equation} 
The conditions \eqref{bR3} are analogous to \eqref{aR3} but involve the Bogoliubov-transformed Heisenberg algebra operators $b_{\pm,n}$ instead of $a_{\pm,n}$.
Thus, in our model, we have two different vacua: the vacuum $\Omega$, which can be interpreted as the ground state of a non-interacting fermion model, and the vacuum $\tilde\Omega$, which can be interpreted as the ground state of a Luttinger model. We emphasize that the normal ordering prescriptions we use in this paper are with respect to the non-interacting vacuum $\Omega$. 
\end{remark}

We now introduce a class of operators that will be useful for us. 

\begin{definition}[Vertex operators] 
\label{def:vertexoperators} 
For arbitrary integer vectors $\mu=(\mu_+,\mu_-)\in\Z^2$ and complex-valued sequences $\alpha=(\alpha_{r,n})_{r=\pm,n\in \Z}$, let 
\begin{equation}\label{Phidef}
\Phi_\mu(\alpha)\coloneqq \; \xxa R_+^{\mu_+}R_-^{{\mu_-}}\ee^{\ii J(\alpha)}\xxe 
\end{equation}
with 
\begin{equation}\label{Jdef}
J(\alpha)\coloneqq  \sum_{r=\pm} \Bigg( \alpha_{r,0}Q_r +\sum_{n\in \Z_{\neq 0}}\alpha_{r,n}b_{r,-n}  \Bigg)  = \sum_{r=\pm} \Bigg( \alpha_{r,0}Q_r +\sum_{n\in \Z_{\neq 0}}( \alpha_{r,n}c_n - \alpha_{-r,-n}s_n) a_{r,-n}  \Bigg). 
\end{equation} 
\end{definition} 

We call such operators $\Phi_\mu(\alpha)$ {\em (regularized) vertex operators}. The following is a technical result about vertex operators which we will invoke repeatedly.

\begin{lemma}\label{lem:Phi}
$(a)$ For arbitrary integer vectors $\mu=(\mu_+,\mu_-)\in\Z^2$ and complex-valued sequences $\alpha=(\alpha_{r,n})_{r=\pm,n\in \Z}$, the vertex operator $\Phi_\mu(\alpha)$ in \eqref{Phidef} is a well-defined sesquilinear form on $\mathcal{D}$ such that
\begin{equation}\label{Phinormalordered}
\Phi_\mu(\alpha)=\ee^{\ii \sum_{r=\pm}\alpha_{r,0} Q_r/2} R_+^{\mu_+}R_-^{\mu_-} \ee^{\ii \sum_{r=\pm}\alpha_{r,0} Q_r/2}\ee^{\ii J^+(\alpha)}\ee^{\ii J^-(\alpha)}
\end{equation} 
where
\begin{equation}\label{projectionsdefinition}
\begin{split} 
J^{+}(\alpha)\coloneqq & \sum_{r=\pm}\sum_{n=1}^{\infty} ( \alpha_{r,n}c_n - \alpha_{-r,-n}s_n) a_{r,-n},\\ 
J^{-}(\alpha)\coloneqq & \sum_{r=\pm}\sum_{n=1}^{\infty} ( \alpha_{r,-n}c_n - \alpha_{-r,n}s_n) a_{r,n},
\end{split} 
\end{equation}
i.e., $J^{+}(\alpha)$ and $J^{-}(\alpha)$ are the creation and annihilation parts of $J(\alpha)$, respectively, satisfying 
\begin{equation} 
J^-(\alpha)\Omega=J^+(\alpha)^\dag\Omega=0.
\end{equation}

$(b)$ The vertex operators $\Phi_\mu(\alpha)$ for  $\alpha=(\alpha_{r,n})_{r=\pm,n\in \Z}$ such that
\begin{equation}\label{HScondition}
\sum_{r=\pm}\sum_{n\in \Z}  |n||\alpha_{r,n}|^2   <\infty
\end{equation}
generate a $*$-algebra of well-defined sesquilinear forms on the domain $\mathcal{D}$, with the star relation given by 
\begin{equation}\label{Phiadjoint} 
\Phi_\mu(\alpha)^\dag = (-1)^{\mu_+\mu_-}\Phi_{-\mu}(-\alpha^*) 
\end{equation}
where 
\begin{equation}
\label{eq:star} 
(\alpha^*)_{r,n}\coloneqq \overline{\alpha_{r,-n}}
\end{equation} 
with the bar indicating complex conjugation, and the multiplication rule 
\begin{equation} 
\label{Phimult}
\Phi_\mu(\alpha)\Phi_{\mu'}(\beta)=\chi_{\mu,{\mu'}}(\alpha,\beta)\Phi_{\mu+{\mu'}}(\alpha+\beta) 
= \frac{\chi_{\mu,\mu'}(\alpha,\beta)}{\chi_{\mu',\mu}(\beta,\alpha)}\Phi_{\mu'}(\beta)\Phi_\mu(\alpha)
\end{equation} 
where $(\mu+{\mu'})_\pm \coloneqq \mu_\pm+{\mu'}_\pm$, $(\alpha+\beta)_{r,n} \coloneqq \alpha_{r,n}+\beta_{r,n}$, 
\begin{equation}\label{cocycle}
\chi_{\mu,{\mu'}}(\alpha,\beta) \coloneqq (-1)^{\mu_-{\mu'}_+}\ee^{\ii\sum_{r=\pm} (\alpha_{r,0}{\mu'}_r-\beta_{r,0}\mu_r) \nu_0/2}
\ee^{-[J^-(\alpha),J^+(\beta)]}
\end{equation}
and
\begin{equation}\label{J-J+}
[J^-(\alpha),J^+(\beta)]= \sum_{r=\pm} \sum_{n=1}^{\infty}n\big(c_n^2\alpha_{r,-n}\beta_{r,n}+s_n^2\alpha_{r,n}\beta_{r,-n}-c_ns_n(\alpha_{-r,n}\beta_{r,n}+\alpha_{-r,-n}\beta_{r,-n}) \big). 
\end{equation} 
Moreover, 
\begin{equation} 
\label{Phiexpectation} 
\langle \Omega, \Phi_{\mu}(\alpha)\Omega \rangle = \delta_{\mu_+,0}\delta_{\mu_-,0}. 
\end{equation} 

$(c)$ The product of an arbitrary number, $N$, of vertex operators $\Phi_{\mu_j}(\alpha_j)$ $(j=1,\ldots,N)$ in the $*$-algebra defined in $(b)$ above is related to its normal ordered form as follows, 
\begin{equation}\label{NPhi}
\Phi_{\mu_1}(\alpha_1)\cdots \Phi_{\mu_N}(\alpha_N) = \prod_{1\leq j<k\leq N}\chi_{\mu_j,\mu_k}(\alpha_j,\alpha_k)  \xxa \Phi_{\mu_1}(\alpha_1)\cdots \Phi_{\mu_N}(\alpha_N) \xxe 
\end{equation} 
with $\chi_{\mu,\mu'}(\alpha,\beta)$ given in \eqref{cocycle}. Moreover, the vacuum expectation value of such a product, 
\begin{equation}\label{NPhiexp}
\langle\Omega, \Phi_{\mu_1}(\alpha_1)\cdots \Phi_{\mu_N}(\alpha_N) \Omega \rangle , 
\end{equation} 
is non-zero only if $\sum_{j=1}^N (\mu_j)_r=0$ for $r=\pm$, and if this is the case it is equal to  
\begin{equation} 
 \prod_{1\leq j<k\leq N}\chi_{\mu_j,\mu_k}(\alpha_j,\alpha_k) .
\end{equation} 

$(d)$ If $\alpha$ satisfies the condition in \eqref{HScondition} and is real in the sense that $\alpha=\alpha^*$ (see Remark~\ref{rem:star} below), then $\Phi_\mu(\alpha)$ is proportional to a unitary operator on $\cF$. 
\end{lemma}

For the convenience of the reader, we give a concise self-contained proof of this lemma further below. Before that, we give a few remarks.

\begin{enumerate} 
\item The results summarized in Lemma~\ref{lem:Phi} are well-known, even though they are usually formulated in a different language (some relevant literature can be tracked down from \cite{langmann2004,langmann2015}, for example). 

\item Lemma~\ref{lem:Phi}$(a)$ and $(c)$ ensure that the computations in the rest of this paper are well-defined. In a nutshell, the strategy we use is as follows. 
As will be seen, anyons are vertex operators $\Phi_\mu(\alpha)$ for sequences  $\alpha$ such that $\alpha_{r,n}\propto \ee^{2\ii\kappa x}/\ii n$ ($n\in\Z_{\neq 0}$, $x\in\R$) and, for such sequences, the condition in \eqref{HScondition} is {\em not} satisfied; this corresponds to the fact that anyons are not operators but operator-valued distributions. However, we can make precise sense of anyons  by regularizing the pertinent sequences $\alpha$ as follows,
\begin{equation*} 
\alpha_{r,n}\to \ee^{-2\kappa \eps|n|}\alpha_{r,n} \quad (n\in\Z_{\neq 0}) 
\end{equation*} 
with $\eps>0$ a regularization parameter. The regularized vertex operators are proportional to unitary operators which can be multiplied without problems; at the end of the computations, the limit can be taken where all regularization parameters are removed. This can be regarded as a generalization of the well-known strategy to make mathematical sense of the $2\ell$-periodic Dirac delta function $\delta(x)$ as the limit of the following $C^\infty$-functions as $\eps\to 0^+$, 
\begin{equation} 
\label{deltaeps}
\delta(x;\eps) \coloneqq \frac1{2\ell}\sum_{n\in\Z} \ee^{2\kappa (\ii n x-|n|\eps)}\quad (x\in\R,\eps>0). 
\end{equation} 
\item As explained in more detail below, the multiplication rule \eqref{Phimult}--\eqref{J-J+} is a simple consequence of \eqref{expQR}, \eqref{RRRR1},  
and the Baker-Campbell-Hausdorff formula for operators $A$ and $B$ which have a $\C$-number commutator: 
\begin{equation}\label{BCH}
[A,B]=c I \Rightarrow \ee^A \ee^B= \ee^{c/2}\ee^{A+B}=\ee^c \ee^B \ee^A \qquad (c\in \C). 
\end{equation}
In practical computations with vertex operators, it is often convenient to use  \eqref{expQR}, \eqref{RRRR1} and \eqref{BCH} directly (rather than \eqref{Phimult}--\eqref{J-J+}). 
\item\label{rem:star}  To motivate the involution $*$ defined above, we recall that the pair of sequences $\alpha=(\alpha_{r,n})_{r=\pm,n\in\Z}$ can be naturally identified with a pair of complex-valued functions on the circle via inverse Fourier transformation, 
\begin{equation*}
\check\alpha_r(x) =\sum_{n\in \Z} \alpha_{r,n}\ee^{2\kappa \ii nx}\quad (r=\pm,x\in[-\ell,\ell]). 
\end{equation*}
With that identification, $\alpha^*$ corresponds to the complex conjugated functions $\overline{\check\alpha_r(x)}$. In particular, $\alpha=\alpha^*$ if and only if the functions $\check\alpha_r(x)$ are real-valued.

\item The set of pairs $(\mu,\alpha)$ with $\mu=(\mu_+,\mu_-)\in\Z^2$ and $\C$-valued sequences $\alpha=(\alpha_{r,n})_{r=\pm,n\in\Z}$ satisfying $\alpha=\alpha^*$ and \eqref{HScondition}  is an Abelian group under addition. The function $\chi_{\mu,{\mu'}}(\alpha,\beta)$ defined in \eqref{cocycle} is a non-trivial $2$-cocycle of this group, and $(\mu,\alpha)\mapsto \Phi_\mu(\alpha)$ is a projective representation of this group.
This group is an important example of a loop group (where one usually interprets $\alpha$ as a pair of functions, as explained in Remark~\ref{rem:star} above).
\end{enumerate} 

\begin{proof}[Proof of Lemma~\ref{lem:Phi}] 
$(a)$ This follows from  Lemma~\ref{lem:normalordering} (interpreting the exponentials as Taylor series). 

$(b)$ The definition \eqref{Jdef} of $J(\alpha)$ implies that $J(\alpha)^\dag  =  J(\alpha^*)$, which implies $J^\pm(\alpha)^\dag=J^{\mp}(\alpha^*)$; see \eqref{projectionsdefinition}.  
Hence, using \eqref{Phinormalordered}, \eqref{aR2}, and \eqref{RRRR1},
\begin{multline*} 
\Phi_\mu(\alpha)^\dag = \big(\ee^{\ii \sum_{r=\pm}\alpha_{r,0} Q_r/2} R_+^{\mu_+}R_-^{\mu_-} \ee^{\ii \sum_{r=\pm}\alpha_{r,0} Q_r/2}\ee^{\ii J^+(\alpha)}\ee^{\ii J^-(\alpha)}\big)^\dag \\
= 
\ee^{-\ii \sum_{r=\pm}\overline{\alpha_{r,0}} Q_r/2}R_-^{-\mu_-}R_+^{-\mu_+}\ee^{-\ii \sum_{r=\pm}\overline{\alpha_{r,0}} Q_r/2}\ee^{-\ii J^+(\alpha^*)}\ee^{-\ii J^-(\alpha^*)} 
 = (-1)^{\mu_+\mu_-}\Phi_{-\mu}(-\alpha^*).
\end{multline*} 
The multiplication rule \eqref{Phimult} is implied by   \eqref{RRRR}, \eqref{expQR}, \eqref{Phinormalordered} and the Baker-Campbell-Hausdorff formula \eqref{BCH}; 
indeed, these formulas imply 
\begin{multline}\label{productof2NOoperators}
\Phi_\mu(\alpha)\Phi_{\mu'}(\beta) = \; \xxa R_+^{\mu_+}R_-^{\mu_-} \ee^{\ii J(\alpha)}\xxe \xxa R_+^{{\mu'}_+}R_-^{{\mu'}_-} \ee^{\ii J(\beta)}\xxe 
\\ = \;  (-1)^{\mu_-{\mu'}_+}\ee^{\ii\sum_{r=\pm} (\alpha_{r,0}{\mu'}_r-\beta_{r,0}\mu_r ) \nu_0/2} \ee^{-[J^-(\alpha),J^+(\beta)]}
\xxa R_+^{\mu_+ +{\mu'}_+}R_-^{\mu_++{\mu'}_-} \ee^{\ii J(\alpha+\beta)}\xxe
\end{multline}
which gives \eqref{Phimult}. On the other hand, from \eqref{aR1} and \eqref{projectionsdefinition}, we deduce that
\begin{multline}
[J^-(\alpha),J^+(\beta)] 
= \sum_{r,r'=\pm}\sum_{n,m=1}^\infty  (\alpha_{r,-n}c_n - \alpha_{-r,n}s_n)(\beta_{r',m}c_m - \beta_{-r',-m}s_m)\underbrace{[a_{r,n}, a_{r',-m}]}_{\delta_{r,r'}n\delta_{n,m}}\\ 
= \sum_{r=\pm}\sum_{n=1}^{\infty} n\big(c_n^2\alpha_{r,-n}\beta_{r,n}+s_n^2\alpha_{-r,n}\beta_{-r,-n}-c_ns_n(\alpha_{-r,n}\beta_{r,n}+\alpha_{r,-n}\beta_{-r,-n})\big).
\end{multline}
Changing $r\to -r$ in two of the terms, we obtain \eqref{J-J+}.
Finally, \eqref{Phiexpectation} follows from \eqref{aR4}   and \eqref{Phinormalordered}. 

$(c)$ We note that, by \eqref{Phinormalordered}, 
\begin{equation}\label{NPhi2} 
\xxa \Phi_{\mu_1}(\alpha_1)\cdots \Phi_{\mu_N}(\alpha_N)\xxe \; = \Phi_{\mu_1+\cdots+\mu_N}(\alpha_1+\cdots+\alpha_N), 
\end{equation} 
and therefore \eqref{NPhi} is equivalent to 
\begin{equation}\label{NPhi1} 
\Phi_{\mu_1}(\alpha_1)\cdots \Phi_{\mu_N}(\alpha_N) = \prod_{1\leq j<k\leq N}\chi_{\mu_j,\mu_k}(\alpha_j,\alpha_k) \Phi_{\mu_1+\cdots+\mu_N}(\alpha_1+\cdots+\alpha_N)
\end{equation} 
for all $N=2,3,\ldots$. This can be proved by induction: For $N=2$, \eqref{NPhi1} is the first identity in \eqref{Phimult}, and the induction step from $N$ to $N+1$ is implied by the first identity in \eqref{Phimult} since 
\begin{equation} 
\chi_{\mu_1+\cdots+\mu_N,\mu_{N+1}}(\alpha_1+\cdots+\alpha_N,\alpha_{N+1}) = \prod_{j=1}^N \chi_{\mu_j,\mu_{N+1}}(\alpha_j,\alpha_{N+1}) 
\end{equation} 
by \eqref{cocycle} and the fact that the map $\alpha\mapsto J^-(\alpha)$ is linear. Combining \eqref{NPhi1} with \eqref{Phiexpectation} we obtain the stated result about the vacuum expectation value in \eqref{NPhiexp}.

$(d)$ To see this, note that \eqref{expQR}, \eqref{Phinormalordered}, and \eqref{BCH} imply that
\begin{equation} 
\Phi_\mu(\alpha) = \ee^{\sum_{r=\pm}\ii\nu_0\alpha_{r,0}\mu_r/2}\ee^{[J^-(\alpha),J^+(\alpha)]/2} U_\mu(\alpha)
\end{equation} 
 with 
\begin{equation} 
U_\mu(\alpha) \coloneqq R_+^{\mu_+}R_-^{\mu_-} \ee^{\ii J(\alpha)} 
\end{equation} 
(recall that $\sum_{r=\pm}\alpha_{r,0}Q_r + J^+(\alpha)+J^-(\alpha)=J(\alpha)$). 
Assuming that $\alpha=\alpha^*$, we have 
\begin{equation} 
[J^-(\alpha),J^+(\alpha)]= \sum_{r=\pm} \sum_{n=1}^{\infty}n\big(c_n^2+s_n^2)|\alpha_{r,n}|^2 - 2 \sum_{r=\pm} \sum_{n=1}^{\infty}nc_n s_n \re(\alpha_{r,n} \alpha_{-r,n}),
\end{equation} 
where the sums are finite if and only if the condition in \eqref{HScondition} holds true\footnote{Since $c_n^2=1+s_n^2$  and $s_n$ goes to zero like $q^n$ as $n\to\infty$,  $\eqref{HScondition}$ controls the convergence of these sums.} and, in this case, $U_\mu(\alpha)$ is a unitary operator: $U_\mu(\alpha)U_\mu(\alpha)^\dag =U_\mu(\alpha)^\dag U_\mu(\alpha)=I$. 
\end{proof}

\section{Anyons and the eCS model}
\label{sec:anyonseCS} 
We give a precise meaning to anyons and the second quantization of the eCS model. This generalizes results in \cite{langmann2004}. 

\subsection{Anyons}
\label{subsec:anyons2}
Anyons can be obtained as special cases of vertex operators as follows. 

\begin{definition}[Anyons] 
\label{def:anyons} 
The (regularized) anyons with statistics parameter $\nu\in\nu_0\Z$ and position $x\in[-\ell,\ell]$ are given by 

\begin{equation}\label{anyondefinition}
\phi_{r,\nu}(x;\epsilon)\coloneqq \; \xxa R_r^{\nu/\nu_0} \ee^{-\ii\nu (2r\kappa Q_r x+K_r(x;\epsilon))} \xxe\quad (r=\pm)
\end{equation} 
with
\begin{equation}\label{Kr} 
K_r(x;\epsilon)\coloneqq \sum\limits_{n\in\Z_{\neq 0}}\frac{1}{\ii n} \ee^{2\kappa(\ii  nrx-|n|\epsilon)} b_{r,n} . 
\end{equation}
\end{definition} 

To motivate this definition, we note that the anyons $\phi_\eps(x)$ defined in \cite[Eq.\ (37)]{langmann2004} can be identified with our anyons $\phi_{+,\nu}(x;\eps)$.\footnote{Note that $Q$, $\nu$, and $\mu\nu$ in  \cite[Eq.\ (37)]{langmann2004} correspond to $Q_+/\nu_0$, $\nu_0$ and, $\nu$ here, respectively.} As is clear from our presentation here, there is another independent set of anyons $\phi_{-,\nu}(x;\eps)$ which are on the same footing as the anyons $\phi_{+,\nu}(x;\eps)$ except for a sign change in the argument $x$: the naive choice for these other anyons would be 
\begin{equation} 
\label{def:tildephi} 
\tilde\phi_\eps(x) \coloneqq \phi_{-,\nu}(-x;\eps), 
\end{equation} 
without this sign change. The reason why we make this sign change becomes clear from the following result summarizing the properties of the anyons defined in \eqref{anyondefinition}: if and only if we make this sign change, we obtain anyon correlation functions which are translationally invariant; see \eqref{correlationfunction}. As will be elaborated in Section~\ref{sec:ncILW}, this sign change has far-reaching consequences for the physics interpretation of our results. 

We now summarize various properties of the anyons.

\begin{proposition}\label{prop:anyons} 
$(a)$ The anyons $\phi_{r,\nu}(x;\eps)$ are proportional to unitary operators on $\cF$, and they satisfy  
\begin{equation}\label{phiphip} 
\phi_{r,\nu}(x;\eps)\phi_{r',\nu'}(x';\eps') = (-1)^{\delta_{r,-}\delta_{r',+}\nu\nu'/\nu_0^2} \tet_{r,r'}(x-x';\eps+\eps')^{\nu\nu'} \xxa \phi_{r,\nu}(x;\eps)\phi_{r',\nu'}(x';\eps') \xxe
\end{equation} 
for all $\nu,\nu'\in\nu_0\Z$, $r,r'\in\{\pm\}$, $x,x'\in[-\ell,\ell]$, and $\eps,\eps'>0$, with the special functions 
\begin{equation} 
\label{tetrr} 
\tet_{r,r'}(x;\eps) \coloneqq \begin{cases} \ttet_1(r\kappa x,q;\kappa\eps) \quad (r=r') \\ \ttet_4(\kappa x,q;\kappa\eps) \quad (r=-r') \end{cases} \quad (r,r'=\pm), 
\end{equation} 
\begin{equation} 
\label{ttet1ttet4}
\begin{split} 
\ttet_1(x,q;\eps)\coloneqq &\;  (\ee^{-\ii  x}-\ee^{\ii x-2\eps})\prod_{m=1}^\infty \big(1-q^{2m}\ee^{2\ii x-2\eps}\big)\big(1-q^{2m}\ee^{-2\ii x-2\eps}\big),\\
\ttet_4(x,q;\eps)\coloneqq &\;  \prod_{m=1}^\infty \big(1-q^{2m-1}\ee^{2\ii x-2\eps}\big)\big(1-q^{2m-1}\ee^{-2\ii x-2\eps}\big).
\end{split} 
\end{equation} 

\noindent $(b)$ The anyons obey the exchange relations
\begin{equation}\label{exchangerelations} 
\phi_{\nu,r}(x;\eps)\phi_{\nu',r'}(x';\eps') = 
\begin{cases} 
\ee^{-\ii \pi\nu\nu' \sign(r(x-x');\eps+\eps')}\phi_{\nu',r'}(x';\eps')\phi_{\nu,r}(x;\eps)\quad &(r'=r)\\
(-1)^{\nu\nu'/\nu_0^2}\phi_{\nu',r'}(x';\eps')\phi_{\nu,r}(x;\eps)\quad &(r'=-r)
\end{cases} 
\end{equation} 
with 
\begin{equation}\label{sgndefinition}
\mathrm{sgn}(x;\epsilon)\coloneqq 
-\frac1{\ii\pi}\log\left(\frac{1-\ee^{2\kappa(\ii x-\eps)}}{\ee^{2\kappa\ii x}-\ee^{-2\kappa\eps}}\right) \quad (x\in\R,\eps>0)
\end{equation}
for all $\nu,\nu'\in\nu_0\Z$, $r,r'\in\{\pm\}$, $x,x'\in[-\ell,\ell]$, and $\eps,\eps'>0$, where the branch of the logarithm is fixed so that $\mathrm{sgn}(x;\epsilon)$ is continuous for $x \in \R$ and $\mathrm{sgn}(0;\epsilon) = 0$.

\noindent $(c)$  The anyons have the following adjoints, 
\begin{equation}\label{anyon_adjoints}
\phi_{r,\nu}(x;\eps)^\dag = \phi_{r,-\nu}(x;\eps), 
\end{equation} 
and the following periodicity properties, 
\begin{equation}\label{anyonsperiodicity}
\phi_{r,\nu}(x+2\ell;\eps) = \ee^{-\ii r \pi \nu Q_{r}} \phi_{r,\nu}(x;\eps)  \ee^{-\ii r \pi\nu Q_{r}}= \ee^{-\ii r\pi \nu^{2}}\phi_{r,\nu}(x;\epsilon)\ee^{-2\ii\pi\nu Q_r} , 
\end{equation} 
for all $\nu\in\nu_0\Z$, $x\in[-\ell,\ell]$, and $\eps>0$.  

\noindent $(d)$ A product of an arbitrary number, $N$, of anyons is related to its normal ordered form as follows, 
\begin{multline}\label{normalorderanyons} 
\phi_{\nu_1,r_1}(x_1;\eps_1)\cdots \phi_{\nu_{N},r_{N}}(x_{N};\epsilon_N) \\
= \prod_{1\leq j<k\leq N} (-1)^{\delta_{r_j,-}\delta_{r_k,+}\nu_j\nu_k/\nu_0^2} \tet_{r_j,r_k}(x_j-x_k;\eps_j+\eps_k)^{\nu_j\nu_k}\\ \times 
\xxa \phi_{\nu_1,r_1}(x_1;\eps_1)\cdots \phi_{\nu_{N},r_{N}}(x_{N};\epsilon_N)\xxe.
\end{multline} 
with the special functions in \eqref{tetrr}--\eqref{ttet1ttet4}, for all $\nu_j\in\nu_0\Z$, $r_j\in\{\pm\}$, $x_j\in[-\ell,\ell]$ and $\eps_j>0$ ($j=1,\ldots,N$).  
Thus, the correlation function $\langle\Omega,\phi_{\nu_1,r_1}(x_1;\eps_1)\cdots \phi_{\nu_{N},r_{N}}(x_{N};\eps_{N})\Omega\rangle$
is non-zero only if the following total charges of the right- and left handed particles are both zero, 
\begin{equation} 
Q_\pm\coloneqq \sum_{j=1}^{N}\delta_{\pm ,r_j}\nu_j/\nu_0 ;  
\end{equation} 
more specifically, 
\begin{multline} 
\label{correlationfunction} 
\langle \Omega,\phi_{\nu_1,r_1}(x_1;\eps_1)\cdots \phi_{\nu_{N},r_{N}}(x_{N};\eps_{N})\Omega\rangle  
\\= \delta_{Q_+,0}\delta_{Q_-,0} 
\prod_{1\leq j<k\leq N} (-1)^{\delta_{r_j,-}\delta_{r_k,+} \nu_j\nu_k/\nu_0^2} \tet_{r_j,r_k}(x_j-x_k;\eps_j+\eps_k)^{\nu_j\nu_k}. 
\end{multline} 
\end{proposition} 
Before giving a proof of Propositon~\ref{prop:anyons} further below, we make a few remarks and define special functions needed later on. 

\begin{enumerate} 
\item The function $\sign(x;\eps)$ defined in \eqref{sgndefinition} is a regularized version of the sign function $\sign(x)$ (where $\sign(x)$ is equal to $+1$ and $-1$ for $x>0$ and $x<0$, respectively): $\lim_{\eps\to 0^+}\sign(x;\eps)=\sgn(x)$. One way to see this is to use the Taylor series of $\log(1-z)$ and analytic continuation to derive 
\begin{equation} \label{sgndef2}
\mathrm{sgn}(x;\epsilon) =\frac{x}{\ell} +\frac1\pi \sum\limits_{n\in\Z_{\neq 0}}\frac1{\ii n}\ee^{2\kappa( \ii nx-|n|\epsilon)}, 
\end{equation} 
which implies $\mathrm{sgn}(-x;\epsilon) =-\mathrm{sgn}(x;\epsilon)$ and $\partial_x\mathrm{sgn}(x;\epsilon) = 2\delta(x;\eps)$ with $\delta(x;\eps)$ the regularized Dirac delta defined in \eqref{deltaeps} (recall that $\kappa/\pi=1/2\ell$); the latter two conditions on $\sgn(x;\eps)$ have a unique solution which, in the limit $\eps\to 0^+$, is equal to $\sgn(x)$. Thus, \eqref{exchangerelations} makes precise, and generalizes, the exchange relations in \eqref{phinuphinu}.

\item Anyons are a generalization of fermions; the latter correspond to the special case $\nu=\pm 1$. Another special case important for physics applications 
is that of composite fermions corresponding to $\nu,\nu'=\pm\sqrt{r_0}$ and $\nu_0=1/\sqrt{r_0}$ (for example\footnote{As explained after Lemma~\ref{lem:r0s0}, $\nu_0$ is not unique; for example, if $\sqrt{r_0}$ is an odd integer, one can also choose $\nu_0=1$.}) with $r_0>0$ an odd integer: in these cases, all exchange relations in \eqref{exchangerelations} become anticommutator relations. 

\item\label{rem3}  One important ingredient in Proposition~\ref{prop:anyons} is the special functions $\ttet_1(x,q;\eps)$ and $\ttet_4(x,q;\eps)$ in \eqref{ttet1ttet4}: they appear when normal ordering products of anyons and, at the same time,  they are the building blocks of anyon correlation functions. By writing  
\begin{equation} 
\begin{split}\label{ttet1ttet4mod} 
\ttet_1(x,q;\eps)=& -2\ii\ee^{-\eps}\sin(x+\ii\eps)\prod_{m=1}^\infty \big(1 -2q^{2m}\ee^{-2\eps}\cos(2 x)+ q^{4m}\ee^{-4\eps}\big), \\
\ttet_4(x,q;\eps)=&  \prod_{m=1}^\infty \big(1 -2q^{2m-1}\ee^{-2\eps}\cos(2x)+ q^{4m-2}\ee^{-4\eps}\big)
\end{split} 
\end{equation} 
one can see that $\ttet_1(x,q;\eps)$ and $\ttet_4(x,q;\eps)$, in the limit $\eps\to 0^+$, are proportional to the Jacobi theta functions $\tet_1(x,q)$ and  $\tet_4(x,q)$, respectively. Indeed, comparing \eqref{ttet1ttet4mod} with \eqref{tet1tet4}, we see that
\begin{equation*}
\lim_{\eps\to 0^+} \ttet_1(x,q;\eps) = \frac{\tet_1(x,q)}{\ii q^{1/4} \prod_{m=1}^\infty (1-q^{2m})}, \qquad
\lim_{\eps\to 0^+} \ttet_4(x,q;\eps) = \frac{\tet_4(x,q)}{\prod_{m=1}^\infty (1-q^{2m})}.
\end{equation*} 
Thus, we regard $\ttet_1(x,q;\eps)$ and $\ttet_4(x,q;\eps)$ as regularized versions of the theta functions $\tet_1(x,q)$ and  $\tet_4(x,q)$, respectively (the normalization differences are irrelevant for our purposes). 
\item We note the following relation, 
\begin{equation}\label{ttet4fromttet1}
\ttet_1(\kappa(x+\ii\delta),q;\eps) = q^{-1/2} \ee^{-\kappa \ii x}\ttet_4(\kappa x,q;\eps),
\end{equation} 
which can be verified from the definitions \eqref{ttet1ttet4} using $\ee^{-2\kappa\delta}=q$. This generalizes a well-known relation between the theta functions $\tet_1(x,q)$ and $\tet_4(x,q)$; see \eqref{tet1fromtet4}.  
\end{enumerate} 

For later use, we define functions closely related to the functions $\tet_{r,r'}(x;\eps)$ defined in \eqref{tetrr} which will be important in Sections \ref{subsec:eCS} and \ref{sec:eCSgen}.
 
\begin{definition}\label{def:wprreps} The regularized eCS potentials are 
\begin{equation} 
\wp_{r,r'}(x;\eps)\coloneqq -\partial_x^2\log \tet_{r,r'}(x;\eps)\quad (r,r'=\pm) 
\end{equation} 
for all $x\in[-\ell,\ell]$ and $\eps>0$. 
\end{definition} 
We note that this definition, \eqref{tetrr}, \eqref{ttet1ttet4mod}, and \eqref{ttet4fromttet1} imply\footnote{We use that $\wp_1(-x;\eps)=\wp_1(x;\eps)$.}
\begin{equation} 
\label{wprreps} 
\wp_{r,r'}(x;\eps) = \begin{cases} \wp_1(x;\eps) \quad &(r=r') \\ \wp_1(x+\ii\delta;\eps) \quad &(r=-r') \end{cases} \quad (r,r'=\pm)
\end{equation} 
with 
\begin{equation} 
\label{wp1eps}
\wp_1(x;\eps)\coloneqq -\partial_x^2\log \Big( \sin(\kappa(x+\ii\eps))\prod_{m=1}^\infty \big(1-2q^{2m}\ee^{-2\kappa\eps}\cos(2 \kappa x)+ q^{4m}\ee^{-4\kappa\eps}\big)\Big)
\end{equation} 
a regularized version of the modified Weierstrass function $\wp_1(x)=-\partial_x\zeta_1(x)$: $\lim_{\eps\to 0^+}\wp_1(x;\eps)=\wp_1(x)$; see \eqref{wp1fromtet1}.

\begin{proof}[Proof of Proposition~\ref{prop:anyons}] 
$(a)$ We observe that anyons are vertex operators as in Definition~\ref{def:vertexoperators}: $\phi_{r,\nu}(x;\eps)=\Phi_\mu(\alpha)$ with $\mu =(\mu_+,\mu_-)\in\Z^2$ and $\alpha = (\alpha_{r',n})_{r'=\pm,n\in \Z}$  given by
\begin{equation}\label{anyonparameters}
\mu_{r'}=\delta_{r,r'}\frac{\nu}{\nu_0}, \quad \alpha_{r',0} =-2r\delta_{r,r'}\nu \kappa x,\quad  \alpha_{r',n} =\frac{\delta_{r,r'}}{\ii n}\nu\ee^{-2\kappa(\ii r' nx+|n|\epsilon)} \quad (n\in \Z\setminus\{0\}); 
\end{equation}
since $\alpha=\alpha^*$ (i.e., $ \alpha_{r',n}= \overline{\alpha_{r',-n}}$) and $\sum_{r'=\pm}\sum_{n=1}^\infty n|\alpha_{r',n}|^2<\infty$, Lemma~\ref{lem:Phi}$(d)$ implies that $\phi_{r,\nu}(x;\eps)$ is proportional to a unitary operator. The proof of \eqref{phiphip} is lengthy and is therefore postponed to the end of the proof. 

\noindent $(b)$ 
Let $\mu =(\mu_+,\mu_-)$ and $\alpha = (\alpha_{r',n})_{r'=\pm,n\in \Z}$ be as in \eqref{anyonparameters}, so that $\phi_{r,\nu}(x;\eps)=\Phi_\mu(\alpha)$. 
Similarly, let $\mu' =(\mu'_+,\mu'_-)$ and $\alpha' = (\alpha'_{r'',n})_{r''=\pm,n\in \Z}$  be such that $\phi_{r',\nu'}(x';\eps')=\Phi_{\mu'}(\alpha')$. 
Then the first identity in \eqref{Phimult} and \eqref{phiphip} show that
\begin{equation}\label{lefttoprove} 
\chi_{\mu,{\mu'}}(\alpha,\alpha') = (-1)^{\delta_{r,-}\delta_{r',+}\nu\nu'/\nu_0^2}\tet_{r,r'}(x-x';\eps+\eps')^{\nu\nu'}.
\end{equation} 
Together with the second identity in \eqref{Phimult} and the relation $(-1)^{\delta_{r,-}\delta_{r',+}-\delta_{r',-}\delta_{r,+}}=(-1)^{\delta_{r,-r'}}$, this implies
\begin{equation}\label{phiphip2}
\phi_{r,\nu}(x;\eps)\phi_{r',\nu'}(x';\eps') =  (-1)^{\delta_{r,-r'}\nu\nu'/\nu_0^2}\left(\frac{\tet_{r,r'}(x-x';\eps+\eps')}{\tet_{r',r}(x'-x;\eps+\eps')}\right)^{\nu\nu'}\phi_{r',\nu'}(x';\eps')\phi_{r,\nu}(x;\eps).
\end{equation} 
We use \eqref{tetrr}--\eqref{ttet1ttet4} to compute, for $r'=r$, 
\begin{multline*} 
\frac{\tet_{r,r}(x-x';\eps+\eps')}{\tet_{r,r}(x'-x;\eps+\eps')} =\frac{\ttet_1(r\kappa(x-x'),q;\kappa(\eps+\eps')) }{\ttet_1(r\kappa(x'-x),q;\kappa(\eps+\eps'))}
= \frac{ \ee^{-\ii\kappa r(x-x')}-\ee^{\ii\kappa r(x-x')-2\kappa(\eps+\eps')}}{ \ee^{\ii\kappa r(x-x')}-\ee^{-\ii \kappa r(x-x')-2\kappa(\eps+\eps')}}\\
= \frac{1-\ee^{2\kappa( \ii r(x-x')-(\eps+\eps'))}}{\ee^{2\kappa \ii r(x-x')}-\ee^{-2\kappa(\eps+\eps')}}=\ee^{-\ii\pi\sgn(r(x-x');\eps+\eps')}
\end{multline*} 
with $\sgn(x;\eps)$ in \eqref{sgndefinition}, and for $r'=-r$, 
\begin{equation*} 
\frac{\tet_{r,-r}(x-x';\eps+\eps')}{\tet_{r,-r}(x'-x;\eps+\eps')} = \frac{\ttet_4(\kappa(x-x'),q;\kappa(\eps+\eps')) }{\ttet_4(\kappa(x'-x),q;\kappa(\eps+\eps'))} =1 ; 
\end{equation*} 
inserting this into \eqref{phiphip2} we get \eqref{exchangerelations}--\eqref{sgndefinition}. 

\noindent $(c)$ 
We specialize \eqref{Phinormalordered}--\eqref{projectionsdefinition} to write the anyons as 
\begin{equation}\label{phiexpanded}
\phi_{r,\nu}(x;\epsilon) = \ee^{-\ii \nu \kappa Q_r rx} R_r^{\nu/\nu_0} \ee^{-\ii \nu \kappa Q_r rx} \ee^{-\ii \nu K_r^+(x;\epsilon)} \ee^{-\ii \nu K_r^-(x;\epsilon)}\quad (r=\pm) 
\end{equation}
where (recall \eqref{BT}--\eqref{cnsn}) 
\begin{equation}\label{Krpm}
\begin{split}  
K_r^+(x; \epsilon)  \coloneqq & -\sum_{n=1}^\infty \frac{1}{\ii n} \Big( \ee^{2\kappa(-\ii  n r x-n\epsilon)} c_n a_{r, -n} + \ee^{2\kappa(\ii  n r x-n\epsilon)} s_n a_{-r ,-n}\Big), \\ 
K_r^-(x; \epsilon)  \coloneqq & \sum_{n=1}^\infty \frac{1}{\ii n} \Big( \ee^{2\kappa(\ii  n r x-n\epsilon)} c_n a_{r, n} + \ee^{2\kappa(-\ii  n r x- n\epsilon)} s_n a_{-r ,n}\Big), 
 \end{split} 
\end{equation}
are the creation and annihilation parts of $K_r(x;\eps)$ in \eqref{anyondefinition} for $+$ and $-$, respectively.
Note that $K^\pm_r$ obey $K_r^\pm(x;\epsilon)^\dagger = K_r^\mp(x;\epsilon)$ as a consequence of \eqref{aR2}. 
With that and \eqref{phiexpanded}, \eqref{anyon_adjoints} is verified by the following computation using  \eqref{aR2}, 
\begin{equation}
\phi_{r,\nu}(x;\epsilon)^{\dag}
= \ee^{\ii r\nu \kappa Q_r x} R_r^{-\nu/\nu_0} \ee^{\ii r\nu \kappa Q_r x} 
\ee^{\ii \nu K_r^+(x;\epsilon)}
\ee^{\ii \nu K_r^-(x;\epsilon)} 
=\phi_{r,-\nu}(x;\epsilon). 
\end{equation}
From the definition of the anyons \eqref{anyondefinition}-- \eqref{Kr} and the $2\ell$-periodicity of $K_r(x;\epsilon)$, we obtain the first equality in \eqref{anyonsperiodicity}. The second equality in \eqref{anyonsperiodicity} then follows from \eqref{expQR}. 

\noindent $(d)$ This is a special case of Lemma~\ref{lem:Phi}$(c)$, using \eqref{lefttoprove} (see the proof of $(a)$ above for details). 

We finally prove \eqref{phiphip}. We first use \eqref{phiexpanded} together with  \eqref{expQR} and the Baker-Campbell-Hausdorff formula \eqref{BCH} to compute 
\begin{multline}\label{phiphi2}
\phi_{r,\nu}(x;\epsilon)\phi_{r',\nu'}(x';\epsilon') = \ee^{-\ii r\nu \kappa Q_r x} R_r^{\nu/\nu_0} \ee^{-\ii r\nu \kappa Q_r x} \ee^{-\ii \nu K_r^+(x;\epsilon)} \ee^{-\ii \nu K_r^-(x;\epsilon)}\\ \times  
\ee^{-\ii r'\nu' \kappa Q_{r'} x'} R_{r'}^{\nu'/\nu_0} \ee^{-\ii r'\nu' \kappa Q_{r'} x'} \ee^{-\ii \nu K_{r'}^+(x';\epsilon')} \ee^{-\ii \nu' K_{r'}^-(x';\epsilon')}
\\ =  \ee^{-\delta_{r,r'}\ii \nu\nu'\kappa(rx-r'x')}\ee^{-\ii r\nu \kappa Q_r x-\ii r'\nu' \kappa Q_{r'} x'} R_r^{\nu/\nu_0}R_{r'}^{\nu'/\nu_0} \ee^{-\ii r\nu \kappa Q_r x-\ii r'\nu' \kappa Q_{r'} x'} \\ \times 
 \ee^{-\nu\nu'[K_r^-(x;\epsilon),K_{r'}^+(x';\epsilon')]}\ee^{-\ii \nu K_r^+(x;\epsilon)-\ii \nu K_{r'}^+(x';\epsilon')} \ee^{-\ii \nu K_r^-(x;\epsilon)-\ii \nu' K_{r'}^-(x';\epsilon')}\\
 = (-1)^{\delta_{r,-}\delta_{r',+}\nu\nu'/\nu_0^2}  \ee^{-\nu\nu'(\delta_{r,r'}\ii \kappa r(x-x')+[K_r^-(x;\epsilon),K_{r'}^+(x';\epsilon')])}\xxa \phi_{r,\nu}(x;\epsilon)\phi_{r',\nu'}(x';\epsilon') \xxe , 
\end{multline} 
where we used $R_r^{\nu/\nu_0}R_{r'}^{\nu'/\nu_0} =(-1)^{\delta_{r,-}\delta_{r',+}\nu\nu'/\nu_0^2}\xxa R_r^{\nu/\nu_0}R_{r'}^{\nu'/\nu_0}\xxa$ in the last step (this is non-trivial only for $(r,r')=(-,+)$, in which case it follows from \eqref{RRRR1}).  Next, using definitions and straightforward computations, we obtain the following commutators
\begin{equation}\label{KpmtoCCt}
[K_{r}^{-}(x;\epsilon),K_{r'}^{+}(x';\epsilon')]=\begin{cases}
C(r(x-x');\epsilon+\epsilon') & (r'=r)  \\
\tilde{C}(x-x';\epsilon+\epsilon') & (r'=-r)	
\end{cases}	
\end{equation}
where
\begin{equation}\label{CCt}
\begin{split}
C(x;\epsilon)\coloneqq &\; \sum_{n=1}^{\infty} \frac{1}{n}\big( \ee^{2\kappa(\ii n x-n\epsilon)}c_n^2+\ee^{2\kappa(-\ii n x-n\epsilon)}s_n^2	\big), \\
\tilde{C}(x;\epsilon)\coloneqq &\; \sum_{n=1}^{\infty} \frac{c_ns_n}{n}\big( \ee^{2\kappa(\ii n x-n\epsilon)}+\ee^{2\kappa(-\ii n x-n\epsilon)}	\big).
\end{split}
\end{equation}
Indeed, for $r'=r$, using \eqref{aR1} and \eqref{Krpm}, 
\begin{multline}
[K_{\pm}^{-}(x;\epsilon),K_{\pm}^{+}(x';\epsilon')] = -\sum_{n,m=1}^\infty \Bigl(  \frac{1}{\ii n}\ee^{2\kappa(\pm\ii  n x-n\epsilon)} c_n  \frac{1}{\ii m}\ee^{2\kappa(\mp\ii  m x'-m\epsilon')} c_m
\underbrace{[a_{\pm, n},a_{\pm, -m}]}_{n\delta_{n,m}} \\ +  \frac{1}{\ii n}\ee^{2\kappa(\mp\ii  n x-n\epsilon)}s_n \frac{1}{\ii m}\ee^{2\kappa(\pm\ii  m x'- m\epsilon')} s_m  
\underbrace{[a_{\mp ,n} a_{\mp ,-m}]}_{n\delta_{n,m}} \Big) \\
= \sum_{n=1}^\infty \frac1n\Big(\ee^{2\kappa(\pm\ii  n (x-x')-n(\epsilon+\epsilon'))}c_n^2 + \ee^{2\kappa(\mp\ii  n (x-x')-n(\epsilon+\epsilon'))}s_n^2\Bigr)  =C(\pm(x-x'),\epsilon+\epsilon'), 
\end{multline} 
and similarly for $r'=-r$, 
\begin{multline} 
[K_{\pm}^{-}(x;\epsilon),K_{\mp}^{+}(x';\epsilon')] = 
-\sum_{n,m=1}^\infty \Big(  \frac{1}{\ii n}\ee^{2\kappa(\pm\ii  n x-n\epsilon)} c_n   \frac{1}{\ii m} \ee^{2\kappa(\mp\ii  m x'- m\epsilon')} s_m \underbrace{[a_{\pm, n},a_{\pm,-m}]}_{n\delta_{n,m}}\\
+  \frac{1}{\ii n}  \ee^{2\kappa(\mp\ii  n x-n\epsilon)} s_n \frac{1}{\ii m} \ee^{2\kappa(\pm\ii  m x'-m\epsilon')} c_m \underbrace{[a_{\mp,n}, a_{\mp, -m}]}_{n\delta_{n,m}}\Big) \\
= \sum_{n=1}^\infty \frac{c_ns_n}{n}\Big(\ee^{2\kappa(\pm\ii  n (x-x')-n(\epsilon+\epsilon'))}  +  \ee^{2\kappa(\mp\ii  n (x-x')-n(\epsilon+\epsilon'))} \Big) = \tilde{C}(x-x',\epsilon+\epsilon') .
\end{multline} 

To proceed, we insert \eqref{cnsn} and use the geometric series and the Taylor series of the function $\log(1-x)$ to compute 
\begin{multline*} 
C(x;\eps) = \sum_{n=1}^{\infty} \frac1{n(1-q^{2n})}\ee^{-2\kappa n\epsilon}\big(\ee^{2\kappa \ii nx}+q^{2n}\ee^{-2\kappa\ii n x}\big)\\
= \sum_{n=1}^{\infty}\sum_{m=0}^\infty  \frac1{n}\ee^{-2\kappa n\epsilon}\big(q^{2nm}\ee^{2\kappa \ii nx}+q^{2n(m+1)}\ee^{-2\kappa\ii n x}\big)\\
= -\log(1-\ee^{2\kappa (\ii x-\eps)}) - \sum_{m=1}^\infty\Big( \log\big(1-q^{2m}\ee^{2\kappa (\ii x-\eps)}) + \log(1-q^{2m}\ee^{-2\kappa (\ii x+\eps)}\big) \Big),  
\end{multline*} 
which implies 
\begin{multline} \label{expC}
\ee^{-\ii \kappa x-C(x;\epsilon)} =( \ee^{-\ii \kappa x}-\ee^{\kappa (\ii x-2\eps)})\prod_{m=1}^\infty \big(1-q^{2m}\ee^{2\kappa (\ii x-\eps)}\big)\big(1-q^{2m}\ee^{-2\kappa (\ii x+\eps)}\big)
= \ttet_1(\kappa x,q; \kappa \epsilon).
\end{multline} 
Similarly, 
\begin{multline*} 
\tilde{C}(x;\epsilon) =  \sum_{n=1}^{\infty} \frac{q^n}{n(1-q^{2n})} \ee^{-2\kappa n\epsilon} \big(\ee^{2\ii\kappa x} + \ee^{-2\ii\kappa x} \big) \\
=  \sum_{n=1}^{\infty}\sum_{m=0}^\infty \frac{1}{n} \ee^{-2\kappa n\epsilon} \big(q^{n(2m+1)}\ee^{2\ii\kappa x} + q^{n(2m+1)}\ee^{-2\ii\kappa x} \big) \\
= - \sum_{m=0}^\infty \Big(\log\big(1-q^{2m+1}\ee^{2\kappa(\ii x-\eps)}\big) + \log\big(1-q^{2m+1}\ee^{2\kappa(-\ii x-\eps)}\big) \Big),
\end{multline*} 
implying 
\begin{equation}\label{expCt}
\ee^{-\tilde{C}(x;\eps)} = \prod_{m=1}^\infty \big(1-q^{2m-1}\ee^{2\kappa(\ii x-\eps)}\big)\big(1-q^{2m-1}\ee^{-2\kappa(\ii x+\eps)}\big)
= \ttet_4(\kappa x,q; \kappa \epsilon).
\end{equation} 
Combining \eqref{KpmtoCCt}, \eqref{expC}, and \eqref{expCt}, we find
\begin{align}\nonumber
 \ee^{-\nu\nu'(\delta_{r,r'}\ii \kappa r(x-x')+[K_r^-(x;\epsilon),K_{r'}^+(x';\epsilon')])} 
 &= 
 \begin{cases} \big( \ee^{-\kappa(x-x')-C(r(x-x');\epsilon + \epsilon') }\big)^{\nu\nu'} & (r'=r),  \\
\big( \ee^{-\tilde{C}(x-x';\epsilon + \epsilon')}\big)^{\nu\nu'} & (r'=-r), \end{cases} 
	\\\label{phiphi3}
& = \tet_{r,r'}(x-x';\eps+\eps')^{\nu\nu'},
\end{align} 
and substituting this into \eqref{phiphi2}, we obtain \eqref{phiphip}--\eqref{ttet1ttet4}. 
\end{proof}

\subsection{Second quantizations of eCS model}\label{subsec:eCS} 
We start by defining the operator $\cH_{3,\nu}$ which plays a central role in this paper. 
 
We recall the definition of the charge operators: $Q_r=\nu_0 a_{r,0}$ ($r=\pm$), and that we have two sets of Heisenberg operators $a_{r,n}$ and $b_{r,n}$: the latter are obtained from the former by the Bogoliubov transformation in \eqref{BT}. 

\begin{definition}\label{def:cH} 
For real parameters $\nu$ and $\nu_0$ such that $\nu/\nu_0$ is a finite non-zero integer, let  
\begin{equation}\label{cH2} 
\cH_{2} \coloneqq \sum_{r=\pm} W_{2,r} 
\end{equation} 
and 
\begin{equation}\label{cH3} 
\cH_{3,\nu} \coloneqq \frac12\Bigg(  \nu \sum_{r=\pm} W_{3,r}+(1-\nu^2)\cC\Bigg), 
\end{equation} 
with 
\begin{equation} 
\label{W2pm} 
W_{2,r}\coloneqq (2\kappa)\Bigg(  \frac12 \sum_{n\in\Z_{\neq 0}} \xxa b_{-n,r}b_{n,r}\xxe +\frac12 Q_r^2\Biggr) \quad (r=\pm) , 
\end{equation} 
\begin{equation} 
\label{W3pm} 
W_{3,r} \coloneqq (2\kappa)^2 \Bigg(\frac13  
\sum_{\substack{n,m\in\Z_{\neq 0}\\n+m\neq 0}} 
\xxa b_{r,-n}b_{r,-m}b_{r,n+m}\xxe  + Q_r \sum_{n\in\Z_{\neq 0}} \xxa b_{-n,r}b_{n,r}\xxe + \frac13 Q_r^3 \Bigg) \quad (r=\pm) , 
\end{equation} 
\begin{equation}\label{cC} 
\cC\coloneqq (2\kappa)^2 \sum_{n=1}^{\infty}\sum_{r=\pm} na_{r,-n}a_{r,n} = (2\kappa)^2 \frac12 \sum_{n\in\Z}\sum_{r=\pm} |n| \xxa a_{r,-n}a_{r,n}\xxe .
\end{equation} 
\end{definition} 

According to Lemma~\ref{lem:normalordering}, $\cH_{2}$ and $\cH_{3,\nu}$ are well-defined at least in the sense of sesquilinear forms on $\cD$; one can show that they actually define self-adjoint operators (this can be proved using results in \cite{grosse1992}, for example).

\begin{remark}\label{rem:L2004a} The regularized eCS Hamiltonians $H_{N;g}(\vx;\eps)$ in \eqref{eCSreg} below and $H_N^\eps$ in \cite[Eq.\ (61)]{langmann2004} differ by a factor $1/2$:  $H_{N;g}(\vx;\eps)=(1/2)H_N^\eps$. This is the reason for the factor $1/2$ in \eqref{cH3}. 
\end{remark} 

\begin{remark}\label{rem:L2004b}
As mentioned, the operator $\cH_{3,\nu}$ defined in \eqref{cH3} is a major player in the present paper, and it extends the operator $\cH$ defined in \cite[Eqs.\ (57)--(59)]{langmann2004}. 
It is not obvious but true that the latter can be written as $\cH = \nu W_{3,+} + (1-\nu^2)\cC -\nu^4 c_\epsilon a_{+,0}^4$, using the notation in the present paper. 
To verify this one has to take into account the following facts: (i) The parameter $2\kappa=\pi/\ell$ was set to $1$ in \cite{langmann2004} and, for this reason, we have to insert suitable powers of the factor $2\kappa$ in our definitions of $W_{3,\pm}$, $W_{2,\pm}$ and $\cC$ here; to compare with \cite{langmann2004}, set $2\kappa=1$. (ii) Note that $Q$ in \cite{langmann2004} corresponds to our $a_{+,0}$ and, for this reason, $W^3$ and $W^2$ in \cite[Eq.\ (58)]{langmann2004} are obtained from our $W_{3,+}$ and $W_{2,+}$ by replacing $Q_+$ with $a_{+,0}$; recall that $Q_+=\nu_0a_{+,0}$. (iii) In \cite{langmann2004}, $\nu=\nu_0$, while in the present paper the parameters $\nu_0$ and $\nu$ are allowed to be different (this detail is important for the generalizations in Section~\ref{sec:eCSgen}). (iv) The factor in front of the $Q^3$-term in \cite[Eq.\ (57)]{langmann2004} is equal to $\nu[(\nu^3-1)/3-(\nu-1)]$. 
\end{remark} 

For non-zero real $\nu$ such that $\nu/\nu_0$ is an integer, $N\in\Z_{\geq 1}$, $\vx=(x_1,\ldots,x_N)\in[-\ell,\ell]^N$, and $\eps>0$, 
we introduce the following shorthand notation for a product of $N$ of anyons, 
\begin{equation} 
\phi^N_{r,\nu}(\vx;\eps)\coloneqq  \phi_{r,\nu}(x_1;\eps)\cdots \phi_{r,\nu}(x_N;\eps)\quad (r=\pm) , 
\end{equation} 
and we define a corresponding regularized version of the eCS Hamiltonian \eqref{eCS}, 
\begin{equation} 
\label{eCSreg} 
H_{N;g}(\vx;\eps)\coloneqq -\sum_{j=1}^N \frac12 \frac{\partial^2}{\partial x_j^2} + \sum_{1\leq j<k\leq N} g(g-1) \wp_1(x_j-x_k;2\eps) \quad (g>0),
\end{equation} 
where $\wp_1(x;\eps)$ is the regularized modified Weierstrass $\wp$-function defined in \eqref{wp1eps}.

The following is a generalization of \cite[Proposition~2]{langmann2004}; this generalization is easy from a mathematical point of view, but conceptually it is a big step: it 
suggests a new physics interpretation of the model that led us to our new results. 

\begin{proposition}\label{prop:eCS}
The operator $\cH_{3,\nu}$ defined in \eqref{cH3} provides a twofold second quantization of the regularized eCS Hamilton \eqref{eCSreg} in the sense that the following relations hold true in both cases $r=\pm$, 
\begin{equation} 
\label{cH3PhiN}
\big[\cH_{3,\nu}, \phi^N_{r,\nu}(\vx;\eps)\big]\Omega = \left(H_{N;\nu^2}(\vx;\eps) + \half N\nu^4 c_\epsilon \right)\phi^N_{r,\nu}(\vx;\eps)\Omega + \Psi_{r,\nu}(\vx;\eps)\Omega, 
\end{equation} 
where the constant  
\begin{equation}\label{ceps}
c_\epsilon \coloneqq \frac1{3}\kappa^2-8\kappa^2\sum_{n=1}^{\infty}{n}s_n^2 \ee^{-4\kappa n\epsilon}-8\kappa^2\sum_{n,m=1}^{\infty} s_n^2s_m^2\ee^{-2\kappa(n+m)\epsilon}(\ee^{-2\kappa |n-m|\epsilon}-\ee^{-2\kappa(n+m)\epsilon})
\end{equation}
reduces to $c_0$ in \eqref{constants} as $\eps\to 0^+$,  and the following corrections vanish in the limit $\eps\to 0^+$,\footnote{We note that \eqref{cRrnu} for $\kappa=1/2$ corrects typos in \cite[Eq.~(64)]{langmann2004}.}
\begin{multline} 
 \Psi_{r,\nu}(\vx;\eps) \coloneqq \sum_{j=1}^N \phi_{r,\nu}(x_1;\eps)\cdots  \phi_{r,\nu}(x_{j-1};\eps)\xxa \cR_{r,\nu}(x_j;\eps) \phi_{r,\nu}(x_{j};\eps)\xxe
 \\ \times    \phi_{r,\nu}(x_{j+1};\eps) \cdots \phi_{r,\nu}(x_N;\eps), 
\end{multline} 
\begin{multline}\label{cRrnu} 
  \mathcal{R}_{r,\nu}(x;\epsilon)
  = \;  (2\kappa)^2\sum_{n,m=1}^{\infty} \bigg( \nu  s_n^2\left(\ee^{2\kappa \ii mr x} b_{r,m}-\ee^{-2\kappa \ii m r x} b_{r,-m}\right) \ee^{-2\kappa n\epsilon}
\\ -\ee^{-2\kappa \ii(n-m) r x}\xxa b_{r,-n} b_{r,m}\xxe \bigg)\left(\ee^{-2\kappa (n+m)\epsilon} - \ee^{-2\kappa |n-m|\epsilon} \right).
\end{multline}

Moreover, for $r=\pm$, $\vx=(x_1,\ldots,x_N) \in [-\ell,\ell]^N$ and $\vy=(y_1,\ldots,y_N) \in [-\ell,\ell]^N$, it holds that
\begin{equation} 
\label{OmOm}
\langle\Omega,  [\cH_{3,\nu},\phi_{r,\nu}^N(\vx;\eps)^\dag \phi_{r,\nu}^N(\vy;\eps)]\Omega\rangle=0.
\end{equation} 
\end{proposition}   

\begin{proof}[Proof based on Ref.~\cite{langmann2004}] We explain how Proposition~\ref{prop:eCS} is obtained from results in \cite{langmann2004}. 
For readers interested in a self-contained proof, we mention that Proposition~\ref{prop:eCS} is a special case of Theorem~\ref{thm:eCSgen} proved in Section~\ref{sec:eCSgen}.  

Using the notation in the present paper, \cite[Eq.\ (60)] {langmann2004} can be written as
\begin{equation} 
\label{L2004}
\big[\nu W_{+,3}+(1-\nu^2)\cC, \phi^N_{+,\nu}(\vx;\eps)\big]\Omega = \bigl( 2H_{N;\nu^2}(\vx;\eps) +N\nu^4c_\epsilon \Bigr)\phi^N_{+,\nu}(\vx;\eps)\Omega + \Psi_{+,\nu}(\vx;\eps)\Omega 
\end{equation} 
(to see this, use that $\nu_0$,  $\nu W_{3,+}+(1-\nu^2)\cC$, $a_{+,0}$, and $2H_{N; \nu^2}(\vx;\eps)$ here correspond to  $\nu$, $\cH+\nu^4c_\epsilon Q$, $Q$, and $H^{2\eps}_N$ in \cite{langmann2004}, respectively --- this is explained in Remarks~\ref{rem:L2004a} and \ref{rem:L2004b};  since $[Q,\phi_\eps^N(\vx)]=N\phi_\eps^N(\vx)$,  adding $\nu^4c_\epsilon Q$ to $\cH$ on the left-hand side in \cite[Eq.\ (60)]{langmann2004} amounts to adding  the term $N\nu^4c_\epsilon\Phi_\eps(\vx)\Omega$ to the right-hand). 
It follows from our definitions that $W_{-,3}$ in \eqref{W3pm} commutes with $\phi^N_{+,\nu}(\vx;\eps)$ and, for this reason, we can replace $\nu W_{+,3}+(1-\nu^2)\cC$ on the left-hand side in \eqref{L2004} by $2\cH_{3,\nu} = \nu W_{+,3}+\nu W_{-,3}+(1-\nu^2)\cC$ without changing the right-hand side. Dividing the resulting identity by $2$ gives \eqref{cH3PhiN} for $r=+$. 

As discussed in Section~\ref{subsec:anyons2}, the anyons $\phi_{-,\nu}(-x;\eps)$ are on the same footing as the anyons $\phi_{+,\nu}(x;\eps)$. Thus, by symmetry, we can get from the case $r=+$ in \eqref{cH3PhiN}, a corresponding result for the case $r=-$ by replacing all $\phi_{+,\nu}(x_j;\eps)$ and $a_{+,n}$ by $\phi_{-,\nu}(-x_j;\eps)$ and $a_{-,n}$, respectively. By changing $\vx\to -\vx$ in the resulting identity using $H_{N,g}(-\vx;\eps)=H_{N,g}(\vx;\eps)$, we obtain \eqref{cH3PhiN} for $r=-$. 

Since $[a_{+,0},\phi_{+,\nu}^N(\vx;\eps)^\dag \phi_{+,\nu}^N(\vy;\eps)]=0$, \cite[Eq.\ (66)]{langmann2004} implies \eqref{OmOm} for $r=+$; by a symmetry argument as above, one obtains the corresponding result for $r=-$.  
\end{proof}

\section{Derivation of the ncILW equation}
\label{sec:ncILW} 
We show that the operator $\cH_{3,\nu}$ in \eqref{cH3} defines a quantum version of the ncILW equation \eqref{ncILW}. 

We start by comparing the formula \eqref{phirnu}, which is a formal expression suggested by physics considerations \cite{berntson2020}, with our mathematically precise definition of anyons \eqref{anyondefinition}.  This motivates the following definition of chiral bosons $\rho_\pm$. 
 
\begin{definition}\label{def:chiralbosons}  
The (regularized) chiral bosons are given by  
\begin{equation}\label{rhor} 
\begin{split}  
\rho_r (x;\eps) \coloneqq  r\partial_x\big(2\kappa Q_r rx +K_r(x;\eps) \bigr) 
=  2\kappa Q_r +  \sum\limits_{n\in\Z_{\neq 0}} 2\kappa \ee^{2\kappa(\ii r nx-|n|\epsilon)} b_{r,n}  \quad (r=\pm) 
\end{split} 
\end{equation} 
for $x\in[-\ell,\ell]$ and $\eps>0$. 
\end{definition} 


The following result summarizes the properties of the chiral bosons, confirming that Definition~\ref{def:chiralbosons} is satisfactory from a physics point of view. 

\begin{lemma}\label{lem:bosons} 
The chiral bosons are Hermitian, $\rho_r(x;\eps)= \rho_r(x;\eps)^\dag$, they have $2\ell$-periodic extensions to the real line, $\rho_r(x+2\ell;\eps)=\rho_r(x;\eps)$ for $x\in\R$, and they satisfy the following canonical commutator relations, 
\begin{equation}\label{eq:CCR}
[\rho_r(x;\eps),\rho_{r'}(x';\eps')]=-2\pi\ii r\delta_{r,r'}\partial_x\delta(x-x';\eps+\eps')\quad (r,r'=\pm) 
\end{equation} 
with the regularized Dirac delta \eqref{deltaeps}. 
\end{lemma} 

\begin{proof} Hermiticity follows from $(b_{r,n})^\dag=b_{r,-n}$; the $2\ell$-periodicity is obvious. To prove \eqref{eq:CCR}, we use \eqref{bR1R2} and \eqref{deltaeps} to compute 
\begin{multline} 
[\rho_r(x;\eps),\rho_{r'}(x';\eps')] = \sum_{n,m\in\Z_{\neq 0}}(2\kappa)^2  \ee^{2\kappa(\ii r nx+\ii r'mx'-|n|\eps-|m|\eps')} \underbrace{[b_{r,n},b_{r',m}]}_{n\delta_{n,-m}\delta_{r,r'}} \\
= (2\kappa)^2\delta_{r,r'}   \sum_{n \in\Z_{\neq 0}} n\ee^{2\kappa[\ii r n(x-x')-|n|(\eps+\eps')]} 
= (2\kappa)^2\delta_{r,r'}   \sum_{n\in\Z_{\neq 0}} r n\ee^{2\kappa[\ii n(x-x')-|n|(\eps+\eps')]} \\
= (2\kappa)^2 \delta_{r,r'} \frac{1}{2\kappa r\ii}(2\ell)\partial_x \delta(x-x';\eps+\eps'), 
\end{multline} 
which is \eqref{eq:CCR} since $(2\kappa)^2(2\ell)/2\kappa r\ii = -2\pi\ii r$. \end{proof} 

From a physics point of view, one would want to know if the operators in Definition~\ref{def:cH} are local, i.e., if they can be written as integrals of normal ordered powers of the chiral bosons. 
The following result shows that this is so in all cases except for the operator $\cC$ which, instead, can be expressed in a local way using the integral operators $T$ and $\tilde{T}$ in \eqref{TT}. 

\begin{lemma}\label{lem:WjrcC} 
The operators in $W_{2,\pm}$, $W_{3,\pm}$, and $\cC$ in Definition~\ref{def:cH} can be expressed in terms of chiral bosons as follows,  
\begin{equation} 
\label{W2pm2}
W_{2,r} = \lim_{\eps\to 0^+} \frac{1}{4\pi}\int_{-\ell}^{\ell} \xxa \rho_r(x;\eps)^2\xxe \dd{x} \quad (r=\pm),  
\end{equation}  
\begin{equation} 
\label{W3pm2} 
W_{3,r} =\lim_{\eps\to 0^+}  \frac1{6\pi} \int_{-\ell}^{\ell} \xxa \rho_r(x;\eps)^3\xxe \dd{x} \quad (r=\pm) , 
\end{equation}  
\begin{equation} 
\label{cC2} 
\cC = -\lim_{\eps\to 0^+}  \frac1{4\pi} \int_{-\ell}^{\ell} \sum_{r=\pm} \xxa \rho_r(x;\eps)(T\rho_{r,x})(x;\eps) +  \rho_{-r}(x;\eps)(\tilde{T}\rho_{r,x})(x;\eps)  \xxe \dd{x} , 
\end{equation} 
using the notation $\rho_{r,x}(x;\eps)\coloneqq \partial_x \rho_{r}(x;\eps)$ and the definition \eqref{TT} of the integral operators $T$, $\tilde{T}$. 
\end{lemma} 

\begin{remark} 
The equality \eqref{W2pm2} should be interpreted in the sense of sesquilinear forms, i.e., \eqref{W2pm2} means that
$$\langle \eta, W_{2,r} \eta'\rangle  = \lim_{\eps\to 0^+} \frac{1}{4\pi}\int_{-\ell}^{\ell} \langle \eta,\xxa \rho_r(x;\eps)^2\xxe \eta'\rangle \dd{x} \qquad \text{for $r=\pm$ and all $\eta,\eta'\in\cD$}.$$ 
A similar interpretation applies to \eqref{W3pm2}--\eqref{cC2} and similar equalities elsewhere in the paper.
\end{remark}

\begin{proof}[Proof of Lemma \ref{lem:WjrcC}] 
We start with 
\begin{equation*} 
\lim_{\eps\to 0}\int_{-\ell}^{\ell}\xxa \rho_r(x;\eps)^2\xxe \dd{x} = (2\kappa)^2(2\ell)\Bigg( \sum_{n\in\Z_{\neq 0}}\xxa b_{r,-n}b_{r,n}\xxa + Q_r^2\Bigg) = (2\kappa)^2 (2\ell) \frac2{(2\kappa)} W_{2,r}, 
\end{equation*}
which gives \eqref{W2pm2} since $(2\kappa)^2(2\ell) 2/(2\kappa)=4\pi$. 
Similarly, 
\begin{multline*} 
\lim_{\eps\to 0}\int_{-\ell}^{\ell}\xxa \rho_r(x;\eps)^3\xxe \dd{x} = (2\kappa)^3(2\ell)\Bigg( 
\sum_{\substack{n,m\in\Z_{\neq 0}\\n+m\neq 0}} 
\xxa b_{r,-n}b_{r,-m}b_{r,n+m}\xxe  \\
+ 3Q_r \sum_{n\in\Z_{\neq 0}} \xxa b_{,-n}b_{r,n}\xxe + Q_r^{3} \Bigg) = (2\kappa)^3 (2\ell) \frac3{(2\kappa)^2} W_{3,r}, 
\end{multline*}
which gives  \eqref{W3pm2} since $(2\kappa)^3(2\ell) 3/(2\kappa)^2=6\pi$. 

To derive the expression \eqref{cC2} for $\cC$, we use \eqref{cC} and the inverse of the Bogoliubov transformation \eqref{BT}, 
\begin{equation} 
\label{BTinv} 
a_{r,n} = c_n b_{r,n} + s_n b_{-r,-n} \quad (n\in\Z_{\neq 0})
\end{equation} 
to compute 
\begin{multline} 
\label{cC3}
\cC =(2\kappa)^2\frac{1}{2}\sum_{r=\pm }\sum_{n\in\Z} |n|\xxa \bigl(c_n b_{r,-n} + s_n b_{-r,n}  \bigr)\bigl( c_n b_{r,n} + s_n b_{-r,-n} \bigr)\xxe \\
= 
(2\kappa)^2\frac{1}{2}\sum_{r=\pm }\sum_{n\in\Z} |n|\xxa (c_n^2+s_n^2) b_{r,-n}b_{r,n} + 2 c_ns_n b_{-r,n}b_{r,n}\xxe .
\end{multline} 
On the other hand, Lemma~\ref{lem:TT} in Appendix~\ref{app:special} and the identity
\begin{equation} 
\frac1{2\ell}\int_{-\ell}^\ell \ee^{2\ii n\kappa (x-x')} \rho_r(x';\eps)\,\dd{x'} =   2\kappa \ee^{2\kappa(\ii nx -|n|\epsilon)} b_{r,rn} \quad (r=\pm,n\in\Z_{\neq 0}) 
\end{equation} 
give
\begin{equation}\label{TrhotTrho}
\begin{split} 
(T\rho_{r,x})(x; \epsilon) =&  -\sum_{n\in\Z_{\neq 0}} (2\kappa)^2 |n|(c_n^2+s_n^2)\ee^{2\kappa(\ii r nx-|n|\epsilon)} b_{r,n} , \\
(\tilde{T}\rho_{r,x})(x; \epsilon) =& -\sum_{n\in\Z_{\neq 0}} (2\kappa)^2 |n|2c_ns_n \ee^{2\kappa(\ii r nx-|n|\epsilon)} b_{r,n}  , 
\end{split} \quad (r=\pm).
\end{equation} 
Thus, noting that the terms with $n=0$ make no contribution after integration,
\begin{multline} 
\lim_{\eps\to 0^+}\int_{-\ell}^{\ell} \sum_{r=\pm} \xxa \rho_r(x;\eps)(T\rho_{r,x})(x;\eps) +  \rho_{-r}(x;\eps)(\tilde{T}\rho_{r,x})(x;\eps)  \xxe \dd{x} \\
= -\lim_{\eps\to 0^+}\int_{-\ell}^{\ell} \sum_{r=\pm} \sum_{n,m\in\Z_{\neq 0}} \xxa  (2\kappa)^3 \ee^{2\kappa(\ii r nx-|n|\epsilon)} b_{r,n} |m|(c_m^2+s_m^2)\ee^{2\kappa(\ii r mx-|m|\epsilon)} b_{r,m}   
\\ +  (2\kappa)^3 \ee^{2\kappa(-\ii r nx-|n|\epsilon)} b_{-r,n} |m|2c_ms_m\ee^{2\kappa(\ii r mx-|m|\epsilon)} b_{r,m}    \xxe \dd{x} \\
= -(2\kappa)^3(2\ell)  \sum_{r=\pm}\sum_{n\in\Z_{\neq 0}} |n|\xxa (c_n^2+s_n^2)b_{r,-n}  b_{r,n} +2 c_ns_n b_{-r,n}  b_{r,n} \xxe \; =-2(2\kappa)(2\ell)\cC, 
\end{multline} 
where we have used \eqref{cC3} in the last step. Since $2(2\kappa)(2\ell)=4\pi$, this proves \eqref{cC2}. 
\end{proof} 

Inserting the results in Lemma~\ref{lem:WjrcC} into the definition of $\cH_{3,\nu}$ in \eqref{cH3}, we obtain $\cH_{3,\nu}$ in \eqref{cH3nu}, up to the limiting procedure $\lim_{\eps\to 0^+}$: with the latter limit understood, \eqref{cH3nu} is mathematically precise as it stands. 

To state a precise version of the quantum ncILW equation, we find it convenient to define a regularized product $\circ$ as follows. 

\begin{definition}[Regularized product]\label{def:circ} 
For $f(x)$ and $g(x)$ two $2\ell$-periodic quantum fields of the variable $x\in\R$, let
\begin{equation} 
(f\circ g)(x;\eps) \coloneqq (\delta_\eps*(fg))(x) 
\end{equation} 
where $(fg)(x)\coloneqq f(x)g(x)$ is the point-wise product and
\begin{equation} 
\label{star}
(\delta_\eps * f)(x)\coloneqq \int_{-\ell}^\ell \delta(x-x';\eps)f(x')\,\dd{x}' \quad (\eps>0) 
\end{equation} 
with the regularized Dirac delta function $\delta(x;\eps)$ given by \eqref{deltaeps}.
(Note that the $*$ in \eqref{star} indicates convolution, not to be confused with the involution of sequences in \eqref{eq:star}.)  
\end{definition} 

It is interesting to write these operations in terms of Fourier series. Using the conventions 
\begin{equation}\label{fn}
f(x)=\sum_{n\in\Z}\hat{f}_n\ee^{2\kappa\ii nx}\quad \Leftrightarrow \quad \hat{f}_n = \frac1{2\ell}\int_{-\ell}^\ell f(x) \ee^{-2\kappa\ii nx}\dd{x} 
\end{equation} 
and similarly for $g(x)$, one can easily check that 
\begin{equation}\label{star1}
(\delta_\eps*f)(x) = \sum_{n\in\Z} \hat{f}_n\ee^{2\kappa(\ii nx-|n|\eps)}
\end{equation} 
and 
\begin{equation}\label{circ1}  
(f\circ g)(x;\eps) = \sum_{n,m\in\Z}  \hat{f}_n \hat{g}_m \ee^{2\kappa\ii (n+m)x}\ee^{-2\kappa|n+m|\eps}. 
\end{equation} 
In particular, our regularized bosons are a special case of this: $\rho_r(x;\eps)=(\delta_\eps*\rho_r)(x)$ (where $\rho_r(x)$ is obtained from $\rho_r(x;\eps)$ by setting $\eps=0$). Clearly, $(\delta_\eps*f)(x) \to f(x)$ and $(f\circ g)(x;\eps) \to f(x)g(x)$ as $\eps\to 0^+$ if these limits exist as operator-valued continuous functions, but this need not be the case:
in general, we think of $f(x;\eps)\coloneqq (\delta_\eps*f)(x)$ as a regularized version of $f(x)$. Moreover, $f(x;\eps)g(x;\eps)$  and $(f\circ g)(x;\eps)$ are two regularized versions of the point-wise product $f(x)g(x)$ which, for operator-valued distributions, is usually ill-defined; the latter two regularizations are different: 
\begin{equation} 
(f\circ g)(x;\eps)-f(x;\eps)g(x;\eps) = \sum_{n,m\in\Z}  \hat{f}_n \hat{g}_m \ee^{2\kappa\ii (n+m)x}\big( \ee^{-2\kappa|n+m|\eps}- \ee^{-2\kappa(|n|+|m|)\epsilon}\big)
\end{equation}
is non-zero in general for $\eps>0$, but they become equal in the limit $\eps\to 0^+$. 

We are now ready to make precise in what sense $\cH_{3,\nu}$ defines a quantum version of the ncILW equation; we give two complementary formulations, $(a)$ and $(b)$. 

\begin{theorem}\label{thm:qncILW} 
$(a)$ The operator $\cH_{3,\nu}$ defined in Definition~\ref{def:cH} can be written as 
\begin{equation} 
\label{cH3uv} 
\cH_{3,\nu} = \lim_{\eps\to 0} \frac{2}{\pi g}\int_{-\ell}^{\ell} \xxa \frac{1}{3}(\hat u^3+\hat v^3) + \frac{g-1}{4}(\hat uT\hat u_x +\hat vT\hat v_x + \hat u\tilde{T}\hat v_x + \hat v\tilde{T}\hat u_x)  \xxe \dd{x} 
\end{equation} 
with $g=\nu^2$,  $\hat u=\hat u(x;\eps)\coloneqq \nu \rho_{+}(x;\eps)/2$ and $\hat v=\hat v(x;\eps)\coloneqq \nu \rho_{-}(x;\eps)/2$ regularized boson fields satisfying the canonical commutator relations:  
\begin{equation} 
\label{CCRuv}
[\hat{u}(x;\eps),\hat{u}(x';\eps')] = -\frac{g\pi\ii}{2}\delta'(x-x';\eps+\eps'),\quad  [\hat{v}(x;\eps),\hat{v}(x';\eps')] = \frac{g\pi\ii}{2}\delta'(x-x';\eps+\eps'),
\end{equation}   
and $[\hat{u}(x;\eps),\hat{v}(x';\eps')]=0$, and $T$, $\tilde{T}$ the integral operators in \eqref{TT}; we use the shorthand notation $\hat u^3=\hat u(x;\eps)^3$, $\hat u T\hat u_x=\hat u(x;\eps) (T\hat u_x)(x;\eps)$, etc. 

\noindent $(b)$  The Heisenberg equations of motion following from $(a)$ are 
\begin{equation}\label{qncILW} 
\begin{split} 
\partial_t\hat{u} \coloneqq \ii[\cH_{3,\nu},\hat{u}] & = -2\xxa\hat{u}\circ\hat{u}_x\xxe -\half(g-1)(T\hat u_{xx}+\tilde{T}\hat{v}_{xx}), \\
\partial_t\hat{v} \coloneqq \ii[\cH_{3,\nu},\hat{v}] & = +2\xxa\hat{v}\circ\hat{v}_x\xxe +\half(g-1)(T\hat v_{xx}+\tilde{T}\hat{u}_{xx}),  
\end{split} 
\end{equation} 
with $\hat{u}=\hat{u}(x,t;\eps)$, $T\hat{u}_{xx}=(T\hat{u}_{xx})(x,t;\eps)$, etc., and $\circ$ the regularized version of the pointwise product in Definition~\ref{def:circ}. 
\end{theorem} 

We regard Theorem~\ref{thm:qncILW}$(a)$ as a possible definition of the quantum ncILW model: it is given by the  Hamiltonian $\cH_{3,\nu}$ written in terms of quantum fields characterized by commutator relations. In Theorem~\ref{thm:qncILW}$(b)$, we state the quantum analog of \eqref{ncILW}. 

\begin{proof}[Proof of Theorem~\ref{thm:qncILW}] Part $(a)$ is a simple consequence of Lemma \ref{lem:bosons}, Lemma~\ref{lem:WjrcC} and \eqref{cH3}. 

Part $(b)$ is implied by commutator relations which, for later reference, we collect in a lemma. 

\begin{lemma}\label{lem:W3commutators}
The following commutation relations hold, 
\begin{equation}\label{W2rrhor} 
\ii[W_{2,r},\rho_{r'}(x;\eps)] = -r\delta_{r,r'}\rho_{r,x}(x;\eps)\quad (r=\pm), 
\end{equation} 
\begin{equation}\label{W3rrhor} 
\ii[W_{3,r},\rho_{r'}(x;\eps)] = -2r\delta_{r,r'}(\rho_r\circ\rho_{r,x})(x;\eps)\quad (r=\pm), 
\end{equation} 
\begin{equation}\label{cCrhor} 
\ii[\cC,\rho_{r}(x;\eps)] = r\Bigl( (T\rho_{r,xx})(x;\eps) + (\tilde{T}\rho_{-r,xx})(x;\eps) \Bigr)  \quad (r=\pm), 
\end{equation} 
with $\circ$ in Definition~\ref{def:circ}, $\rho_{r,x}(x;\eps)=\partial_x\rho_{r}(x;\eps)$, $\rho_{r,xx}(x;\eps)=\partial^2_x\rho_{r}(x;\eps)$, and $T$, $\tilde{T}$ in \eqref{TT}. 
\end{lemma}

\begin{proof} 
We start by proving \eqref{cCrhor}.  We use \eqref{rhor}, \eqref{cC3}, and the identity (note that $\xxa b_{r',\mp n}b_{\pm r',n}\xxe$ and $b_{r',\mp n}b_{\pm r',n}$ differ only by a constant)
\begin{multline}
\sum_{r'=\pm}[\xxa b_{r',\mp n}b_{\pm r',n}\xxe, \rho_r(x;\eps)] = \sum_{m\in\Z_{\neq 0}} 2\kappa \ee^{2\kappa(\ii r mx-|m|\epsilon)} [b_{r',\mp n}b_{\pm r',n},b_{r,m}]  \\
= \sum_{m\in\Z_{\neq 0}} 2\kappa \ee^{2\kappa(\ii r mx-|m|\epsilon)}\big( b_{r',\mp n}\delta_{\pm r',r}n\delta_{n,-m} \mp \delta_{r',r}n\delta_{\mp n,-m}b_{\pm r',n} \big) \\
= 2\kappa n \big(  \ee^{2\kappa(-\ii nr x-|n|\epsilon)} b_{\pm r,\mp n} \mp   \ee^{2\kappa(\pm \ii r nx-|n|\epsilon)} b_{\pm r,n}\big)
\end{multline} 
to compute 
\begin{multline} \label{iCrhorxeps}
\ii[\cC,\rho_{r}(x;\eps)] = \ii (2\kappa)^2\frac{1}{2}\sum_{n\in\Z}|n|\sum_{r'=\pm }\big[ \xxa (c_n^2+s_n^2) b_{r',-n}b_{r',n} + 2 c_ns_n b_{-r',n}b_{r',n}\xxe,\rho_r(x;\eps)\big] \\
=  \ii (2\kappa)^3\frac{1}{2}\sum_{n\in\Z}|n|n\big[ (c_n^2+s_n^2)( \ee^{2\kappa(-\ii nr x-|n|\epsilon)} b_{r,-n} - \ee^{2\kappa(\ii r nx-|n|\epsilon)} b_{r,n}) 
\\ + 2 c_ns_n ( \ee^{2\kappa(-\ii nr x-|n|\epsilon)} b_{-r,n} +   \ee^{2\kappa(-\ii r nx-|n|\epsilon)} b_{-r,n})\big] \\
= - \ii(2\kappa)^3\sum_{n\in\Z}|n|n\big[ (c_n^2+s_n^2) \ee^{2\kappa(\ii r nx-|n|\epsilon)} b_{r,n}- 2 c_ns_n  \ee^{2\kappa(-\ii nr x-|n|\epsilon)} b_{-r,n} \big] \\
= -r\partial_x \sum_{n\in\Z} (2\kappa)^2|n|\big[ (c_n^2+s_n^2) \ee^{2\kappa(\ii r nx-|n|\epsilon)} b_{r,n}+ 2 c_ns_n  \ee^{2\kappa(-\ii nr x-|n|\epsilon)} b_{-r,n} \big], 
\end{multline}
where we renamed the summation index $n\to -n$ in one of the terms in the third step. Comparing (\ref{iCrhorxeps}) with \eqref{TrhotTrho} and using that $\partial_x$ commutes with the integral operators $T$ and $\tilde{T}$ when acting on zero-mean functions, we obtain \eqref{cCrhor}. 

The commutator relations \eqref{W2rrhor}--\eqref{W3rrhor} are proved in Appendix~\ref{app:proofssec5} by a different method. 
\end{proof} 

Recalling \eqref{cH3},  we get from \eqref{W3rrhor}--\eqref{cCrhor}
\begin{equation} 
\partial_t\rho_r \coloneqq \ii[\cH_{3,\nu},\rho_r] = -\nu r\rho_r\circ\rho_{r,x} -\half r(\nu^2-1)(T\rho_{r,xx} + \tilde{T}\rho_{-r,xx})\quad (r=\pm), 
\end{equation} 
suppressing the common argument $(x,t;\eps)$. This is equivalent to \eqref{qncILW}. 
\end{proof} 

\section{Second quantization of a generalized eCS model}\label{sec:eCSgen}
We generalize the results in Section~\ref{subsec:eCS} and obtain a second quantization of a generalized eCS model that can describe arbitrary numbers of four types of particles. 

To motivate this generalization, we note that the operator $\cH_{3,\nu}$ defined in \eqref{cH3} obeys $\cH_{3,\nu}=-\nu^2\cH_{3,-1/\nu}$. 
This suggests that by changing $\nu\to -1/\nu$ in \eqref{cH3PhiN} and using the latter symmetry relation, one should obtain a second quantization of the eCS model with coupling $(-1/\nu)^2=1/g$. However, for this to be possible, the anyons $\phi_{-1/\nu,r}(x;\eps)$ must also be well-defined, and, for this to be the case, we need to find a parameter $\nu_0>0$ such that not only $\nu/\nu_0$ but also $-1/\nu\nu_0$ is an integer; fortunately such a $\nu_0$ often exists. 

\begin{lemma}\label{lem:r0s0} 
There exists a parameter $\nu_0>0$ such that both $\nu/\nu_0$ and $-1/\nu\nu_0$ are integers if and only if $g\coloneqq \nu^2$ is rational and, in this case, 
\begin{equation}\label{eq:r0s0} 
g =\frac{r_0}{s_0}>0,\quad \nu_0=\frac1{\sqrt{r_0s_0}} 
\quad  (r_0,s_0\in\Z_{\neq 0}). 
\end{equation} 
\end{lemma} 
\begin{proof} 
We set $r_0\coloneqq \nu/\nu_0$ and $s_0\coloneqq 1/\nu\nu_0$. This implies $\nu=r_0\nu_0=1/s_0\nu_0$ and $\nu_0=\nu/r_0=1/\nu s_0$, giving the relation in \eqref{eq:r0s0}. 
Conversely, \eqref{eq:r0s0} and $\nu=\pm\sqrt{g}$ imply $\nu/\nu_0=\pm \sqrt{r_0/s_0}\sqrt{r_0s_0}=\pm |r_0|$ and $-1/\nu\nu_0=\mp \sqrt{s_0/r_0}\sqrt{r_0s_0}=\mp |s_0|$. 
\end{proof} 

Thus $\cH_{3,\nu}$ provides a four-fold second quantization of the eCS model if and only if $g>0$ is rational. For given rational $g$, one can write $g=r_0/s_0$ with positive integers $r_0$ and $s_0$ which do not have a common divisor and, by that, the integers $r_0$ and $s_0$ are uniquely determined by $g$. However, this can be generalized as follows: if $(r_0,s_0)$ is a pair of integers satisfying \eqref{eq:r0s0}, then $(r_0',s_0')=(kr_0,ks_0)$ is another such pair for any $k\in\Z_{\neq 0}$. Our construction is different for different such pairs (this is shown in a special case in \cite{atai2017}); for this reason, we allow for this generalization. Thus, our construction is determined by a pair of non-zero integers $(r_0,s_0)$ such that $r_0/s_0=g>0$, rather than by $g$ itself. For simplicity, we do not indicate this generalization in our notation. 

These four different second quantizations of the eCS model correspond to the four different anyons $\phi_{\nu,r}(x;\eps)$ and  $\phi_{-1/\nu,r}(x;\eps)$, $r=\pm$. 
The following result shows that Proposition~\ref{prop:eCS} has an interesting  generalization where $\phi^N_{\nu,r}(\vx;\eps)=\phi_{\nu,r}(x_1;\eps)\cdots \phi_{\nu,r}(x_N;\eps)$ can be replaced by any mixed product $ \phi_{\nu_1,r_1}(x_1;\eps)\cdots  \phi_{\nu_N,r_N}(x_N;\eps)$, with $\nu_j\in\{\nu,-1/\nu\}$ and $r_j=\pm$ for $j=1,\ldots,N$ arbitrary.  We recall that $\Omega$ is the vacuum in the Fock space and $\wp_{r,r'}(x;\eps)$ for $r,r'=\pm$ are regularized eCS potentials defined in Definition~\ref{def:wprreps}. We find it convenient to parametrize the choices $\nu_j=\nu$ or $\nu_j=-1/\nu$  as $\nu_j=m_j\nu$ with $m_j=1$ or $m_j = -1/\nu^2=-1/g$. 

\begin{theorem}\label{thm:eCSgen} 
$(a)$ Let $r_0,s_0\in\Z$ such that $g=r_0/s_0>0$ and $\nu_0=1/\sqrt{r_0s_0}$, $N$ an arbitrary positive integer, $\nu=\sqrt{g}$ and 
\begin{equation} 
\phi^N_{\vr, \vm}(\vx;\eps)\coloneqq \phi_{r_1,m_1\nu}(x_1;\eps)\cdots  \phi_{r_N,m_N\nu}(x_N;\eps)
\end{equation}
with arbitrary parameters $\vm\coloneqq (m_1,\ldots,m_N)\in\{1,-1/g\}^N$ and $\vr=(r_1,\ldots,r_N)\in\{\pm\}^N$, for $\vx=(x_1,\ldots,x_N)\in[-\ell,\ell]^N$ and $\eps>0$. Then the operator $\cH_{3,\nu}$ in Definition~\ref{def:cH} satisfies 
\begin{equation} 
[\cH_{3,\nu},\phi^N_{\vr,\vm}(\vx;\eps)]\Omega = \left( H^{(\vr,\vm)}_{N;\nu^2}(\vx;\eps) + c^{(\vm)}_{N;\nu}(\eps) \right)\phi^N_{\vr,\vm}(\vx;\eps)\Omega +  \Psi^N_{\vr,\vm}(\vx;\eps)\Omega
\end{equation} 
with the following regularized Hamiltonian, 
\begin{equation} 
H^{(\vr,\vm)}_{N;g}(\vx;\eps)\coloneqq -\sum_{j=1}^N \frac1{2m_j} \frac{\partial^2}{\partial x_j^2} +  \sum_{1\leq j<k\leq N} m_jm_kg(g-1) \wp_{r_j,r_k}(x_j-x_k; 2\eps),  
\end{equation} 
the following constant, 
\begin{equation} 
c^{(\vm)}_{N;\nu}(\eps) = \half \nu \sum_{j=1}^N (\nu m_j)^3  c_\epsilon 
\end{equation} 
with $c_\epsilon$ given by \eqref{ceps}, and the following correction terms vanishing in the limit $\eps\to 0^+$, 
\begin{multline} 
 \Psi^N_{\vr,\vm}(\vx;\eps) = \nu^2 \sum_{j=1}^N m_j \phi_{r_1, m_1 \nu}(x_1;\eps)\cdots  \phi_{r_{j-1},m_{j-1}\nu}(x_{j-1};\eps)\xxa \cR_{r_j,m_j\nu}(x_j;\eps) \phi_{r_j,m_j\nu}(x_{j};\eps)\xxe
 \\ \times    \phi_{r_{j+1},m_{j+1}\nu}(x_{j+1};\eps) \cdots \phi_{r_N,m_N\nu}(x_N;\eps), 
\end{multline} 
where
\begin{align}\nonumber
\mathcal{R}_{r,\nu}(x;\epsilon)\coloneqq &\; (2\kappa)^2\sum_{n,m=1}^{\infty} \bigg( \nu  s_n^2\left(\ee^{2\kappa \ii mr x} b_{r,m}-\ee^{-2\kappa \ii m r x} b_{r,-m}\right) \ee^{-2\kappa n\epsilon}
-\ee^{-2\kappa \ii(n-m) r x}\xxa b_{r,-n} b_{r,m}\xxe \bigg)
	\\ \label{Rrnuxepsilondef}
& \times \left(\ee^{-2\kappa (n+m)\epsilon} - \ee^{-2\kappa |n-m|\epsilon} \right).
\end{align}

$(b)$ The identity
\begin{align} \label{OmegaH3phiOmega}
& \langle\Omega,  [\cH_{3,\nu},\phi_{\vr,\vm}^N(\vx;\eps)^\dag R_+^{\mu_+} R_-^{\mu_-} \phi_{\vr',\vm'}^M(\vy;\eps)]\Omega\rangle=0
\end{align} 
holds for any integers $\mu_\pm \in \Z$, for any non-negative integers $N, M$, and for any choice of the parameters 
\begin{align*}
& \vm=(m_1\ldots,m_N)\in\{1,-1/\nu^2\}^N, && \vm'=(m_1'\ldots,m_M')\in\{1,-1/\nu^2\}^M,
	\\
& \vr = (r_1, \ldots,r_N)\in\{\pm\}^N, && \vr'=(r_1',\ldots,r_M')\in\{\pm\}^M, 
	\\
& \vx=(x_1,\ldots,x_N)\in[-\ell,\ell]^N, && \vy =(y_1,\ldots,y_N)\in[-\ell,\ell]^M.
\end{align*}
\end{theorem}   

\begin{remark}
The identity \eqref{OmegaH3phiOmega} is primarily of interest in the case when $\mu_+$ and $\mu_-$ are given by
\begin{align}\label{mursumsum}
& \mu_r = \sum_{j=1}^N \delta_{r_j,r} m_j\nu/\nu_0 -  \sum_{j=1}^M  \delta_{r'_j,r} m'_j\nu/\nu_0 \quad (r=\pm).
\end{align}
Indeed, if \eqref{mursumsum} is not satisfied then the left-hand side of \eqref{OmegaH3phiOmega} vanishes as a consequence of \eqref{Phiexpectation}.
\end{remark}

It is interesting to note that parts (a) and (b) of Theorem~\ref{thm:eCSgen} are  needed to derive kernel function identities generalizing \cite[Proposition~3]{langmann2004}. 
We do not discuss this further in the present paper since these kernel function identities are already known: they can be obtained from \cite[Corollary 2.2]{langmann2010} by shifting parts of the variables by the imaginary half-period $\ii\delta$, as explained in Appendix~\ref{app:eCSgen}. We expect that these kernel function identities can be used to construct eigenfunctions of the generalized eCS model, in generalization of results in \cite{langmann2014}.

\subsection{Proof of Theorem~\ref{thm:eCSgen}}

In addition to the operators $W_{2,r}$ and $W_{3,r}$ defined in \eqref{W2pm2}--\eqref{W3pm2}, we also consider the operator
\begin{equation} 
\label{W1pm2}
W_{1,r} = \lim_{\eps\to 0^+} \frac{1}{2\pi}\int_{-\ell}^{\ell} \rho_r(x;\eps)\, \dd{x} \quad (r=\pm).
\end{equation}  
Note that
\begin{equation}\label{Wkdefinition}
\begin{split}
&W_{k,r} = \lim_{\epsilon\to 0^+} \frac{1}{2k\pi}\int_{-\ell}^{\ell} \xxa \rho_r(x;\eps)^{k}\xxe\, \dd{x},\qquad (k=1,2,3, \; r = \pm).
\end{split}
\end{equation}

\begin{lemma}\label{lemma2}
For $r, r' = \pm$ and $k = 1,2,3$, the operator $W_{k,r'}$ and the anyon $\phi_{r,\nu}(x;\epsilon)$ defined in Definition \ref{def:anyons} satisfy the commutation relations
\begin{align}\label{W1commutators}
 [W_{1,r'}, \phi_{r,\nu}(x;\epsilon)]= &\; \nu \delta_{r,r'} \phi_{r,\nu}(x;\epsilon)
	\\ \label{W2commutators}
 [W_{2, r'},\phi_{r,\nu}(x;\epsilon)]= &\; \ii r \delta_{r,r'} (\phi_{r, \nu})'(x;\epsilon),
	\\ \nonumber
[W_{3,r'},\phi_{r,\nu}(x;\epsilon)]=&\; \delta_{r,r'}\Big(- \nu^{-1} (\phi_{r,\nu})''(x;\epsilon)+\mathrm{i} r (\nu^2-1)\xxa\rho_r'(x;\epsilon)\phi_{r,\nu}(x;\epsilon)\xxe
	\\ \label{W3commutators}
&+ 2\nu\xxa \mathcal{R}_{r,\nu}(x;\epsilon)\phi_{r,\nu}(x;\epsilon)\xxe+ \nu^3c_\epsilon \phi_{r,\nu}(x;\epsilon)\Big),
\end{align}
with $c_\epsilon$ as in \eqref{ceps} and $\mathcal{R}_{r,\nu}$ defined by (\ref{Rrnuxepsilondef}).
\end{lemma}
\begin{proof}
See Appendix \ref{lemma2app}.
\end{proof}

Recalling \eqref{cH3}, we compute the commutators of $\sum_{r=\pm} W_{3,r}$ and $\cC$ with $\phi^N_{\vr,\vm}(\vx;\eps)$.

First, using Lemma \ref{lemma2}, we compute
\begin{align}
& \frac{\nu}{2} \sum_{r=\pm} [W_{3, r},\phi^N_{\vr,\vm}(\vx;\eps)] \nonumber
	\\ \nonumber
&= \frac{\nu}{2} \sum_{r=\pm}  \sum_{j=1}^{N} \phi_{r_1,m_1\nu}(x_1;\epsilon)\cdots[W_{3, r},\phi_{r_j,m_j\nu}(x_j;\epsilon)]\cdots {\phi_{r_N,m_N\nu}}(x_N;\epsilon)
	 \\\nonumber
&= - \sum_{j=1}^N  \frac{1}{2m_j} \partial_{x_j}^2 \phi^N_{\vr,\vm}(\vx;\eps)
+ \mathcal{X} 
+ \frac{\nu}{2} c_\epsilon \sum_{j=1}^N m_j^3\nu^3 \phi^N_{\vr,\vm}(\vx;\eps),
	\\ \label{W3_actiondeformed}
&\phantom{=\;}+ \nu^2  \sum_{j=1}^{N} m_j \phi_{r_1, m_1\nu}(x_1;\epsilon)\cdots\xxa \mathcal{R}_{r_j, m_j\nu}(x_j;\epsilon)\phi_{r_j, m_j\nu}(x_j;\epsilon)\xxe\cdots \phi_{r_N, m_N\nu}(x_N;\epsilon)
\end{align}
where $\mathcal{X}$ is shorthand notation for
\begin{align}\label{calXdef}
 \mathcal{X} := \frac{\nu}{2} \sum_{j=1}^{N} &\mathrm{i} r_j(m_j^2\nu^2-1) \phi_{r_1,m_1\nu}(x_1;\epsilon)\cdots\xxa \rho_{r_j}'(x_j;\epsilon)\phi_{r_j,m_j\nu}(x_j;\epsilon)\xxe
\cdots \phi_{r_N,m_N\nu}(x_N;\epsilon).
 \end{align}

The second commutator is given by the following lemma. 
 
\begin{lemma}\label{Clemma}
The operator $\mathcal{C}$ satisfies
\begin{align}\label{Ccommutator}
 \frac{1-\nu^2}{2} [\mathcal{C},\phi^N_{\vr,\vm}(\vx;\eps)]\Omega
=& - \mathcal{X} \Omega
+  \sum_{1 \leq j < k \leq N} m_j m_k  \nu^2(\nu^2 - 1) \wp_{r_j, r_k}(x_k-x_j; 2\epsilon) \phi^N_{\vr,\vm}(\vx;\eps)\Omega,
\end{align}
where $\mathcal{X}$ is given by (\ref{calXdef}). 
\end{lemma}
\begin{proof}
See Appendix \ref{Clemmaapp}.
\end{proof}

Thus, \eqref{cH3}, \eqref{W3_actiondeformed}, and \eqref{Ccommutator} imply that
\begin{align*}
 &[\cH_{3,\nu},\phi^N_{\vr,\vm}(\vx;\eps)]\Omega
 =  - \sum_{j=1}^N  \frac{1}{2m_j} \partial_{x_j}^2 \phi^N_{\vr,\vm}(\vx;\eps)\Omega
+ \mathcal{X} \Omega
+ \frac{\nu}{2} c_\epsilon \sum_{j=1}^N m_j^3\nu^3 \phi^N_{\vr,\vm}(\vx;\eps)\Omega
	\\
&+ \nu^2  \sum_{j=1}^{N} m_j \phi_{r_1, m_1\nu}(x_1;\epsilon)\cdots\xxa \mathcal{R}_{r_j, m_j\nu}(x_j;\epsilon)\phi_{r_j, m_j\nu}(x_j;\epsilon)\xxe\cdots \phi_{r_N, m_N\nu}(x_N;\epsilon)\Omega
	\\
&- \mathcal{X} \Omega
+\sum_{1 \leq j < k \leq N} m_j m_k   \nu^2(\nu^2 - 1) \wp_{r_j, r_k}(x_j-x_k; 2\epsilon) \phi^N_{\vr,\vm}(\vx;\eps)\Omega, 
\end{align*}
which proves part $(a)$ of Theorem \ref{thm:eCSgen}. 

To prove $(b)$, let $\Phi \coloneqq \; \xxa \phi_{\vr,\vm}^N(\vx;\eps)^\dag R_+^{\mu_+} R_-^{\mu_-} \phi_{\vr',\vm'}^M(\vy;\eps) \xxe$.
Up to a multiplicative constant, $\Phi$ equals $\phi_{\vr,\vm}^N(\vx;\eps)^\dag R_+^{\mu_+} R_-^{\mu_-} \phi_{\vr',\vm'}^M(\vy;\eps)$. Moreover, $\Phi$ has the form \eqref{Phidef}. Indeed, we have
$$\Phi = \; \xxa R_+^{\mu_+'} R_-^{\mu_-'}\ee^{\ii J(\alpha)}\xxe$$
where 
\begin{align*}
&\mu_r' \coloneqq \mu_r - \bigg(\sum_{j=1}^N \delta_{r_j,r} m_j\nu/\nu_0 -  \sum_{j=1}^M  \delta_{r'_j,r} m'_j\nu/\nu_0\bigg) \quad (r=\pm),
	\\
& J(\alpha) \coloneqq \sum_{j=1}^Nm_j\nu (2r_j \kappa Q_r x_j+K_{r_j}(x_j;\epsilon)) -\sum_{j=1}^M m_j'\nu (2r_j' \kappa Q_{r_j'} y_j+K_{r_j}(y_j;\epsilon)).
\end{align*}
Thus, part $(b)$ of Theorem \ref{thm:eCSgen} is a direct consequence of the following lemma.

\begin{lemma}\label{lemma6.5}
 Suppose $\Phi = \Phi_{\mu'}(\alpha) = \; \xxa R_+^{\mu_+'} R_-^{\mu_-'}\ee^{\ii J(\alpha)}\xxe$, where $J(\alpha)$ has the form (\ref{Jdef}). Then,
 $$\langle\Omega,  [\cH_{3,\nu},\Phi] \Omega\rangle=0.$$
\end{lemma}
\begin{proof}
See Appendix \ref{C4app}.
\end{proof}
  
\section{Conclusions}\label{sec:conclusions} 
We constructed a system defined by the quantum field theory operator $\cH_{3,\nu}$ in Definition \ref{def:cH}. This system unifies seemingly different integrable systems: on the one hand, it can be regarded as a second quantization of a quantum many-body system generalizing the eCS model (see Theorem~\ref{thm:eCSgen}),  on the other hand, it corresponds to a quantum version of a soliton equation known as the periodic ncILW equation (see Theorem~\ref{thm:qncILW}).

In the trigonometric case $q=0$, the generalized eCS model reduces to the deformed Calogero-Sutherland model which has well-studied exact eigenfunctions known as super Jack polynomials \cite{sergeev2005,atai2019}; in this case, these eigenfunctions can be used to construct exact eigenstates of the trigonometric limit of $\cH_{3,\nu}$ \cite{atai2017}. We expect that, in a similar way, the eigenfunctions of the generalized eCS model provide eigenstates of the operator $\cH_{3,\nu}$ constructed and studied in the present paper. At present, these eigenfunctions are not known in full generality; see \cite{langmann2014} for some results in this direction. Thus, the present paper provides a strong motivation to construct the eigenfunctions of the generalized eCS model. 

The CFT model studied in this paper is defined on a boson Fock space $\cF$,  and the CFT operator $\cH_{3,\nu}$ can be interpreted as a Hamiltonian in a quantum field theory of bosons. By mathematical results known as boson-fermion correspondence, the Fock space $\cF$ is identical with the fermion Fock space $\cF_{\F}$ generated by fermions $\psi_r(x)$, $r=\pm$ and $x\in[-\ell,\ell]$, satisfying canonical anti-commutator relations 
\begin{equation}\label{CAR}  
\{\psi^{\pdag}_r(x),\psi_{r'}^\dag(y)\} = \delta_{r,r'}\delta(x-y),\quad \{\psi_r(x),\psi_{r'}(y)\} = 0 \quad (r,r'=\pm;x,y\in[-\ell,\ell]) 
\end{equation} 
with $\psi^\dag_r(x)\coloneqq \psi_r(x)^\dag$, and with the Hamiltonian (formally) defined as 
\begin{equation}\label{H0}
\cH_2^{\F} =  \sum_{r=\pm}\int_{-\ell}^\ell \nna \psi^\dag_r(x)r(-\ii\partial_x)\psi_r(x) \nne \dd{x}
\end{equation} 
where the colons indicate (additive) normal ordering; see Appendix~\ref{app:fermions} for a mathematically precise formulation. 
More specifically, as elaborated in Appendix~\ref{app:fermions}, the fermions $\psi_r(x)$ can be obtained by taking the limit $\eps\to 0$ of the anyons $\phi_{r,-1}(x;\eps)/\sqrt{2\ell}$ (i.e., $\nu=-1$) for $\nu_0=1$, and the fermion Hamiltonian $\cH_2^{\F}$ and the boson Hamiltonian $\cH_2$ defined in \eqref{cH2} are related as follows, 
\begin{equation}\label{H2H2F}
  \cH_2 = G^2\cH_2^{\F}-\kappa(1/\nu_0^2-1)\sum_{r=\pm} Q_r^2 
\end{equation} 
with the constant 
\begin{equation}\label{G} 
G\coloneqq \prod_{m=1}^\infty(1-q^{2m})
\end{equation} 
(this relation is well-known in the special case $\nu_0=1$ and $q=0$; see e.g.\ \cite[Proposition~2.14]{langmann2015}). In fact, there is a similar relation between the boson operator  $\cH_{3,\nu}$ defined in \eqref{cH3} and the fermion operator $\cH^{\F}_{3,\nu}$ defined by
\begin{multline}\label{cH3nufermion}
\cH^{\F}_{3,\nu} =  \frac{\nu}{2}\int_{-\ell}^{\ell} \sum_{r=\pm} \nna \psi_r^\dag(x)(-\partial_x^2)\psi_r(x)\nne \dd{x}   \\ + 
 \frac{1}2(\nu^2-1)G^2 \int_{-\ell}^\ell \fpint{-\ell}{\ell}  \sum_{r=\pm} J_r(x)\wp_1(x-x')J_r(x')\,\dd{x} \,\dd{x'} 
 \\ + 
 \frac{1}2(\nu^2-1)G^2 \int_{-\ell}^\ell \int_{-\ell}^{\ell}  \sum_{r=\pm} J_r(x)\wp_1(x-x'+\ii\delta)J_{-r}(x')\,\dd{x} \,\dd{x'},
\end{multline} 
where the fermion densities $J_r(x)$ are defined by $J_r(x) =  \nna \psi^\dag_r(x)\psi_r(x)\nne$. Indeed, as will be explained in Appendix \ref{app:fermions} (see \eqref{cH3nufermion1}), we have
\begin{align}\label{H3H3F}
\cH_{3,\nu} = G^2\cH^{\F}_{3,\nu} 
+  \frac{\nu}{2}G^2 \int_{-\ell}^{\ell} 
\sum_{r=\pm} \nna \psi_r^\dag(x) \bigg(2\tilde\kappa Q_r r\ii\partial_x + \frac{2+\nu_0}{3}(\tilde\kappa Q_r)^2 +c_0 \bigg) \psi_r(x) \nne \dd{x},
\end{align}
where $\tilde\kappa \coloneqq 2\kappa (1/\nu_0-1)$ and $c_0$ is defined in \eqref{constants}.

The operator in \eqref{cH3nufermion} has a natural physics interpretation: it describes non-relativistic fermions with a flavor index $r=\pm$ interacting with a particular two-body interaction of standard density-density type. Thus, the quantum field theory studied in the present paper has two complimentary descriptions: one in terms of bosons, and another one in terms of fermions. Note that $\nu^2=1$ is a free fermion point: in this special case, \eqref{cH3nufermion} defines a quantum field theory model of non-interacting fermions which, by the boson-fermion correspondence, is equivalent to a quantum field theory of interacting bosons. 

Ruijsenaars discovered a relativistic generalization of the eCS model \cite{ruijsenaars1987} when trying to develop mathematical tools to construct a quantum version of the sine-Gordon theory  \cite{ruijsenaars2001}. Quantum sine-Gordon theory is a relativistically invariant quantum field theory of interacting bosons, and it corresponds to a quantum version of a soliton equation known as the sine-Gordon equation. Moreover, by a famous conjecture due to Coleman \cite{coleman1975} recently proved at the free fermion point \cite{bauerschmidt2020}, quantum sine-Gordon theory is equivalent to the massive Thirring model, which is a relativistically invariant quantum field theory describing interacting fermions. We believe that it is possible to interpret quantum sine-Gordon theory posed on a circle as a second quantization of the elliptic Ruijsenaars model;\footnote{The limit where the circle becomes the real line corresponds to the hyperbolic Ruijsenaars model \cite{ruijsenaars2001}.} this suggests to us that  (i) the ncILW equation is a non-relativistic limit of the sine-Gordon equation, (ii) the Hamiltonian  $\cH_{3,\nu}^{\F}$ in \eqref{cH3nufermion} can be obtained as a non-relativistic limit of the massive Thirring model, and (iii) the relation \eqref{H3H3F} between  $\cH_{3,\nu}$ and $\cH_{3,\nu}^{\F}$ is a non-relativistic limit of the Coleman correspondence. 
Thus, we believe that our results can help to make quantum sine-Gordon theory and the Coleman conjecture mathematically precise.

\noindent
{\bf Acknowledgements.} {\it We thank Martin Halln\"as, Masatoshi Noumi, Hjalmar Rosengren,  and Junichi Shiraishi for inspiring discussions. B.K.B. acknowledges support from the Olle Engkvist Byggm\"{a}stare Foundation, Grant 211-0122. E.L. acknowledges support from the European Research Council, Grant Agreement No.\ 2020-810451. J.L. acknowledges support from the Ruth and Nils-Erik Stenb\"ack Foundation, the Swedish Research Council, Grant No.\ 2021-03877, and the European Research Council, Grant Agreement No. 682537.}

\appendix
\section{Special functions}\label{app:special} 
We collect the special functions and associated identities used in the main text. We largely follow the conventions of \cite{DLMF} except that we write the half-periods of the elliptic functions as $(\omega_1,\omega_2)$ rather than $(\omega_1,\omega_3)$. In particular, $\tau\coloneqq \omega_2/\omega_1$ and $q \coloneqq \ee^{\ii\pi \tau}$.  In addition, we use the abbreviation $\kappa\coloneqq \pi/2\omega_1$.  In the main text, the half-periods are set to $(\omega_1,\omega_2)=(\ell,\ii\delta)$. Throughout this appendix, $z$ is a complex variable. 

\subsection{Elliptic functions}
We recall the standard definitions of the Weierstrass $\zeta$- and $\wp$-functions with half-periods $(\omega_1,\omega_2)$  \cite[Section 23.2]{DLMF}, 
\begin{equation} 
\zeta(z) \coloneqq \frac1{z} + \sum_{(n,m)\in \Z^2\setminus(0,0)}\left(\frac1{z-2n\omega_1-2m\omega_2} + \frac1{2n\omega_1+2m\omega_2} +  \frac{z}{(2n\omega_1+2m\omega_2)^2}  \right) 
\end{equation} 
and 
\begin{equation} 
\wp(z) \coloneqq \frac1{z^2} + \sum_{(n,m)\in \Z^2\setminus(0,0)}
\left(\frac1{(z-2n\omega_1-2m\omega_2)^2} - \frac1{(2n\omega_1+2m\omega_2)^2} \right).  
\end{equation}  
The modified Weierstrass functions $\zeta_1(z)$ and $\wp_1(z)$ defined in the main text are related to the standard ones as follows, 
\begin{equation}\label{wp1def} 
\zeta_1(z) = \zeta(z) -\frac{\eta_1}{\omega_1}z,\quad \wp_1(z)=\wp(z)+\frac{\eta_1}{\omega_1},\quad \eta_1\coloneqq \zeta(\omega_1)
\end{equation} 
(this can be seen by comparing \cite[Eq.~23.8.4]{DLMF} with our definition \eqref{zeta1}, for example; recall that $\wp(z)=-\partial_z\zeta(z)$ and that we define $\wp_1(z)\coloneqq -\partial_z\zeta_1(z)$ in the main text). 

We use Jacobi's notation for the following two basic theta functions 
 \cite[Section~20.5(i)]{DLMF},
\begin{equation}\label{tet1tet4}
\begin{split} 
\theta_1(z,q) \coloneqq   & 2q^{1/4}\sin(z)\prod_{m=1}^\infty(1-q^{2m})(1-2q^{2m}\cos(2z)+q^{4m}),\\
 \theta_4(z,q) \coloneqq & \prod_{m=1}^\infty(1-q^{2m})(1-2q^{2m-1}\cos(2z)+q^{4m-2}).
\end{split} 
\end{equation}
The modified Weierstrass functions $\zeta_1(z)$ and $\wp_1(z)$ above are related to the theta function $\tet_1(z,q)$  as follows,   
\begin{equation} 
\label{zeta1fromtet1} 
\zeta_1(z) =  \partial_z\log\theta_1(\kappa z,q)
\end{equation}  
and 
\begin{equation} 
\label{wp1fromtet1} 
\wp_1(z) =  -\partial_z^2\log\theta_1(\kappa z,q); 
\end{equation} 
to see this, use the following relation of $\tet_1(z,q)$ to the Weierstrass $\sigma$-function \cite[Eq.~23.6.9]{DLMF} 
\begin{equation}
\sigma(z) =  2\omega_1 \exp\bigg( \frac{\eta_1}{2\omega_1} z^2 \bigg) \frac{\theta_1(\kappa z,q)}{\pi\theta_1'(0,q)}
\end{equation}
together with the well-known relations $\zeta(z)=\partial_z\log\sigma(z)$ and $\wp(z)=-\partial_z\zeta(z)$. 

It is important to note that $\zeta_1(z)$ is $2\omega_1$-periodic (but not $2\omega_2$-periodic). We also note that the identity (see \cite[Eq.~20.2.14]{DLMF})
\begin{equation}
\label{tet1fromtet4} 
\theta_1(z+\pi\tau/2,q)=\ii \ee^{-\frac{\ii \pi \tau}{4}}\ee^{-\ii z} \theta_4(z,q)
\end{equation}
implies that
\begin{equation} 
\label{wp1fromtet4} 
\wp_1(z+\ii\delta) = -\partial_z^2 \log\theta_4(\kappa z,q).
\end{equation}

\subsection{Integral transforms}
We derive some identities related to the integral transforms $T$ and $\tilde{T}$ in \eqref{TT}. As in the main text, we set $(\omega_1,\omega_2)=(\ell,\ii\delta)$. 

The function $\zeta_1(z)$ \eqref{zeta1} has the following series representation \cite[Eq. 23.8.2]{DLMF}
\begin{equation}\label{zeta1series} 
\zeta_1(z)=\kappa\cot(\kappa z)+4\kappa\sum_{n=1}^{\infty} \frac{q^{2n}}{1-q^{2n}}\sin(2n\kappa z),\qquad (|\im(z)|< 2\delta, z\notin \Lambda), 
\end{equation}
using $\kappa=\pi/2\ell$. This result is established in \cite{lawden2013} by computing the Fourier series for $\partial_z\log \theta_1(z,q)$. Similarly, we can use the series \cite[Eq.~20.5.13]{DLMF}
\begin{equation}
\frac{\theta_4'(z;q)}{\theta_4(z;q)}=4\sum_{n=1}^{\infty} \frac{q^n}{1-q^{2n}}\sin(2nz)\qquad (|\im(z)|< \pi\,\im(\tau/2)) 
\end{equation}
and the identities \eqref{wp1fromtet1} and \eqref{tet1fromtet4}
to establish 
\begin{equation}\label{zeta1series2} 
\zeta_1(z+\ii\delta)=-\ii\kappa +4\kappa\sum\limits_{n=1}^{\infty} \frac{q^n}{1-q^{2n}}\sin(2n\kappa z)\qquad (|\im (z)|< \delta).
\end{equation}
In the limit $\delta\to \infty$, we have $q\to 0$ and hence \eqref{zeta1series} and \eqref{zeta1series2} imply that the integral transforms $T$ and $\tilde{T}$, defined in \eqref{TT}, satisfy
\begin{equation}\label{Tdeltalimit}
\lim_{\delta\to\infty} (Tf)(x)= \frac{1}{2\ell} \fpint{-\ell}{\ell} \cot(\kappa(x'-x)) f(x')\,\mathrm{d}x' 
\end{equation}
and
\begin{equation}\label{Ttdeltalimit}
\lim_{\delta\to\infty} (\tilde{T}f)(x)= -\frac{\ii}{2\ell}\int_{-\ell}^{\ell} f(x')\,\mathrm{d}x', 
\end{equation}
using $\kappa=\pi/2\ell$; in particular, $\tilde{T}\to 0$ on zero-mean functions. It is interesting to note that the right-hand side in \eqref{Tdeltalimit} is the Hilbert transform $(Hf)(x)$ of the $2\ell$-periodic  function $f(x)$.  

To conclude, we state and prove a result we need that relates the integral operators $T$, $\tilde{T}$ to the coefficients $c_n$, $s_n$ appearing in the Bogoliubov transformation \eqref{BT}. 

\begin{lemma}\label{lem:TT}
The action of the integral operators $T$ and $\tilde{T}$ in \eqref{TT} on derivatives $f_x(x)\coloneqq \partial_x f(x)$ of $2\ell$-periodic $C^3$-functions $f(x)$ of $x\in\R$ can be computed as 
\begin{equation} \label{TTseries2} 
\begin{split} 
(Tf_x)(x) = & -\sum_{n\in\Z_{\neq 0}} 2\kappa|n|(c_n^2+s_n^2)\frac1{2\ell}\int_{-\ell}^\ell \ee^{2\ii n\kappa (x-x')}f(x')\,\dd{x'}, \\
(\tilde{T}f_x)(x)= &-\sum_{n\in\Z_{\neq 0}} 4\kappa|n|c_ns_n\frac1{2\ell}\int_{-\ell}^\ell \ee^{2\ii n\kappa (x-x')}f(x')\,\dd{x'} , 
\end{split} 
\end{equation} 
with $c_n$ and $s_n$ in \eqref{cnsn}.
\end{lemma}

\begin{proof} 
We compute, for $x\in\R$ and $\eps>0$, 
\begin{equation}
\cot(\kappa(x\mp \ii \eps)) = \ii \frac{\ee^{\ii\kappa x\pm \kappa \eps} + \ee^{-\ii\kappa x\mp \kappa \eps} }{\ee^{\ii\kappa x\pm \kappa \eps} - \ee^{-\ii\kappa x\mp \kappa \eps}}
= \pm \ii \frac{1+ \ee^{\mp 2\ii\kappa x-2\kappa \eps} }{1 - \ee^{\mp 2\ii\kappa x-2\kappa \eps}} = \pm \ii\Big(1 + 2\sum_{n=1}^\infty  \ee^{\mp 2\ii n\kappa x-2\kappa n\eps} \Big) , 
\end{equation} 
and thus 
\begin{equation}\label{cot} 
\cot(\kappa(x-\ii \eps)) + \cot(\kappa(x+ \ii \eps)) = 4\sum_{n=1}^\infty \sin(2n\kappa x)\ee^{-2\kappa n \eps}\quad (x\in\R, \eps>0). 
\end{equation} 
Since $\zeta_1(z)$ has a simple pole with residue $1$ at $z = 0$, the Cauchy principal value integral in the first definition in  \eqref{TT} is equivalent to
\begin{equation} 
(Tf)(x) = \lim_{\eps\to 0^+} \frac1{2\pi}\int_{-\ell}^\ell \big(\zeta_1(x'-x+\ii\eps) + \zeta_1(x'-x - \ii\eps) \big) f(x')\,\dd{x'}. 
\end{equation} 
Using \eqref{zeta1series} and \eqref{cot}, as well as the identity
$$\sin(2n\kappa (x+\ii\eps)) +  \sin(2n\kappa (x-\ii\eps)) = 2\cosh(2n\kappa \eps) \sin(2n\kappa x),$$  
we obtain 
\begin{equation} \label{zeta1pluszeta1}
\zeta_1(x+\ii\eps) + \zeta_1(x-\ii\eps) = 4\kappa\sum_{n=1}^\infty \bigg( \ee^{-2\kappa n\eps} \\
+  \frac{2q^{2n}}{1-q^{2n}} \cosh(2n\kappa \eps) \bigg) \sin(2n\kappa x).
\end{equation} 
Since $0 < q < 1$, the series in (\ref{zeta1pluszeta1}) converges absolutely and uniformly for $x \in \mathbb{R}$ for all sufficiently small $\epsilon > 0$. 
Using that $4\kappa/2\pi=1/\ell$, we obtain, for $x \in \mathbb{R}$,
\begin{align}\label{Tfxlimeps}
(Tf_x)(x) = &\; \lim_{\eps\to 0^+} \sum_{n=1}^\infty\bigg( \ee^{-2\kappa n\eps} 
 +  \frac{2q^{2n}}{1-q^{2n}}\cosh(2n\kappa \eps) \bigg) \frac{1}{\ell} \int_{-\ell}^\ell   \sin(2n\kappa (x'-x))
 f_x(x')\,\dd{x'},
\end{align} 
where we have used the uniform convergence of the series to interchange summation and integration.
By assumption, $f_x$ is $2\ell$-periodic and $C^2$ on $\R$. Thus, integrating by parts twice, we find
$$\int_{-\ell}^\ell   \sin(2n\kappa (x'-x))
 f_x(x')\,\dd{x'} = -\int_{-\ell}^\ell  \frac{\sin(2n\kappa (x'-x))}{(2n\kappa)^2}
 f_{xxx}(x')\,\dd{x'} = O(n^{-2})$$
 uniformly for $x \in \R$. In particular, for each $x \in \mathbb{R}$, the terms in the series in (\ref{Tfxlimeps}) are bounded by $C n^{-2}$ uniformly for all sufficiently small $\epsilon > 0$. We can therefore take the limit inside the sum in (\ref{Tfxlimeps}). Noting that $1+2q^{2n}/(1-q^{2n})=c_n^2+s_n^2$ for $n \geq 1$, this yields
\begin{equation} \label{TTseries} 
(Tf_x)(x) = \sum_{n=1}^{\infty} (c_n^2+s_n^2)\frac1{\ell}\int_{-\ell}^\ell \sin(2n\kappa (x'-x))f_x(x')\,\dd{x'}.
\end{equation} 
Computing  
\begin{multline}\label{fx} 
\frac1{\ell}\int_{-\ell}^\ell \sin(2n\kappa (x'-x))f_{x}(x')\dd{x'} = -2n\kappa \frac1{\ell}\int_{-\ell}^\ell \cos(2n\kappa (x'-x))f(x')\,\dd{x'} \\
=  -2n\kappa \frac1{2\ell}\int_{-\ell}^\ell \big(\ee^{-2\ii n\kappa (x-x')}  + \ee^{2\ii n\kappa (x-x')}  \bigr)f(x')\,\dd{x'} 
\end{multline} 
by partial integration and simplifying, we obtain the first identity in \eqref{TTseries2}. 

Let us prove the second identity in \eqref{TTseries2}. 
Using  \eqref{zeta1series2} and the second definition in \eqref{TT}, we find 
\begin{equation} 
(\tilde{T}f_x)(x) = \frac1{\pi}\int_{-\ell}^\ell \Big( -\ii\kappa +4\kappa\sum\limits_{n=1}^{\infty} \frac{q^n}{1-q^{2n}}\sin(2n\kappa (x'-x))\Big) f_x(x')\,\dd{x'} . 
\end{equation} 
Since $q^n/(1-q^{2n})=c_ns_{n}$ for $n \geq 1$ and $4\kappa/\pi=2/\ell$, we obtain 
\begin{equation}
(\tilde{T}f_x)(x)=-\frac{\ii}{2\ell}\int_{-\ell}^{\ell} f_x(x')\,\mathrm{d}x' + \sum_{n=1}^{\infty} 2c_ns_n\frac1{\ell}\int_{-\ell}^\ell \sin(2n\kappa (x'-x))f_x(x')\,\dd{x'}
\end{equation} 
where the interchange of summation and integration is justified because $0 < q < 1$. The first term on the right-hand side vanishes due to the $2\ell$-periodicity of $f$, so by inserting \eqref{fx} we obtain the result. 
\end{proof} 

\section{Integrability of the generalized eCS model}
\label{app:int}
There are several complementary arguments that strongly suggest the deformed eCS model is quantum integrable.
In Appendix~\ref{app:deCS}, we discuss these arguments. In Appendix~\ref{app:eCSgen},  we show that the generalized eCS model is quantum integrable if and only if the deformed eCS model defined in \eqref{deCS} is. 

\subsection{Deformed eCS model} 
\label{app:deCS} 
The eCS model has a relativistic quantum integrable generalization known as the elliptic Ruijsenaars model \cite{ruijsenaars1987}. 
There exists a deformed elliptic Ruijsenaars model which is quantum integrable as well \cite{hallnas2022}. 
Since the deformed eCS model is a limiting case of the deformed elliptic Ruijsenaars model, the results in Ref.~\cite{hallnas2022} strongly suggest that the deformed eCS model is quantum integrable. 
However, the limit from the Ruijsenaars models to CS models is delicate and, for this reason, it requires non-trivial work to deduce from the quantum integrability of the former the quantum integrability of the latter (see \cite[Section~4.3]{ruijsenaars1999}); to our knowledge, this work has not been done for the deformed Ruijsenaars model. 

Below we mention three further reasons to believe that the deformed eCS model is quantum integrable. 
\begin{enumerate} 
\item The quantum integrability of the deformed eCS model in the special case $M=1$, for general $N\in\Z_{\geq 1}$, was proved in Ref.~\cite{khodarinova2005}. 
\item An argument for the quantum integrability of the eCS model for general $N,M$ was recently given in \cite{chen2020}. 
\item The results in the present paper strongly suggest that the generalized eCS model in \eqref{eCSgen} is quantum integrable in the same strong sense as the eCS model. 
\end{enumerate} 

\subsection{Generalized eCS model} 
\label{app:eCSgen}
We show that the quantum integrability of the deformed eCS model in \eqref{deCS} implies the quantum integrability of the generalized eCS model in \eqref{eCSgen} (this is a generalization of a well-known argument going back to Calogero). 

We start with the deformed eCS Hamiltonian in \eqref{deCS} and analytically continue a group of the variables as follows,
\begin{equation} 
\label{substitution}
x_{j-N_1} = \tilde{x}_{j}+\ii\delta\quad (j-N_1=1,\ldots,M_1),\quad y_{j-N_2} = \tilde{y}_{j}+\ii\delta\quad (j-N_2=1,\ldots,M_2); 
\end{equation} 
one can check that this substitution gives the generalized eCS Hamiltonian \eqref{eCSgen}.

Quantum integrability of the deformed eCS model means that there is a sufficiently large family of commuting differential operators, including the one in \eqref{deCS}. Clearly, we can make the substitution \eqref{substitution} in all these differential operators and, by that, get a family of commuting differential operators for the generalized eCS model. Thus, the quantum integrability of the deformed eCS model implies the quantum integrability of the generalized eCS model.

\section{Proofs}
\label{app:proofs} 
We give proofs of certain results stated in the main text. 

\subsection{Commutator relations}\label{app:proofssec5}
We prove the commutator relations in \eqref{W2rrhor}--\eqref{W3rrhor}. 

We start by two implications of Definition~\ref{def:circ} and \eqref{fn}--\eqref{circ1}:  For $f(x)$ such that $f(x;\eps)\coloneqq (\delta_\eps*f)(x)$ is well-defined as a sesquilinear form on $\cD$, we have 
\begin{equation}\label{epsepsp}
\lim_{\eps'\to 0^+}\int_{-\ell}^\ell \delta(x-x';\eps+\eps')f(x';\eps')\,\dd{x}' = f(x;\eps) \quad (\eps>0) 
\end{equation} 
and 
\begin{equation}\label{epsepsp2}
\lim_{\eps'\to 0^+}\int_{-\ell}^\ell \delta(x-x';\eps+\eps')f(x';\eps')^2\,\dd{x}' = (f\circ f)(x;\eps) \quad (\eps>0). 
\end{equation}

To prove  \eqref{W2rrhor} we compute, using that $\xxa\rho_{r'}(x';\eps')^2\xxe$ and $\rho_{r'}(x';\eps')^2$ differ only by a well-defined $\C$-number function,  
\begin{multline}\label{rhor2rhorp}
[\rho_r(x;\eps),\xxa\rho_{r'}(x';\eps')^2\xxe] = [\rho_r(x;\eps),\rho_{r'}(x';\eps')^2]  \\
=  \rho_{r'}(x';\eps')[\rho_r(x;\eps),\rho_{r'}(x';\eps')] + [\rho_r(x;\eps),\rho_{r'}(x';\eps')]\rho_{r'}(x';\eps')\\
= -4\pi\ii r\delta_{r,r'}\partial_x\delta(x-x';\eps+\eps')  \rho_{r'}(x';\eps'), 
\end{multline} 
inserting \eqref{eq:CCR} in the last step. 
Thus we can compute, using \eqref{W2pm2} and \eqref{epsepsp},  
\begin{multline} 
\ii[\rho_r(x;\eps),W^2_{r'}] = \lim_{\eps'\to 0^+} \frac{\ii}{4\pi} \int_{-\ell}^{\ell} [\rho_r(x;\eps),\xxa\rho_{r'}(x';\eps')^2\xxe]\,\dd{x}' \\
= \lim_{\eps'\to 0^+} \frac{\ii}{4\pi} \int_{-\ell}^{\ell} \big( -4\pi\ii r\delta_{r,r'}\partial_x\delta(x-x';\eps+\eps')  \rho_{r'}(x';\eps') \big)\,\dd{x}'
\\ = r\delta_{r,r'}  \partial_x \lim_{\eps'\to 0^+} \int_{-\ell}^{\ell}\delta(x'-x;\eps+\eps')\rho_r(x';\eps')\,\dd{x}' = r\delta_{r,r'}\partial_x  \rho_r(x;\eps), 
\end{multline} 
which is equivalent to \eqref{W2rrhor}. 

To prove  \eqref{W3rrhor} we note that, by Definition~\ref{def:normalorder}, \eqref{BT}, and \eqref{rhor}, 
\begin{equation} 
\xxa\rho_r(x;\eps)^3\xxe \, =\,  \xxa\rho_r(x;\eps)^2\xxe\tilde\rho_r^-(x;\eps) + \tilde\rho_r^+(x;\eps)\xxa\rho_r(x;\eps)^2\xxe 
\end{equation} 
with $\tilde\rho_r^\pm(x;\eps)\coloneqq \kappa Q_r+\rho_r^\pm(x;\eps)$ where 
\begin{equation} \label{rhorplusminus}
\rho_r^\pm(x;\eps) \coloneqq \sum\limits_{n=1}^\infty 2\kappa \Big(\ee^{2\kappa(\mp \ii r nx-|n|\epsilon)} c_n a_{r,\mp n} - \ee^{2\kappa(\pm \ii r nx-|n|\epsilon)} s_n a_{-r ,\mp n} \Big)  
\end{equation} 
are the creation and annihilation parts of $\rho_r(x;\eps)$ for $+$ and $-$, respectively. Thus, 
\begin{multline} 
[\rho_r(x;\eps),\xxa\rho_{r'}(x';\eps')^3\xxe] 
= [\rho_r(x;\eps),\xxa\rho_{r'}(x';\eps')^2\xxe]\tilde\rho_{r'}^-(x';\eps') +\tilde\rho_{r'}^+(x';\eps')[\rho_r(x;\eps),\xxa\rho_{r'}(x';\eps')^2\xxe] \\ 
+ \xxa\rho_{r'}(x';\eps')^2\xxe [\rho_r(x;\eps),\tilde\rho_{r'}^-(x';\eps')] + [\rho_r(x;\eps),\tilde\rho_{r'}^+(x';\eps')] \xxa\rho_{r'}(x';\eps')^2\xxe \\
= [\rho_r(x;\eps),\xxa\rho_{r'}(x';\eps')^2\xxe]\tilde\rho_{r'}^-(x';\eps') +\tilde\rho_{r'}^+(x';\eps')[\rho_r(x;\eps),\xxa\rho_{r'}(x';\eps')^2\xxe] \\ 
+ [\rho_r(x;\eps),\rho_{r'}(x';\eps')] \xxa\rho_{r'}(x';\eps')^2\xxe, 
\end{multline} 
since  $[\rho_r(x;\eps),\tilde\rho_{r'}^\pm(y;\eps')]$ are well-defined $\C$-number functions and $\tilde\rho_{r'}^{+}(y;\eps')+\tilde\rho_{r'}^-(y;\eps')=\rho_{r'}(y;\eps)$. 
Inserting \eqref{eq:CCR}  and \eqref{rhor2rhorp} we get, using again Definition~\ref{def:normalorder}, 
\begin{multline} 
[\rho_{r}(x;\eps),\xxa\rho_{r'}(x';\eps')^3\xxe] = 
-4\pi\ii r\delta_{r,r'}\partial_x\delta(x-x';\eps+\eps')  \rho_{r'}(x';\eps')\tilde\rho_{r'}^-(x';\eps')  \\ 
-4\pi\ii r\delta_{r,r'}\partial_x\delta(x-x';\eps+\eps')  \tilde\rho_{r'}^+(x';\eps') \rho_{r'}(x';\eps')\\
- 2\pi\ii r\delta_{r,r'}\partial_x\delta(x-x';\eps+\eps') \xxa\rho_{r'}(x';\eps')^2\xxe \\
= - 6\pi\ii r\delta_{r,r'}\partial_x\delta(x-x';\eps+\eps') \xxa\rho_r(x';\eps')^2\xxe .
\end{multline} 
Thus, similarly as above, using \eqref{W3pm2} and \eqref{epsepsp2},  
\begin{multline} 
\ii[\rho_{r}(x;\eps),W^3_{r'}] = \lim_{\eps'\to 0^+} \frac{\ii}{6\pi} \int_{-\ell}^{\ell} [\rho_r(x;\eps),\xxa\rho_{r'}(x';\eps')^3\xxe]\,\dd{x}' \\
= \lim_{\eps'\to 0^+} \frac{\ii}{6\pi} \int_{-\ell}^{\ell} \big( -6\pi\ii r\delta_{r,r'}\partial_x\delta(x-x';\eps+\eps') \xxa \rho_r(x';\eps')^2\xxe \big)\,\dd{x}'
\\ = r\delta_{r,r'} \partial_x\lim_{\eps'\to 0^+} \int_{-\ell}^{\ell}\delta(x-y;\eps+\eps')\xxa \rho_r(x';\eps')^2\xxe\,\dd{x}' = r\delta_{r,r'}\partial_x \xxa(\rho_r\circ\rho_r)(x;\eps)\xxe, 
\end{multline} 
which is equivalent to \eqref{W3rrhor} since $\xxa\rho_r\circ\rho_{r,x}\xxe \; =\; \xxa\rho_{r,x}\circ\rho_r\xxe$; the latter is shown by the following computation, 
\begin{multline} 
\xxa(\rho_{r,x}\circ\rho_{r})(x;\eps)\xxe \, = \lim_{\eps\to 0^+}\int_{-\ell}^{\ell}\delta(x -x';\eps+\eps')\xxa \rho_{r,x'}(x';\eps')\rho_r(x';\eps')\xxe \dd{x}' \\
=  \lim_{\eps\to 0^+}\int_{-\ell}^{\ell}\delta(x-x';\eps+\eps')\xxa \rho_{r}(x';\eps')\rho_{r,x'}(x';\eps')\xxe\,\dd{x}' = \; \xxa(\rho_r\circ\rho_{r,x})(x;\eps)\xxe. 
\end{multline}

\subsection{Proof of Lemma \ref{lemma2}} \label{lemma2app}
We prove \eqref{W1commutators}--\eqref{W3commutators} in a unified way: by deriving commutation relations between $\ee^{t \rho_{r'}(x';\epsilon')}$ and $\phi_{r,\nu}(x;\epsilon)$, we can compute the commutator of $W_{k,r'}$ with the anyon $\phi_{r,\nu}(x;\epsilon)$ for any value of $k$. 

From (\ref{rhorplusminus}) and \eqref{Krpm}, we observe that
\begin{equation}\label{Kprimerho}
r (K_r^{\pm})'(x;\epsilon)=\rho_r^{\pm}(x;\epsilon),
\end{equation}
where here and below primes denotes differentiation with respect to the first argument of the function. Equations \eqref{anyondefinition}, \eqref{rhor}, and \eqref{Kprimerho} imply the identities
\begin{align}\label{Kprojectionderivatives}
&(\phi_{r,\nu})'(x;\epsilon)=-\ii\nu r \xxa \rho_r(x;\epsilon)\phi_{r,\nu}(x;\epsilon)\xxe
\end{align}
and
\begin{equation}\label{Kprojectionderivatives2}
(\phi_{r,\nu})''(x;\epsilon)=-\ii\nu r \xxa\rho_r'(x;\epsilon)\phi_{r,\nu}(x;\epsilon)\xxe-\nu^2\xxa\rho_r(x;\epsilon)^2\phi_{r,\nu}(x;\epsilon)\xxe, 
\end{equation}
which we will need later on. Next we compute, by differentiating \eqref{KpmtoCCt},
\begin{equation}\label{rhoKcommutator}
 r[\rho_{r}^{\mp}(x;\epsilon),K_{r'}^{\pm}(x';\epsilon')]=\begin{cases}
rC'(\pm r(x-x');\epsilon+\epsilon') & r=r' \\
\tilde{C}'(\pm(x-x');\epsilon+\epsilon') & r=-r'.	
\end{cases}	
\end{equation}

From \eqref{CCt}, we have
\begin{equation}
\begin{split}
C'(z;\epsilon) = &\; 2\kappa\ii \sum_{n=1}^{\infty} \big(  \ee^{2\kappa(\ii n z-|n|\epsilon)}c_n^2 - \ee^{2\kappa(-\ii n z-|n|\epsilon)}s_n^2 \big)
	\\
= &\; \frac{2\pi \ii}{2\ell} \sum_{n=1}^{\infty} \big( \ee^{2\kappa \ii n z} + 2i\sin(2\kappa n z)s_n^2 \big) \ee^{-2\kappa|n|\epsilon}
= 2\pi \ii \delta_-(z;\eps) - 4 \pi j(z;\eps) , 
	\\
\tilde{C}'(z;\epsilon) = &\; 2\kappa\ii \sum_{n=1}^{\infty} c_ns_n \big(\ee^{2\kappa(\ii n z-|n|\epsilon)}- \ee^{2\kappa(-\ii n z-|n|\epsilon)} \big)
	\\
= & - \frac{4\pi}{2\ell} \sum_{n=1}^{\infty} c_ns_n \sin(2\kappa n z) \ee^{-2\kappa|n|\epsilon}
= - 4 \pi \tilde{j}(z;\eps),
\end{split}
\end{equation}
where
\begin{equation}\label{deltapm} 
\delta_\pm(x;\eps)\coloneqq \frac1{2\ell}\sum_{n=1}^\infty \ee^{2\kappa(\mp \ii nx-|n|\eps)}\quad (x\in\R,\eps>0) , 
\end{equation} 
\begin{equation}\label{j1j4} 
\begin{split} 
j(x;\eps) \coloneqq & \frac1{2\ell} \sum_{n=1}^\infty s_n^2 \sin(2\kappa nx)\ee^{-2\kappa n\eps} 
	\\
\tilde{j}(x;\eps) \coloneqq &\frac1{2\ell} \sum_{n=1}^\infty c_ns_n \sin(2\kappa nx)\ee^{-2\kappa n\eps}
\end{split} \quad (x\in\R,\eps>0).
\end{equation} 
Hence,
\begin{equation}\label{rhosigmaKpmcommutators}
 r[\rho_{r}^{\mp}(x;\epsilon),K_{r'}^{\pm}(x';\epsilon')]=\begin{cases}
r(2\pi \ii \delta_-(\pm r(x-x');\epsilon+\epsilon') - 4 \pi j(\pm r(x-x');\epsilon+\epsilon')) & r=r' 
	\\
- 4 \pi \tilde{j}(\pm(x-x');\epsilon+\epsilon') & r=-r',	
\end{cases}	
\end{equation}


Using \eqref{phiexpanded}, the definition of normal ordering, the Baker-Campbell-Hausdorff formula \eqref{BCH}, and the commutators \eqref{rhosigmaKpmcommutators}, we obtain
\begin{align*}
\phi_{r,\nu}(x;\epsilon) \xxa \ee^{t\rho_{r'}(x';{\epsilon'})} \xxe
= &\; \ee^{-\ii r\nu \kappa Q_r x} R_r^{\nu/\nu_0} \ee^{-\ii r\nu \kappa Q_r x} 
\ee^{-\ii \nu K_r^+(x;\epsilon)} \ee^{-\ii \nu K_r^-(x;\epsilon)}
 \ee^{2t\kappa Q_{r'}} e^{t\rho_{r'}^+(x';{\epsilon'})} e^{t\rho_{r'}^-(x';{\epsilon'})} 
	\\
= &\;  \ee^{-t\kappa  \delta_{r,r'}  \nu}
\ee^{[-\ii \nu K_r^-(x;\epsilon), t\rho_{r'}^+(x';{\epsilon'})]}
\xxa \phi_{r,\nu}(x;\epsilon)  \ee^{t\rho_{r'}(x';{\epsilon'})} \xxe
	\\
= &\;  
\ee^{-\nu t (\kappa \delta_{r,r'} + \ii [K_r^-(x;\epsilon), \rho_{r'}^+(x';{\epsilon'})])}
\xxa \phi_{r,\nu}(x;\epsilon)  \ee^{t\rho_{r'}(x';{\epsilon'})} \xxe
	\\
= &\;  
\ee^{-\nu t \delta_{r,r'}(\kappa  +  2\pi  \delta_-(r(x-x');\epsilon+\epsilon') + 4 \pi \ii j(r(x-x');\epsilon+\epsilon'))}
\ee^{-\nu t \delta_{r,-r'}( -4 r \pi \ii \tilde{j}(x-x';\epsilon+\epsilon'))} \\
&\times 
\xxa \phi_{r,\nu}(x;\epsilon)  \ee^{t\rho_{r'}(x';{\epsilon'})} \xxe.
\end{align*}
and
\begin{align*}
\xxa \ee^{t\rho_{r'}(x';{\epsilon'})} \xxe \phi_{r,\nu}(x;\epsilon)
= &\;  \ee^{2t\kappa Q_{r'}} \ee^{t\rho_{r'}^+(x';{\epsilon'})} e^{t\rho_{r'}^-(x';{\epsilon'})} 
\ee^{-\ii r\nu \kappa Q_r x} R_r^{\nu/\nu_0} \ee^{-\ii r\nu \kappa Q_r x} 
\ee^{-\ii \nu K_r^+(x;\epsilon)} \ee^{-\ii \nu K_r^-(x;\epsilon)}
	\\
= &\;  e^{t\kappa  \delta_{r,r'}  \nu}
\ee^{[ t\rho_{r'}^-(x';{\epsilon'}), -\ii \nu K_r^+(x;\epsilon)]}
\xxa \phi_{r,\nu}(x;\epsilon)  \ee^{t\rho_{r'}(x';{\epsilon'})} \xxe
	\\
= &\;  
\ee^{\nu t (\kappa \delta_{r,r'} - \ii [\rho_{r'}^-(x';{\epsilon'}), K_r^+(x;\epsilon)])}
\xxa \phi_{r,\nu}(x;\epsilon)  \ee^{t\rho_{r'}(x';{\epsilon'})} \xxe
	\\
= &\;  
\ee^{\nu t \delta_{r,r'} (\kappa  + 2\pi \delta_+(r(x-x');\epsilon+\epsilon') - 4 \pi \ii j(r(x-x');\epsilon+\epsilon') 
)}
\ee^{\nu t \delta_{r,-r'} ( 4 r \pi \ii \tilde{j}(x-x';\epsilon+\epsilon'))}
\\
&\; \times \xxa \phi_{r,\nu}(x;\epsilon)  \ee^{t\rho_{r'}(x';{\epsilon'})} \xxe.
\end{align*}

We define the functions
\begin{equation}\label{Deltapm} 
\Delta_\pm(x;\eps)\coloneqq 
\kappa + 2\pi \delta_\pm(x;\eps)\mp 4\pi \ii j(x;\eps) \quad (x\in\R,\eps>0),
\end{equation} 
which provide decompositions of the regularized Dirac delta in \eqref{deltaeps} as follows, 
\begin{equation} 
\delta(x;\eps) = \frac1{2\ell} + \delta_+(x;\eps)+\delta_-(x;\eps)
= \frac{\Delta_+(x;\eps)+\Delta_-(x;\eps)}{2\pi} \quad (x\in\R,\eps>0). 
\end{equation} 
Then we can write the above as
\begin{equation}\label{anyonexpexchange}
\begin{split}
&\phi_{r,\nu}(x;\epsilon) \xxa \ee^{t\rho_{r'}(x';{\epsilon'})}\xxe \; =
\begin{cases}
\ee^{-\nu t\Delta_-(r(x-x');\epsilon+\epsilon')}\xxa\phi_{r,\nu}(x;\epsilon)\ee^{t\rho_{r'}(x';{\epsilon'})}\xxe, & r = r', \\
 \ee^{4\pi \ii\nu t \tilde{j}(r(x-x');\epsilon+\epsilon')}\xxa\phi_{r,\nu}(x;\epsilon)\ee^{t\rho_{r'}(x';{\epsilon'})}\xxe, & r = -r',
 \end{cases}
	\\
& \xxa \ee^{t\rho_{r'}(x';{\epsilon'})}\xxe \phi_{r,\nu}(x;\epsilon)
= \begin{cases}
\ee^{\nu t\Delta_+(r(x-x');\epsilon+\epsilon')} \xxa\phi_{r,\nu}(x;\epsilon)\ee^{t\rho_{r'}(x';{\epsilon'})}\xxe, & r=r',	\\
\ee^{4\pi \ii\nu t \tilde{j}(r(x+x');\epsilon+\epsilon')} \xxa\phi_{r,\nu}(x;\epsilon)\ee^{t\rho_{r'}(x';{\epsilon'})}\xxe& r = -r'. 
\end{cases}
\end{split}
\end{equation}
The commutators
\begin{equation}
\label{exp_commutators}
[\phi_{r,\nu}(x;\epsilon),\xxa \ee^{t\rho_{r'}(x';\epsilon')}\xxe]= \delta_{r,r'} \big(\ee^{-\nu t\Delta_-(r(x-x');\epsilon+\epsilon')}-\ee^{\nu t\Delta_+(r(x-x');\epsilon+\epsilon')}\big)\xxa\phi_{r, \nu}(x;\epsilon)\ee^{t\rho_{r'}(x';{\epsilon'})}\xxe
\end{equation}
follow immediately from \eqref{anyonexpexchange}, and from here  the commutators of the anyons with $W^k$, $\tilde{W}^k$ will follow by differentiation. We start with the case $k=1$; we differentiate \eqref{exp_commutators} with respect to $t$ and set $t=0$ to obtain
\begin{align}\nonumber
[\phi_{r,\nu}(x;\epsilon),\rho_{r'}(x';{\epsilon'})] 
& = \delta_{r,r'} \big(-\nu \Delta_-(r(x-x');\epsilon+\epsilon') - \nu \Delta_+(r(x-x');\epsilon+\epsilon')\big)\xxa\phi_{r, \nu}(x;\epsilon) \xxe
	\\\label{k1_commutators}
& = -\delta_{r,r'} 2\pi\nu\delta(r(x-x');\epsilon+\epsilon')\phi_{r, \nu}(x;\epsilon),
\end{align}
where we have used that
\begin{align}\label{Deltasum}
\Delta_-(z;\epsilon)+\Delta_+(z;\epsilon)= 2\pi\delta(z;\epsilon). 
\end{align}
Integrating \eqref{k1_commutators} with respect to $x'$ from $-\ell$ to $\ell$ and using
\begin{equation}\label{eq:deltaint}
\int_{-\ell}^{\ell} \delta(r(x-x');\epsilon)\,\mathrm{d}x'=1,
\end{equation}
we obtain \eqref{W1commutators} in the limit $\epsilon'\to 0^+$.

We now consider the case $k=2$. We differentiate \eqref{exp_commutators} twice with respect to $t$ before setting $t=0$. Employing (\ref{Deltasum}) and the identity
\begin{align}\label{Deltaminus2Deltaplus2}
\Delta_-(z;\epsilon)^2-\Delta_+(z;\epsilon)^2= (4\pi)^2\ii\delta(z;\epsilon)j(z;\epsilon)-2\pi\ii\delta'(z;\epsilon),
\end{align}
we obtain
\begin{equation}
\label{k2_commutators}
\begin{split}
&[\phi_{r,\nu}(x;\epsilon),\xxa\rho_{r'}(x';\epsilon')^2\xxe]= \\
&\delta_{r,r'} \nu^2\big(\Delta_-(r(x-x');\epsilon+\epsilon')^2 - \Delta_+(r(x-x');\epsilon+\epsilon')^2\big)\xxa\phi_{r, \nu}(x;\epsilon)\xxe
	\\
& -2 \delta_{r,r'} \nu \big(\Delta_-(r(x-x');\epsilon+\epsilon') + \Delta_+(r(x-x');\epsilon+\epsilon')\big)\xxa\phi_{r, \nu}(x;\epsilon)\rho_{r'}(x';{\epsilon'})\xxe \; =
	\\
	& \delta_{r,r'} \nu^2 \left[(4\pi)^2 \ii \delta(r(x-x');\epsilon+\epsilon')j(r(x-x');\epsilon+\epsilon')-2\pi \ii\delta'(r(x-x');\epsilon+\epsilon')\right]\phi_{r,\nu}(x;\epsilon)
	\\
& -4\delta_{r,r'}\pi\nu\delta(r(x-x');\epsilon+\epsilon')\xxa\rho_r(x';\epsilon')\phi_{r,\nu}(x;\epsilon)\xxe.
\end{split}
\end{equation}
Again integrating $x'$ between $-\ell$ and $\ell$, we use the integrals
\begin{equation}\label{integrals2}
\begin{split}
&\int_{-\ell}^{\ell} \delta(r(x-x');\epsilon+\epsilon')j(r(x-x');\epsilon+\epsilon')\,\mathrm{d}x'=0,
	\\
&\int_{-\ell}^{\ell} {\delta'(r(x-x');\epsilon +\epsilon')}\,\mathrm{d}x'=0, 
	\\
& \int_{-\ell}^{\ell} \delta(r(x-x');\epsilon+\epsilon')\rho_r(x';\epsilon')\,\mathrm{d}x'= \rho_r(x; \epsilon+2\epsilon'),
\end{split}
\end{equation}
to obtain 
\begin{equation}
\begin{split}
\lim_{\epsilon' \to 0} \frac{1}{4\pi} \int_{-\ell}^{\ell} [\phi_{r,\nu}(x;\epsilon), \xxa\rho_{r'}(x';\epsilon')^2\xxe] \,\mathrm{d}x'
&=-\delta_{r,r'} \nu \lim_{\epsilon' \to 0}  \xxa\rho_r(x;\epsilon + 2\epsilon')\phi_{r,\nu}(x;\epsilon)\xxe
\\
&=-\delta_{r,r'} \nu \xxa\rho_r(x;\epsilon)\phi_{r,\nu}(x;\epsilon)\xxe.
\end{split}
\end{equation}
Recalling (\ref{Kprojectionderivatives}), this becomes
$$[\phi_{r,\nu}(x;\epsilon), W_{2,r'}] =-\ii r\delta_{r,r'} (\phi_{r,\nu})'(x;\epsilon),$$
which is \eqref{W2commutators}. 

We lastly consider the case $k=3$. Differentiating \eqref{exp_commutators} three times with respect to $t$ and setting $t=0$, we obtain
\begin{multline}\label{X1X2X3}
[{\phi}_{r,\nu}(x),\xxa\rho_{r'}(x';\epsilon')^3\xxe]= 
- \nu^3 \delta_{r,r'} \big(\Delta_-(r(x-x');\epsilon+\epsilon')^3 + \Delta_+(r(x-x');\epsilon+\epsilon')^3\big)\xxa\phi_{r, \nu}(x;\epsilon) \xxe
	 \\
 + 3\nu^2\delta_{r,r'} \big(\Delta_-(r(x-x');\epsilon+\epsilon')^2 - \Delta_+(r(x-x');\epsilon+\epsilon')^2\big)\xxa\phi_{r, \nu}(x;\epsilon)\rho_{r'}(x';{\epsilon'})\xxe
	\\
 -3\nu \delta_{r,r'} \big(\Delta_-(r(x-x');\epsilon+\epsilon') + \Delta_+(r(x-x');\epsilon+\epsilon')\big)\xxa\phi_{r, \nu}(x;\epsilon)\rho_{r'}(x';{\epsilon'})^2\xxe
	\\
=:  X_1+X_2+ X_3.
\end{multline}
Applying the identity \cite{langmann2004}
\begin{align}
\Delta_-(z;\epsilon)^3+\Delta_+(z;\epsilon)^3
= -96\pi^3\delta(z;\epsilon)j(z;\epsilon)^2
+ 24\pi^2 \delta'(z;\epsilon)j(z;\epsilon)
- \pi \delta''(z;\epsilon)
+ 2\pi \kappa^2 \delta(z;\epsilon),
\end{align}
and (\ref{Deltaminus2Deltaplus2}) and (\ref{Deltasum}) in $X_1, X_2$, and $X_3$, respectively, we can write
\begin{equation}
\label{k3_commutators}
\begin{split}
X_1 = & - \nu^3 \delta_{r,r'} \Big(-96\pi^3\delta(r(x-x'); \epsilon + \epsilon')j(r(x-x'); \epsilon + \epsilon')^2
	\\
& + 24\pi^2 \delta'(r(x-x'); \epsilon + \epsilon')j(r(x-x'); \epsilon + \epsilon'),
	\\
&
 - \pi \delta''(r(x-x'); \epsilon + \epsilon')
+ 2\pi \kappa^2 \delta(r(x-x'); \epsilon + \epsilon')\Big) \phi_{r, \nu}(x;\epsilon) ,
	\\
X_2 = &\; 3\nu^2\delta_{r,r'} \big((4\pi)^2\ii\delta(r(x-x');\epsilon+\epsilon')j(r(x-x');\epsilon+\epsilon')
	\\
& -2\pi\ii\delta'(r(x-x');\epsilon+\epsilon')\big)\xxa\phi_{r, \nu}(x;\epsilon)\rho_{r}(x';{\epsilon'})\xxe
	\\
X_3 = & -3\nu \delta_{r,r'} 2\pi\delta(r(x-x');\epsilon+\epsilon')\xxa\phi_{r, \nu}(x;\epsilon)\rho_{r}(x';{\epsilon'})^2\xxe.
\end{split}
\end{equation}
Integrating $X_1$ with respect to $x'$ from $-\ell$ to $\ell$ using the integrals 
\begin{equation}
\begin{split}
&\begin{multlined} \int_{-\ell}^{\ell} \delta(r(x-x');\epsilon + \epsilon')j(r(x-x');\epsilon + \epsilon')^2\,\mathrm{d} {x'} \\
= \frac{1}{2(2\ell)^2}\sum_{n,m=1}^{\infty}s_n^2s_m^2\ee^{-2\kappa(n+m)(\epsilon+\epsilon')}\big( \ee^{-2\kappa |n-m|(\epsilon+\epsilon')}-\ee^{-2\kappa(n+m)(\epsilon+\epsilon')}  \big), 
\end{multlined}	\\
&\int_{-\ell}^{\ell}\delta'(r(x-x');\epsilon+\epsilon')j(r(x-x');\epsilon+\epsilon')\,\mathrm{d}x'
=-j'(0;2(\epsilon+\epsilon')), 
	\\
&\int_{-\ell}^{\ell} \delta''(r(x-x');\epsilon + \epsilon')\,\mathrm{d}x'=0, 
	\\
&\int_{-\ell}^{\ell} \delta(r(x-x');\epsilon)\,\mathrm{d}x'=1,
\end{split}
\end{equation}
we obtain
\begin{align*}
\lim_{\epsilon' \to 0}\int_{-\ell}^\ell X_1 \,\dd{x'} = & 
- \nu^3 \delta_{r,r'} \Big(-96\pi^3\frac{1}{2(2\ell)^2}\sum_{n,m=1}^{\infty}s_n^2s_m^2\ee^{-2\kappa(n+m)\epsilon}\big( \ee^{-2\kappa |n-m|\epsilon}-\ee^{-2\kappa(n+m)\epsilon}  \big)
	\\
& - 24\pi^2j'(0;2\epsilon)
+ 2\pi \kappa^2 \Big) \phi_{r, \nu}(x;\epsilon).
\end{align*}
Integrating $X_2$ with respect to $x'$ from $-\ell$ to $\ell$ using the integrals
\begin{equation}
\begin{split}
&\lim_{\epsilon' \to 0}\int_{-\ell}^{\ell}\delta(r(x-x');\epsilon+\epsilon')j(r(x-x');\epsilon+\epsilon')\rho_{r}(x';\epsilon')\,\mathrm{d}x' 
	\\
& = 
- \frac{\ii \pi}{(2\ell)^2}\sum_{n,m=1}^{\infty} s_n^2\left(\ee^{2\kappa \ii m r x} b_{r,m}-\ee^{-2\kappa \ii m r x} b_{r,-m}\right)\left(\ee^{-2\kappa |n-m|\epsilon}-\ee^{-2\kappa (n+m)\epsilon} \right)\ee^{-2\kappa n\epsilon}
\end{split}
\end{equation}
and
$$\int_{-\ell}^{\ell} \delta'(r(x-x');\epsilon+\epsilon') \rho_{r}(x';\epsilon') \,\mathrm{d}x' = r\rho_r'(x;\epsilon + 2\epsilon'),$$
we obtain
\begin{align*}
\lim_{\epsilon' \to 0}\int_{-\ell}^{\ell} X_2\, \dd{x'} = &\; 3\nu^2\delta_{r,r'} (4\pi)^2\ii 
\xxa \bigg(- \frac{\ii \pi}{(2\ell)^2}\sum_{n,m=1}^{\infty} s_n^2\left(\ee^{2\kappa \ii m r x} b_{r,m}-\ee^{-2\kappa \ii m r x} b_{r,-m}\right)
	\\
&\times \left(\ee^{-2\kappa |n-m|\epsilon}-\ee^{-2\kappa (n+m)\epsilon} \right)\ee^{-2\kappa n\epsilon}\bigg)\phi_{r, \nu}(x;\epsilon) \xxe
	\\
& - 3\nu^2\delta_{r,r'} 2\pi\ii  \xxa r \rho_r'(x;\epsilon) \phi_{r, \nu}(x;\epsilon) \xxe. 
\end{align*}
Integrating $X_3$ with respect to $x'$ from $-\ell$ to $\ell$ using the following identity which is a consequence of \eqref{epsepsp2}: 
\begin{equation}
\lim\limits_{\epsilon'\rightarrow 0^+}\int_{-\ell}^{\ell}\delta(r(x-x');\epsilon+\epsilon')\rho_r(x';\epsilon')^2\,\mathrm{d}x'
=(2\kappa)^2 \sum_{n,m\in\mathbb{Z}}\ee^{2\kappa \ii(n+m) r x - 2\kappa |n+m|\epsilon} \tilde{b}_{r,n} \tilde{b}_{r,m}
\end{equation}
with
\begin{equation}\label{btilde}
\tilde{b}_{r,n} := \begin{cases} b_{r,n}, & n \neq 0, \\ Q_r = \nu_0 b_{r,0}, & n = 0, \end{cases}
\end{equation}
we obtain
\begin{align*}
\lim_{\epsilon' \to 0}\int_{-\ell}^{\ell}  X_3\, \dd{x'} = & -3\nu \delta_{r,r'} 2\pi (2\kappa)^2 \xxa\bigg(\sum_{n,m\in\mathbb{Z}}\ee^{2\kappa \ii(n+m)r x - 2\kappa |n+m|\epsilon}   \tilde{b}_{r,n} \tilde{b}_{r,m}\bigg) \phi_{r, \nu}(x;\epsilon) \xxe.
\end{align*}
Combining the above and using the identity
\begin{equation*}
j'(0;2\epsilon) = \frac{2\kappa^2}{\pi} \sum_{n=1}^{\infty}n s_n^2 \ee^{-4\kappa n\epsilon},
\end{equation*}
we arrive at
\begin{align}\nonumber \label{W3rphicommutatorsimplified}
& [W_{3,r'},\phi_{r,\nu}(x;\epsilon)]
=  -\frac{1}{6\pi} \lim_{\epsilon' \to 0}\int_{-\ell}^{\ell} (X_1+X_2+ X_3)\, \dd{x'} = Y_1 + Y_2 + Y_3,
\end{align}
where
\begin{align*}
Y_1 \coloneqq &\; \kappa^2 \nu^3  \delta_{r,r'} \Bigg(-8 \sum_{n,m=1}^{\infty}s_n^2s_m^2\ee^{-2\kappa(n+m)\epsilon}\big( \ee^{-2\kappa |n-m|\epsilon}-\ee^{-2\kappa(n+m)\epsilon}  \big)
	\\
& - 8 \sum_{n=1}^{\infty}n s_n^2 \ee^{-4\kappa n\epsilon}
+ \frac{1}{3} \Bigg) \phi_{r, \nu}(x;\epsilon)
	\\\nonumber
Y_2 \coloneqq & - 8 \kappa^2\nu^2\delta_{r,r'}  
\xxa  \sum_{n,m=1}^{\infty} s_n^2\left(\ee^{2\kappa \ii m r x} b_{r,m}-\ee^{-2\kappa \ii m r x} b_{r,-m}\right)
	\\
&\times \left(\ee^{-2\kappa |n-m|\epsilon}-\ee^{-2\kappa (n+m)\epsilon} \right)\ee^{-2\kappa n\epsilon} \phi_{r, \nu}(x;\epsilon) \xxe
	\\
Y_3 \coloneqq & \; \ii r \nu^2\delta_{r,r'} \xxa \rho_r'(x;\epsilon) \phi_{r, \nu}(x;\epsilon) \xxe
	 + \delta_{r,r'} \nu  \xxa\bigg((2\kappa)^2\sum_{n,m\in\mathbb{Z}}\ee^{2\kappa \ii(n+m)rx - 2\kappa |n+m|\epsilon}   \tilde{b}_{r,n} \tilde{b}_{r,m}\bigg) \phi_{r, \nu}(x;\epsilon) \xxe.
\end{align*}
Clearly, $Y_1 = \nu^3 \delta_{r,r'} c_\epsilon  \phi_{r, \nu}(x;\epsilon)$, where $c_\epsilon$ is given by (\ref{ceps}).

The next step is to simplify the expression for $Y_3$. 
Using the definition \eqref{btilde} of $\tilde{b}_{r,n}$ and (\ref{def:chiralbosons}), we write
$$\rho_r (x;\eps) =  \sum\limits_{n\in\Z} 2\kappa \ee^{2\kappa(\ii r nx-|n|\epsilon)} \tilde{b}_{r,n}  \quad (r=\pm).$$
Thus,
\begin{equation}\label{Edef}
\xxa (2\kappa)^2 \sum_{n,m\in \Z} \ee^{2\kappa( \ii (n+m)x-|n+m|\epsilon)} \tilde{b}_{r,n}\tilde{b}_{r,m}-\rho_r(x;\epsilon)^2\xxe \; =E, 	
\end{equation}
where $E$ is short-hand notation for
\begin{align*}
E & \coloneqq (2\kappa)^2  \sum_{n,m\in\mathbb{Z}} \ee^{2\kappa \ii(n+m)r x}\left(\ee^{-2\kappa |n+m|\epsilon}-\ee^{-2\kappa (|n|+|m|)\epsilon}\right)\xxa \tilde{b}_{r,n} \tilde{b}_{r,m}\xxe  
	\\ 
 & = - 2 (2\kappa)^2\sum_{n,m=1}^{\infty} \ee^{-2\kappa \ii(n-m) r x}\left(\ee^{-2\kappa (n+m)\epsilon}-\ee^{-2\kappa |n-m|\epsilon}\right)\xxa b_{r,-n} b_{r,m}\xxe.
\end{align*}
Utilizing \eqref{Kprojectionderivatives2} and \eqref{Edef} and adding and subtracting $\ii r \delta_{r,r'}\xxa \rho_{r}'(x;\epsilon)\phi_{r,\nu}(x;\epsilon)\xxe $, we may write
\begin{align*}\nonumber
Y_3 = &\; \ii r \nu^2\delta_{r,r'} \xxa \rho_r'(x;\epsilon) \phi_{r, \nu}(x;\epsilon) \xxe
 + \delta_{r,r'} \nu  \xxa \bigg((2\kappa)^2\sum_{n,m\in\mathbb{Z}}\ee^{2\kappa \ii(n+m)x - 2\kappa |n+m|\epsilon}   \tilde{b}_{r,n} \tilde{b}_{r,m}\bigg) \phi_{r, \nu}(x;\epsilon) \xxe
 	\\
= &\; \ii r (\nu^2 - 1)\delta_{r,r'} \xxa \rho_r'(x;\epsilon) \phi_{r, \nu}(x;\epsilon) \xxe
+ \ii r\delta_{r,r'} \xxa \rho_r'(x;\epsilon) \phi_{r, \nu}(x;\epsilon) \xxe
	\\
&  + \delta_{r,r'} \nu  \xxa \rho_r(x;\epsilon)^2  \phi_{r, \nu}(x;\epsilon) \xxe
 + \delta_{r,r'} \nu  \xxa E \phi_{r, \nu}(x;\epsilon) \xxe
 	\\
= &\; 
\ii r (\nu^2 - 1)\delta_{r,r'} \xxa \rho_r'(x;\epsilon) \phi_{r, \nu}(x;\epsilon) \xxe
- \frac{1}{\nu} \delta_{r,r'} (\phi_{r,\nu})''(x;\epsilon) 
 + \delta_{r,r'} \nu  \xxa E \phi_{r, \nu}(x;\epsilon) \xxe.
\end{align*}

To establish (\ref{W3commutators}), it only remains to show that $Y_2$ together with the above term $\delta_{r,r'} \nu  \xxa E \phi_{r, \nu}(x;\epsilon) \xxe$ combine to give the term with $\mathcal{R}_{r,\nu}$ in (\ref{W3commutators}); more precisely, it remains to show that
\begin{align*}
& 2\nu\xxa \mathcal{R}_{r,\nu}(x;\epsilon)\phi_{r,\nu}(x;\epsilon)\xxe
\;  =  \nu  \xxa E \phi_{r, \nu}(x;\epsilon) \xxe  + Y_2.
\end{align*}
But this identity holds provided that
\begin{align*}
  \mathcal{R}_{r,\nu}(x;\epsilon)
 = &\;  \frac{1}{2}  E
 - 4 \kappa^2\nu   \sum_{n,m=1}^{\infty} s_n^2\left(\ee^{2\kappa \ii mr x} b_{r,m}-\ee^{-2\kappa \ii m r x} b_{r,-m}\right) \left(\ee^{-2\kappa |n-m|\epsilon}-\ee^{-2\kappa (n+m)\epsilon} \right)\ee^{-2\kappa n\epsilon} 
 	\\
  = &\;  (2\kappa)^2\sum_{n,m=1}^{\infty} \bigg( \nu  s_n^2\left(\ee^{2\kappa \ii mr x} b_{r,m}-\ee^{-2\kappa \ii m r x} b_{r,-m}\right) \ee^{-2\kappa n\epsilon}
-\ee^{-2\kappa \ii(n-m) r x}\xxa b_{r,-n} b_{r,m}\xxe \bigg)
	\\
& \times \left(\ee^{-2\kappa (n+m)\epsilon} - \ee^{-2\kappa |n-m|\epsilon} \right),
\end{align*}
which agrees with the expression \eqref{Rrnuxepsilondef} for $\mathcal{R}_{r,\nu}(x;\epsilon)$. This completes the proof of the lemma.

\subsection{Proof of Lemma \ref{Clemma}}\label{Clemmaapp}
We first establish the following:

\textit{For all sesquilinear forms $\Phi$ of the form (\ref{Phidef}), we have}
\begin{equation}\label{lemma3result}
\mathcal{C}\Phi+\Phi\mathcal{C}=2 \xxa\Phi \mathcal{C}\xxe-\mathrm{i}\xxa J''(\alpha,\beta)\Phi\xxe,
\end{equation}
\textit{where}
\begin{equation}
J''(\alpha) {\coloneqq} -
(2\kappa)^2  \sum_{r=\pm}\sum_{n\in\mathbb{Z}}   
n^2 \alpha_{r, n}b_{r,-n}.
\end{equation}

This is proved as follows. Let $\Phi$ be of the form (\ref{Phidef}). By \eqref{Phinormalordered},
\begin{equation}
\Phi = \ee^{\ii \sum_{r=\pm}\alpha_{r,0} Q_r/2} R_+^{\mu_+}R_-^{\mu_-} \ee^{\ii \sum_{r=\pm}\alpha_{r,0} Q_r/2}\ee^{\ii J^{+}(\alpha)}\ee^{\ii J^{-}(\alpha)},
\end{equation}
where $J^{\pm}(\alpha)$ are defined in (\ref{projectionsdefinition}).
We now assume $r \in \{\pm\}$ and $n>0$. By (\ref{Phinormalordered}) and (\ref{BCH}), we have $e^{t a_{r,\pm n}}\Phi = e^{\mathrm{i}t[a_{r,\pm n},J^{\pm}(\alpha)]}\Phi e^{ta_{r,\pm n}}$, and so
\begin{equation}
a_{r,\pm n}\Phi= \partial_t e^{t a_{r,\pm n}}\Phi |_{t=0}=\partial_te^{\mathrm{i}t[a_{r, \pm n},J^{\pm}(\alpha)]}\Phi e^{ta_{r,\pm n}} |_{t=0}
= \Phi a_{r,\pm n}+\mathrm{i}[a_{r, \pm n},J^{\pm}(\alpha)] \Phi.
\end{equation}
Using the commutators
\begin{align*}
[a_{r,\pm n},J^{\pm}(\alpha)]
& = \bigg[a_{r,\pm n}, \sum_{r'=\pm}\sum_{m=1}^{\infty} ( \alpha_{r', \pm m}c_m - \alpha_{-r',\mp m}s_m) a_{r', \mp m}\bigg]
 = \pm n ( \alpha_{r, \pm n}c_n - \alpha_{-r,\mp n}s_n),
\end{align*}
we find
\begin{align}\label{aPhicommutators}
[a_{r,\pm n}, \Phi ] = \pm \mathrm{i}  n ( \alpha_{r, \pm n}c_n - \alpha_{-r,\mp n}s_n)  \Phi. \end{align}
Recalling the definition of normal ordering and the definition \eqref{cC} of $\mathcal{C}$, we have
\begin{equation}
\xxa \mathcal{C}\Phi \xxe
\; =
(2\kappa)^2 \sum_{r=\pm}\sum_{n=1}^{\infty} na_{r,-n}\Phi a_{r,n},
\end{equation}
which together with \eqref{cC} and (\ref{aPhicommutators}) implies 
\begin{align*}
\mathcal{C}\Phi+\Phi\mathcal{C}-2\xxa\mathcal{C}\Phi\xxe \,
= &\; (2\kappa)^2 \sum_{r=\pm}\sum_{n=1}^{\infty} n \Big(a_{r,-n} a_{r,n} \Phi + \Phi a_{r,-n} a_{r,n} - 2 a_{r,-n} \Phi a_{r,n}  \Big)
	\\
= &\; (2\kappa)^2 \sum_{r=\pm}\sum_{n=1}^\infty n\Big(a_{r,-n}[a_{r,n}, \Phi] 
 + [\Phi, a_{r,-n}]a_{r,n} \Big)
	\\
= &\;  (2\kappa)^2  \sum_{r=\pm}\sum_{n=1}^\infty \mathrm{i} 
n^2 \Big(a_{r,-n} ( \alpha_{r, n}c_n - \alpha_{-r, -n}s_n)  \Phi + ( \alpha_{r, - n}c_n - \alpha_{-r, n}s_n)  \Phi a_{r,n} \Big)
	\\
= &\;  (2\kappa)^2  \sum_{r=\pm}\sum_{n\in\mathbb{Z}}\mathrm{i} 
n^2 \xxa \Big(- \alpha_{r, n}s_n a_{-r, n}  +  \alpha_{r, n} c_n a_{r,-n}  \Big)\Phi\xxe
	\\
= &\;  (2\kappa)^2 \sum_{r=\pm} \sum_{n\in\mathbb{Z}}  \mathrm{i} 
n^2 \xxa \alpha_{r, n}b_{r,-n}  \Phi \xxe,
\end{align*}
which is equivalent to \eqref{lemma3result}.

Now let $\Phi \coloneqq \; \xxa \phi^N_{\vr,\vm}(\vx;\eps) \xxe $.
As in \eqref{Phidef}, we can write
$$\Phi = \; \xxa R_+^{\mu_+}R_-^{{\mu_-}}\ee^{\ii J}\xxe$$
where 
$$\mu_\pm = \sum_{j = 1}^N \delta_{\pm,r_j} m_j \nu/\nu_0$$ 
and
\begin{align*}
J &= -\sum_{j=1}^Nm_j\nu (2r_j \kappa Q_{r_j} x_j+K_{r_j}(x_j;\epsilon))
	\\
& = -\sum_{j=1}^Nm_j\nu 2r_j \kappa Q_{r_j} x_j
-\sum_{j=1}^Nm_j\nu \sum_{n\in\Z_{\neq 0}}\frac{1}{\ii n} \ee^{2\kappa(\ii  n r_j x_j -|n|\epsilon)} b_{r_j,n}.
\end{align*}
In this case, since
$$
\rho_r'(x;\eps) = (2\kappa)^2 \sum_{n\in\Z_{\neq 0}} \ii r n \ee^{2\kappa(\ii r nx-|n|\epsilon)} b_{r,n},$$
we have
\begin{align}
J''=& -\ii (2\kappa)^2 \sum_{j=1}^Nm_j\nu \sum_{n\in\Z_{\neq 0}} n \ee^{2\kappa(\ii  n r_j x_j -|n|\epsilon)} b_{r_j,n}
= -\sum_{j=1}^N m_j\nu r_j \rho_{r_j}'(x_j;\epsilon),
\end{align}
so (\ref{lemma3result}) can be written as
\begin{align}\label{CPhiPhiC}
\mathcal{C}\Phi + \Phi\mathcal{C} 
= 2\xxa \Phi \mathcal{C} \xxe+ \ii \sum_{j=1}^N m_j\nu r_j \xxa \rho_{r_j}'(x_j;\epsilon) \Phi \xxe.
\end{align}

By (\ref{normalorderanyons}), we have
\begin{align}\label{phiBPhi}
\phi^N_{\vr, \vm} = B \Phi,
\end{align}
where we use the shorthand notation
$$\phi^N_{\vr, \vm} \equiv \phi^N_{\vr, \vm}(\vx;\eps), \qquad B \equiv B^N_{\vr, \vm}(\vx;\eps) := \prod_{1\leq j<k\leq N}\tet_{r_j,r_k}(x_j-x_k;\eps_j+\eps_k)^{m_jm_k \nu^2}.$$
Hence, using \eqref{CPhiPhiC}, we find
\begin{align*}\nonumber
[\mathcal{C}, \phi^N_{\vr,\vm}]
= & -2\phi^N_{\vr,\vm}\mathcal{C} + \mathcal{C}\phi^N_{\vr,\vm}
+ \phi^N_{\vr,\vm} \mathcal{C} 
	\\
=&  -2\phi^N_{\vr,\vm} \mathcal{C} 
+ B\big(\mathcal{C} \Phi + \Phi \mathcal{C} \big)
	\\
=&  -2\phi^N_{\vr,\vm} \mathcal{C} 
+ 2 \xxa B \Phi \mathcal{C} \xxe + \ii \sum_{j=1}^N m_j\nu r_j \xxa \rho_{r_j}'(x_j;\epsilon) B\Phi \xxe.
\end{align*}
Thus, using (\ref{phiBPhi}) again,
\begin{align*}\nonumber
[\mathcal{C}, \phi^N_{\vr,\vm}]
=&-2\phi^N_{\vr,\vm}\mathcal{C}+2\xxa\mathcal{C}\phi^N_{\vr,\vm}\xxe
	\\ \nonumber
&+\mathrm{i} \sum\limits_{j=1}^N m_j\nu r_j \left[(\rho_{r_j}^+)'(x_j;\epsilon)\phi^N_{\vr,\vm}
+\phi^N_{\vr,\vm} (\rho_{r_j}^-)'(x_j;\epsilon)\right].
\end{align*}
Upon right-multiplication by $\Omega$, only the second line survives. Hence,
\begin{align}\label{CPhiOmega}
(1-\nu^2)[\mathcal{C}, \phi^N_{\vr,\vm}]\Omega=& \; (1-\nu^2) \mathrm{i} \sum\limits_{j=1}^N m_j\nu r_j \left[(\rho_{r_j}^+)'(x_j;\epsilon)\phi^N_{\vr,\vm}
+\phi^N_{\vr,\vm} (\rho_{r_j}^-)'(x_j;\epsilon)\right] \Omega.
\end{align}
Differentiation of (\ref{rhoKcommutator}) with respect to $x$ gives
\begin{equation}\label{rhoprimeKcommutator}
 r[(\rho_{r}^{\mp})'(x;\epsilon),K_{r'}^{\pm}(x';\epsilon')]=\begin{cases}
\pm C''(\pm r(x-x');\epsilon+\epsilon') & r=r' \\
\pm \tilde{C}''(\pm(x-x');\epsilon+\epsilon') & r=-r'.	
\end{cases}	
\end{equation}
By (\ref{expC}) and (\ref{expCt}),
$$C(rx;\epsilon) =  -\ii \kappa r x - \log \ttet_1(\kappa r x,q; \kappa \epsilon), 
\qquad \tilde{C}(x;\eps) = -\log \ttet_4(\kappa x,q; \kappa \epsilon),$$
and so, by (\ref{tetrr}) and Definition \ref{def:wprreps},
$$C''(r x;\epsilon) =  - \partial_x^2 \log \theta_{r,r}(x; \epsilon) = \wp_{r,r}(x; \epsilon), 
\qquad \tilde{C}''(x;\eps) = -\partial_x^2 \log \theta_{r,-r}(x; \epsilon) = \wp_{r,-r}(x; \epsilon).$$
Hence (\ref{rhoprimeKcommutator}) can be written as
\begin{equation}\label{rhoprimeKcommutator2}
 r[(\rho_{r}^{\mp})'(x;\epsilon),K_{r'}^{\pm}(x';\epsilon')] = \pm \wp_{r,r'}(\pm(x-x'); \epsilon + \epsilon').
\end{equation}
In view of (\ref{anyondefinition}), it follows that
\begin{equation}
 r[(\rho_{r}^{\mp})'(x;\epsilon),\phi_{r',\nu'}(x';\epsilon')] = \mp \mathrm{i}\nu' \wp_{r,r'}(\pm(x-x'); \epsilon + \epsilon')\phi_{r',\nu'}(x';\epsilon')
\end{equation}
and from here we compute
\begin{align*}
&r_j \Big[ (\rho_{r_j}^+)'(x_j;\epsilon)\phi^N_{\vr,\vm}
+\phi^N_{\vr,\vm} (\rho_{r_j}^-)'(x_j;\epsilon)\Big]
	\\
= &\;r_j \phi_{r_1, m_1\nu}(x_1;\epsilon)\cdots\xxa  \rho_{r_j}'(x_j;\epsilon)\phi_{r_j, m_j\nu}(x_j;\epsilon)\xxe\cdots\phi_{r_N, m_N\nu}(x_N;\epsilon)
	\\
&+ \mathrm{i}\bigg(\sum\limits_{k=1}^{j-1}  m_k \nu \wp_{r_j, r_k}(x_k-x_j; 2\epsilon)
+\sum\limits_{k=j+1}^{N} m_k \nu \wp_{r_j, r_k}(x_j-x_k; 2\epsilon)
\bigg)\phi^N_{\vr,\vm}.
\end{align*}
Substituting the above expression into (\ref{CPhiOmega}) and using the identities
\begin{align*}
& (1-\nu^2)m_j = -(m_j^2\nu^2 -1), \qquad j = 1, \dots, N,
\end{align*}
we find
\begin{align}\nonumber
& (1-\nu^2)[\mathcal{C}, \phi^N_{\vr,\vm}]\Omega
	\\ \nonumber
=& \; - \nu\sum_{j=1}^N \mathrm{i}r_j(m_j^2\nu^2 -1)
\phi_{r_1, m_1\nu}(x_1;\epsilon)\cdots\xxa  \rho_{r_j}'(x_j;\epsilon)\phi_{r_j, m_j\nu}(x_j;\epsilon)\xxe\cdots\phi_{r_N, m_N\nu}(x_N;\epsilon) \Omega
	\\\nonumber
&+ (1-\nu^2) \mathrm{i} \sum_{j=1}^N m_j\nu \mathrm{i}\bigg(\sum\limits_{k=1}^{j-1}  m_k \nu \wp_{r_j, r_k}(x_k-x_j; 2\epsilon)
+\sum\limits_{k=j+1}^{N} m_k \nu \wp_{r_j, r_k}(x_j-x_k; 2\epsilon)
\bigg)\phi^N_{\vr,\vm}\Omega
	\\ 
= &  -2 \mathcal{X} \Omega
- 2(1-\nu^2) \nu^2 \sum_{1 \leq j < k \leq N} m_j m_k \wp_{r_j, r_k}(x_j-x_k; 2\epsilon) \phi^N_{\vr,\vm}\Omega.
\end{align}
Since $\wp_{r, r'}(x; \epsilon)$ is an even function of $x$, equation (\ref{Ccommutator}) follows. This completes the proof.

\subsection{Proof of Lemma \ref{lemma6.5}}\label{C4app}
It follows from \eqref{Phiexpectation} that $\langle\Omega,  [\cH_{3,\nu},\Phi_{\mu'}(\alpha)] \Omega\rangle=0$ unless $\mu_+' = \mu_-' = 0$. We will therefore henceforth assume that $\Phi = \Phi_0(\alpha) = \; \xxa \ee^{\ii J(\alpha)}\xxe$.
Since $\mathcal{C} \Omega = 0$, we have
$$\cH_{3,\nu} \Omega = \frac{\nu}{2} \sum_{r=\pm} W_{3, r} \Omega.$$
Let
$$W_{3,r}^\pm \coloneqq \lim_{\eps\to 0^+}  \frac{1}{6\pi} \int_{-\ell}^{\ell} \rho_r^\pm(x;\eps)^3\, \dd{x}.$$ 
By (\ref{W3pm2}) and the definition of normal-ordering, for $r=\pm$, we obtain
\begin{align*}
W_{3,r}\Omega 
& =\lim_{\eps\to 0^+}  \frac1{6\pi} \int_{-\ell}^{\ell} \xxa (2\kappa Q_r + \rho_r^+(x;\eps) + \rho_r^-(x;\eps))^3\xxe \dd{x} \,\Omega
	\\
& =\lim_{\eps\to 0^+}  \frac{1}{6\pi} \int_{-\ell}^{\ell} \xxa \rho_r^+(x;\eps)^3 \xxe \dd{x}\, \Omega
= W_{3,r}^+ \Omega.
\end{align*}
Hence, since $(W_{3,r}^+)^\dagger = W_{3,r}^-$,
\begin{align*}
\langle\Omega,  [\cH_{3,\nu},\Phi] \Omega\rangle
& = \langle \Omega, \cH_{3,\nu} \Phi - \Phi \cH_{3,\nu}  \Omega\rangle
= \frac{\nu}{2} \sum_{r=\pm}  \langle \Omega, (W_{3,r}^- \Phi - \Phi W_{3,r}^+)\Omega\rangle,
\end{align*}
so it is enough to show that $ \sum_{r=\pm} \langle \Omega, (W_{3,r}^- \Phi - \Phi W_{3,r}^+)\Omega\rangle = 0$.

By a straightforward but lengthy calculation, it may be shown that 
\begin{align}\label{eq:W3pm}
W_{3,r}^\pm = &\; (2\kappa)^2
\sum_{n,m=1}^\infty 
 \bigg(- c_{n} c_{m}  s_{n + m} 
 a_{r,\mp n} 
a_{r,\mp m} 
a_{-r ,\mp (n + m)} 
 + s_{n}  s_{m} c_{n + m} 
a_{-r ,\mp n}   
a_{-r ,\mp m}   
a_{r,\mp (n + m)} \bigg).
\end{align}

By (\ref{aPhicommutators}), for $n>0$, we have
\begin{equation}[a_{r, - n}, \Phi ] = -\mathrm{i}  n ( \alpha_{r, - n}c_n - \alpha_{-r, n}s_n)  \Phi \end{equation}
and, consequently,
\begin{equation}a_{r, n} \Phi \Omega =  \mathrm{i}  n ( \alpha_{r, n}c_n - \alpha_{-r, -n}s_n)  \Phi \Omega
=  \mathrm{i}  n \beta_{r,n}  \Phi \Omega, \end{equation}
where
\begin{equation}\label{betadefinition}
\beta(\alpha)_{r,n}\coloneqq c_n \alpha_{r,n} - s_n \alpha_{-r ,-n} \quad (n\in\Z_{\neq 0}).
\end{equation}	
Together with \eqref{eq:W3pm}, this gives (with $\Phi = \Phi_{0}(\alpha)$)
\begin{align}\nonumber 
\langle \Omega, W_{3,r}^- \Phi_{0}(\alpha) \Omega \rangle
= &\; (2\kappa)^2
\sum_{n,m=1}^\infty 
\bigg\langle \Omega, 
 \bigg(- c_{n} c_{m}  s_{n + m} 
 a_{r, n} 
a_{r, m} 
a_{-r , (n + m)} 
	\\\nonumber
& + s_{n_1}  s_{n} c_{n + m} 
a_{-r , n}   
a_{-r , m}   
a_{r, (n + m)} \bigg)
 \Phi_{0}(\alpha) \Omega \bigg\rangle
 	\\\nonumber
= &\; (2\kappa)^2
\sum_{n,m=1}^\infty 
\mathrm{i}^3 n m (n+m)
 \bigg(- c_{n} c_{m}  s_{n + m} 
\beta(\alpha)_{r, n}
 \beta(\alpha)_{r, m}
 \beta(\alpha)_{-r, (n + m)}
 	\\ \label{OmegaWPhiOmega}
& + s_{n}  s_{m} c_{n + m} 
 \beta(\alpha)_{-r, n}
 \beta(\alpha)_{-r, m}
 \beta(\alpha)_{r, (n + m)}
 \bigg)
 \big\langle \Omega, \Phi_{0}(\alpha) \Omega \big\rangle.
\end{align}

We use symmetry to compute $\langle \Omega,\Phi_0(\alpha)W_{3,r}^+\Omega\rangle$ from \eqref{OmegaWPhiOmega}. By (\ref{Phiadjoint}),
\begin{align}\label{Phiadj}
\Phi_0(\alpha)^\dag = \Phi_{0}(-\alpha^*) 
\end{align}
so that
\begin{equation}\langle \Omega, \Phi_0(\alpha) W_{3,r}^+  \Omega \rangle
= \overline{\langle \Omega, W_{3,r}^- \Phi_0(\alpha)^\dagger  \Omega \rangle}
= \overline{\langle \Omega, W_{3,r}^- \Phi_{0}(-\alpha^*) \Omega \rangle}.
\end{equation}
Thus, replacing $\alpha \to -\alpha^*$ in \eqref{OmegaWPhiOmega} yields
\begin{align*}
\langle \Omega, W_{3,r}^- \Phi_{0}(-\alpha^*) \Omega \rangle
= &
(2\kappa)^2
\sum_{n,m=1}^\infty 
\mathrm{i}^3 nm (n+m)
 \bigg(- c_{n} c_{m}  s_{n + m} 
\beta(-\alpha^*)_{r, n}
 \beta(-\alpha^*)_{r, m}
 \beta(-\alpha^*)_{-r, (n + m)}
 	\\
& + s_{n}  s_{m} c_{n + m} 
 \beta(-\alpha^*)_{-r, n}
 \beta(-\alpha^*)_{-r, m}
 \beta(-\alpha^*)_{r, (n + m)}
 \bigg)
 \big\langle \Omega, \Phi_{0}(-\alpha^*) \Omega \big\rangle.
\end{align*}
Taking the complex conjugate and using that
\begin{equation*}\overline{\beta(-\alpha^*)_{r,n}} = -c_n \alpha_{r,-n} + s_n \alpha_{-r ,n} = -\beta(\alpha)_{r,-n},\end{equation*}
which follows from \eqref{betadefinition} and \eqref{eq:star}, 
we arrive at
\begin{align}\nonumber
& \langle \Omega, \Phi_0(\alpha) W_{3,r}^+  \Omega \rangle = \overline{\langle \Omega, W_{3,r}^- \Phi_{0}(-\alpha^*) \Omega \rangle}
	\\ \nonumber
&= (2\kappa)^2
\sum_{n,m=1}^\infty 
\mathrm{i}^3 nm (n+m)
 \bigg(- c_{n} c_{m}  s_{n + m} 
\beta(\alpha)_{r, -n}
\beta(\alpha)_{r, -m}
\beta(\alpha)_{-r, -(n + m)}
 	\\ \label{OmegaPhiWOmega}
&\phantom{=\;} + s_{n}  s_{m} c_{n + m} 
\beta(\alpha)_{-r, -n}
\beta(\alpha)_{-r, -m}
\beta(\alpha)_{r, -(n + m)}
 \bigg)
 \big\langle \Omega, \Phi_{0}(-\alpha^*) \Omega \big\rangle.
\end{align}
Note that $\big\langle \Omega, \Phi_{0}(\alpha) \Omega \big\rangle = 1 = \big\langle \Omega, \Phi_{0}(-\alpha^*) \Omega \big\rangle$ by \eqref{Phiexpectation}.
Consequently, combining (\ref{OmegaWPhiOmega}) and (\ref{OmegaPhiWOmega}) gives
\begin{align*}
& \langle \Omega, W_{3,r}^- \Phi_{0}(\alpha) \Omega \rangle
- \langle \Omega, \Phi_0(\alpha) W_{3,r}^+  \Omega \rangle
	\\
& =  
(2\kappa)^2
\sum_{n,m=1}^\infty 
\mathrm{i}^3 nm (n+m)
	\\
&\phantom{=\;}\times \bigg(-c_{n} c_{m}  s_{n + m} 
\Big(\beta(\alpha)_{r, n}
 \beta(\alpha)_{r, m}
 \beta(\alpha)_{-r, (n + m)}
 - \beta(\alpha)_{r, -n}
\beta(\alpha)_{r, -m}
\beta(\alpha)_{-r, -(n + m)}\Big)
 	\\
& \phantom{=\;} + s_{n}  s_{m} c_{n + m} 
\Big( \beta(\alpha)_{-r, n}
 \beta(\alpha)_{-r, m}
 \beta(\alpha)_{r, (n + m)}
 -
\beta(\alpha)_{-r, -n}
\beta(\alpha)_{-r, -m}
\beta(\alpha)_{r, -(n + m)}
\Big) \bigg).
 \end{align*}
 Summing over $r = \pm$ and then letting $r \to -r$ in some terms, we obtain
 \begin{align*}
& \sum_{r=\pm} \Big(\langle \Omega, W_{3,r}^- \Phi_{0}(\alpha) \Omega \rangle
- \langle \Omega, \Phi_0(\alpha) W_{3,r}^+  \Omega \rangle\Big)
	\\
& = 
(2\kappa)^2
\sum_{n,m=1}^\infty 
\mathrm{i}^3 nm (n+m)
 \Big(-c_{n} c_{m}  s_{n + m} 
+ s_{n}  s_{m} c_{n + m} \Big)
	\\
&\phantom{=\;}\times\Big(\beta(\alpha)_{r, n}
 \beta(\alpha)_{r, m}
 \beta(\alpha)_{-r, (n + m)}
 - \beta(\alpha)_{r, -n}
\beta(\alpha)_{r, -m}
\beta(\alpha)_{-r, -(n + m)}\Big) = 0,
 \end{align*}
 where we have used the identity
 $$c_{n} c_{m}  s_{n + m} 
= s_{n}  s_{m} c_{n + m}$$
in the last step. This completes the proof of the lemma.

\section{Fermion representation of the model}\label{app:fermions} 
We make precise the discussion in Section~\ref{sec:conclusions} that $\cH_{3,\nu}$ in Definition~\ref{def:cH} can be interpreted as a Hamiltonian in a quantum field theory of fermions.

We start by recalling how the quantum field theory formally defined by the canonical anticommutator relation \eqref{CAR} and the fermion Hamiltonian $H_0$ in \eqref{H0} can be made mathematically precise, following \cite{langmann2015}. The pertinent quantum field algebra, $\cA_{\F}$, is a $*$-algebra with star operation $\dag$ generated by fermion field operators $\psi^\dag_{r,n}$ and $\psi_{r,n}\coloneqq (\psi^\dag_{r,n})^\dag$ ($r=\pm$, $n\in\Z +\frac{1}{2}$) satisfying the canonical anti-commutator relations 
\begin{equation}\label{CAR3}  
\{\psi_{r,n},\psi^\dag_{r',m}\} = \delta_{r,r'}\delta_{n,m},\quad \{\psi_{r,n},\psi_{r',m}\} =0 
\end{equation} 
for all $r,r'=\pm$ and $n,m\in\Z+\frac{1}{2}$. This algebra $\cA_{\F}$ has a natural representation on a fermion Fock space $\cF_{\F}$ which is fully determined by the following conditions: (i) for $A\in\cF_{\F}$, $A^\dag$ is the Hilbert space adjoint, (ii) the following highest weight conditions are fulfilled, 
\begin{equation}\label{HWC} 
\psi_{\pm,n}\tilde\Omega = \psi_{\pm,-n-1}^\dag\tilde\Omega=0\quad (n\in\Z_{\geq 0}), 
\end{equation} 
for some $\tilde\Omega\in\cF_{\F}$, (iii) the Hilbert space product $\langle\cdot,\cdot\rangle$ is determined by the condition $\langle\tilde\Omega,\tilde\Omega\rangle=1$. 
Using well-known arguments, one can construct from this the fermion Fock space $\cF_{\F}$; see \cite[Section~II.A]{langmann2015}. 

\begin{remark}\label{rem:fermionnotation} 
Note that $L$, $R_+$, $R_-$, $\Omega$,  $\hat J_{r}(2\pi r n/L)$,  and $\hat\psi_{r}( 2\pi r n/L)$ in  \cite{langmann2015} correspond to $2\ell$, $R_+$, $R_-^{-1}$, $\tilde\Omega$, $b_{r,n}$ and $\psi_{r,n}$ here (some of these identifications are suggested by comparing Eqs.~(2.20), (2.21c) and (2.38) in \cite{langmann2015} with our Eqs.\ \eqref{bR1R2} and \eqref{bR3}, other motivations are given in a footnote below). 
\end{remark} 

The fermions $\psi_r(x)$ characterized by the anti-commutator relations \eqref{CAR} can be obtained as limits $\eps\to 0^+$ of two different regularized fermions, $\psi_r(x;\eps)$ and $\tilde\psi_r(x;\eps)$, which both are useful in different ways. We start with the latter (the former is given further below),  
\begin{equation}\label{psirxeps} 
\tilde\psi_r(x;\eps)\coloneqq  \frac1{\sqrt{2\ell}} \sum_{n\in\Z + \frac{1}{2}} \psi_{r,n}\ee^{2\kappa( \ii r n x-|n|\eps)}\quad (r=\pm, \; x\in[-\ell,\ell])  
\end{equation} 
and $\tilde\psi^\dag_r(x;\eps)\coloneqq \tilde\psi_r(x;\eps)^\dag$ ($\eps>0$). These regularized fermions are analogous to the regularized bosons in \eqref{rhor}. In particular, for all $r,r'\in\{\pm\}$, $x,x'\in[-\ell,\ell]$, $\eps,\eps'>0$, they obey the anti-commutator relations 
\begin{equation} 
\{\tilde\psi_r(x;\eps),\tilde\psi^\dag_{r'}(x';\eps')\} = \delta_{r,r'}\delta_{\F}(x-x';\eps+\eps'),\quad \{\tilde\psi_r(x;\eps),\tilde\psi_{r'}(x';\eps')\} =0,
\end{equation} 
with the regularized antiperiodic Dirac delta $\delta_{\F}(x;\eps)$ given by
$$\delta_\F(x;\eps) \coloneqq \frac1{2\ell}\sum_{n\in\Z + \frac{1}{2}} \ee^{2\kappa (\ii n x-|n|\eps)}\quad (x\in\R,\eps>0);$$
this can be proved similarly as Lemma~\ref{lem:bosons}. 
Using additive normal ordering,\footnote{Note that we normal order with respect to $\Omega$ and not with respect to $\tilde{\Omega}$; one can show that this difference is of minor significance for additive normal ordering.}
\begin{equation}\label{nn}
\nna A\nne \, \coloneqq A-\langle \Omega,A\Omega\rangle , 
\end{equation} 
one can show by straightforward computations that
\begin{equation}\label{WFkr} 
W^{\F}_{k,r} \coloneqq \lim_{\eps\to 0}\int_{-\ell}^\ell \nna \tilde\psi^\dag_r(x;\eps)(-r\ii \partial_x)^{k-1}\tilde\psi_r(x;\eps)\nne \dd{x}\quad (r=\pm, \; k=1,2,3)  
\end{equation} 
is equal to 
\begin{equation}\label{WFkr1}  
W^{\F}_{k,r} = \sum_{n\in\Z+\frac{1}{2}} (2\kappa n)^{k-1}\nna\psi^\dag_{r,n}\psi_{r,n}\nne \quad (r=\pm, \; k=1,2,3). 
\end{equation} 
The representation in \eqref{WFkr1} makes manifest that $W^{\F}_{k,r}$ are  well-defined self-adjoint operators on $\cF_{\F}$ for $k=1,2,3$. 
For $k=2$,  this implies in particular that 
\begin{equation} 
\cH_2^{\F} \coloneqq \lim_{\eps\to 0}\sum_{r=\pm} \int_{-\ell}^\ell \nna \tilde\psi^\dag_r(x;\eps)r(-\ii\partial_x)\tilde\psi_r(x;\eps)\nne \dd{x} = 
\sum_{r=\pm}\sum_{n\in\Z+\frac{1}{2}} 2\kappa n\nna\psi^\dag_{r,n}\psi_{r,n}\nne 
\end{equation} 
is a self-adjoint operator bounded from below, and $\tilde\Omega$ is the groundstate of $\cH_2^{\F}$; see \cite[Lemma~2.3]{langmann2015} for further details. 
Let us stress an important detail: as explained in Remark~\ref{rem:vacua}, we have two vacua $\Omega$ and $\tilde\Omega$, with the first characterized by \eqref{aR3}, and the second by \eqref{bR3}.

There is a well-known collection of mathematical facts known as {\em boson-fermion correspondence} which states, loosely speaking, that the quantum field theory of fermions above can be mapped bijectively to the boson quantum field theory defined in Section~\ref{sec:anyonseCS}. In particular, the fermion Fock space $\cF_{\F}$ is identical with the boson Fock space $\cF$ underlying our construction, and the fermions with the relations in \eqref{CAR3}--\eqref{HWC} can be obtained from anyon operators as in Definition~\ref{def:anyons} in the special case $\nu=-1$, $\nu_0=1$ as follows, 
\begin{equation} 
\psi_{r,n} \coloneqq \lim_{\eps\to 0^+} \int_{-\ell}^\ell \frac1{\sqrt{2\ell}}\ee^{-2\kappa\ii n rx}\frac1{\sqrt{2\ell}}\xxa R_r^{-1}\ee^{\ii (2\kappa b_{r,0} rx + K_r(x;\eps/2))}\xxe \dd{x} \quad (r=\pm,n\in\Z +\tfrac{1}{2}) 
\end{equation} 
see Propositions~2.12 and 3.7 in \cite{langmann2015} for precise statements.\footnote{Note that, with our notation and definition \eqref{anyondefinition} of anyons, \cite[Eq.\ (3.2)]{langmann2015} can be written as 
\begin{equation*} 
\psi_r(x;\eps) = \frac1{\sqrt{2\ell}}\xxa R_r^{-1}\exp\Big(2\ii r\kappa \hat{J}_r(0)x + \sum_{n\neq 0} \frac{1}{n} \hat J_r(2\kappa r n) \ee^{2\kappa (\ii r nx-\eps|n|/2)}\Big)\xxe =\frac1{\sqrt{2\ell}}\phi_{r,-1}(x;\eps/2) 
\end{equation*} 
provided $\nu_0=1$,  $Q_r= \hat{J}_r(0)$ and $b_{r,n}=J_r(2\kappa r n)$ for $n\in\Z_{\neq 0}$, as already pointed out in Remark~\ref{rem:fermionnotation}. 
} 

Since 
\begin{equation} 
\psi_{r,n} \coloneqq \lim_{\eps\to 0^+} \int_{-\ell}^\ell \frac1{\sqrt{2\ell}}\ee^{-2\kappa\ii n rx} \tilde\psi_{r}(x;\eps)\,\dd{x} 
\end{equation} 
by \eqref{psirxeps}, 
\begin{equation}\label{tpsireps} 
\psi_r(x;\eps)\coloneqq \frac1{\sqrt{2\ell}}\xxa R_r^{-1} \ee^{\ii (2\kappa b_{r,0} rx + K_r(x;\eps/2))}\xxe 
\end{equation} 
is a regularization of the fermion $\psi_r(x)$. This regularization is useful to find representations of the operators in Definition~\ref{def:cH} in terms of fermions. In fact, as shown at the end of this appendix, the boson operators $Q_r$, $W_{2,r}$, and $W_{3,r}$ in Definition~\ref{def:cH} can be written in terms of the fermion operators $W^{\F}_{k,r}$ in \eqref{WFkr} as follows, 
\begin{equation}\label{W1rW1rF} 
Q_r = G^2\nu_0  W^{\F}_{1,r} \quad (r=\pm), 
\end{equation} 
\begin{equation}\label{W2rW2rF} 
W_{2,r}=  G^2 \Big( W^{\F}_{2,r} - \kappa(1/\nu_0-1)(\nu_0+1)Q_r W^{\F}_{1,r}\Big)  \quad (r=\pm), 
\end{equation} 
\begin{equation}\label{W3rW3rF}
W_{3,r}= G^2 \bigg( W^{\F}_{3,r} - 4\kappa(1/\nu_0-1)Q_r W^{\F}_{2,r} +  (2\kappa)^2\frac{(1/\nu_0-1)^2(2+\nu_0)}{3} Q_r^2W^{\F}_{1,r} +c_0W^{\F}_{1,r} \bigg),
\end{equation} 
and the operator $\cC$ in \eqref{cC} can be represented as 
\begin{multline}\label{cCF} 
\cC = - \lim_{\eps\to 0^+} G^4 \int_{-\ell}^{\ell} \fpint{-\ell}{\ell}  \sum_{r=\pm} \xxa J_r(x;\eps)\wp_1(x'-x)J_{r}(x';\eps)  \xxe \dd{x'}\,\dd{x} \\ - 
\lim_{\eps\to 0^+} G^4 \int_{-\ell}^{\ell} \int_{-\ell}^{\ell} \sum_{r=\pm} \xxa J_{-r}(x;\eps)\wp_1(x'-x+\ii\delta)J_{r}(x';\eps)  \xxe \dd{x'}\,\dd{x},
\end{multline} 
where the regularized fermion densities $J_r(x;\eps)$ are defined by
\begin{equation}\label{Jreps}  
J_r(x;\eps) \coloneqq \nna \tilde\psi_r^\dag(x;\eps)\tilde\psi_r (x;\eps)\nne \quad (r=\pm).
\end{equation} 

Before deriving the identities \eqref{W1rW1rF}--\eqref{cCF}, let us explain how they lead to the relations (\ref{H2H2F}) and (\ref{H3H3F}).
By (\ref{cH2}), (\ref{W1rW1rF}), and (\ref{W2rW2rF}),
$$
\cH_{2} 
= G^2 \sum_{r=\pm} \Big( W^{\F}_{2,r} - \kappa(1/\nu_0-1)(\nu_0+1)Q_r W^{\F}_{1,r}\Big)
= G^2 \mathcal{H}_2^{\F} - \kappa (1/\nu_0^2-1) \sum_{r=\pm} Q_r^2,
$$
which is (\ref{H2H2F}).
To obtain (\ref{H3H3F}), we substitute the expressions (\ref{W3rW3rF}) and (\ref{cCF}) for $W_{3,r}$ and $\mathcal{C}$ into the definition \eqref{cH3} of $\cH_{3,\nu}$ and then use (\ref{WFkr}) to simplify. 
This yields
\begin{multline}\label{cH3nufermion1} 
\cH_{3,\nu} =  \frac{\nu}{2}G^2\lim_{\eps\to 0^+}\int_{-\ell}^{\ell} 
\sum_{r=\pm} \nna \psi_r^\dag(x;\eps) \bigg( -\partial_x^2 + 2\tilde\kappa Q_r r\ii\partial_x + \frac{2+\nu_0}{3}(\tilde\kappa Q_r)^2 +c_0 \bigg) \psi_r(x;\eps)\nne \dd{x}   \\ + 
 \frac{1}2(\nu^2-1)G^4 \lim_{\eps\to 0^+} \int_{-\ell}^\ell \fpint{-\ell}{\ell}  \sum_{r=\pm} \xxa \tilde{J}_r(x;\eps)\wp_1(x'-x)\tilde{J}_r(x';\eps) \xxe \dd{x} \,\dd{x'} 
 \\ + 
 \frac{1}2(\nu^2-1)G^4 \lim_{\eps\to 0^+}\int_{-\ell}^\ell \int_{-\ell}^{\ell}  \sum_{r=\pm} \xxa \tilde{J}_r(x;\eps)\wp_1(x'-x+\ii\delta)\tilde{J}_{-r}(x';\eps) \xxe \dd{x} \,\dd{x'},
\end{multline} 
where $\tilde\kappa = 2\kappa (1/\nu_0-1)$. Recalling the definition (\ref{cH3nufermion}) of $\cH^{\F}_{3,\nu}$, the relation (\ref{H3H3F}) follows.

\begin{remark} A comment on the mathematical status of the results \eqref{W1rW1rF}--\eqref{cH3nufermion1} is in order. To keep the rest of this appendix short, we skip mathematical details which would be needed to make our derivations mathematically precise. For this reason, these results have the status of conjectures. However, we are confident that they can be proved. 
\end{remark}

\begin{proof}[Derivation of \eqref{W1rW1rF}--\eqref{cCF}] Since $\psi_r(x;\eps)=\phi_{r,-1}(x;\eps/2)/\sqrt{2\ell}$ and $\psi^\dag_r(x;\eps)=\phi_{r,1}(x;\eps/2)/\sqrt{2\ell}$ for $\nu_0=1$ by \eqref{anyon_adjoints}, we can use Proposition~\ref{prop:anyons}(a) for $(\nu,\nu')=(-1,1)$ to compute,  for $x,x-a\in[-\ell,\ell]$ and $\eps>0$,
\begin{equation*} 
\psi_r^\dag(x;\eps)\psi_r(x-a;\eps) = \frac1{\ttet_1(\kappa ra,q;\kappa\eps)}\xxa \psi_r^\dag(x;\eps)\psi_r(x-a;\eps) \xxe \quad (r=\pm) 
\end{equation*} 
with $\ttet_1(x,q;\eps)$ in \eqref{ttet1ttet4}. Since $\langle\Omega,\xxa \psi_r^\dag(x;\eps)\psi_r(x-a;\eps) \xxe \Omega\rangle =1/2\ell$ by \eqref{Phiexpectation}, this and \eqref{tpsireps}--\eqref{nn} imply 
\begin{equation*}
\begin{split} 
\nna \psi_r^\dag(x;\eps)\psi_r(x-a;\eps) \nne \, =  & \frac{\xxa \psi_r^\dag(x;\eps)\psi_r(x-a;\eps) \xxe -1/2\ell}{\ttet_1(\kappa ra,q;\kappa\eps)} 
\\ = & \frac{\xxa \ee^{\ii (-2\kappa b_{r,0} ra + K_r(x-a;\eps/2)-K_r(x;\eps/2))}\xxe-I}{2\ell\ttet_1(\kappa ra,q;\kappa\eps)}.
\end{split} 
\end{equation*} 
The idea is to expand both sides of this identity in power series in $a$ and equate coefficients of $a^k$ for $k=0,1,2$. However, to get correct results, one has to take the limit $\eps\to 0^+$ {\em before} expanding in power series in $a$. Thus, in the following, we evaluate 
\begin{equation}\label{psirpsir}
\nna \psi_r^\dag(x)\psi_r(x-a) \nne \, =  \frac{\xxa \ee^{\ii (-2\kappa b_{r,0} ra +K_r(x-a)-K_r(x))}\xxe-I}{2\ell\ttet_1(\kappa ra,q)}
\end{equation} 
where the missing argument $\eps$ means that we set $\eps=0$. Since $\eps=0$, we do not claim that the following computations are mathematically precise.

Expanding the left-hand side in \eqref{psirpsir}, we find
\begin{equation}\label{LHS}
\begin{split} 
\nna \psi_r^\dag(x)\psi_r(x-a) \nne \; = &  \nna \psi_r^\dag(x)\psi_r(x) \nne - a\nna \psi_r^\dag(x)\partial_x\psi_r(x) \nne 
\\ & +\frac{a^2}{2} \nna \psi_r^\dag(x) \partial_x^2\psi_r(x) \nne + \; O(a^3).
\end{split} 
\end{equation} 
Next, we expand the denominator on the right-hand side in  \eqref{psirpsir}; using \eqref{ttet1ttet4} and the identity 
$$\sum_{m=1}^\infty q^{2m}/(1-q^{2m})^2=\sum_{n=1}^\infty nq^{2n}/(1-q^{2n}) \quad \text{for} \quad 0\leq q<1,$$ 
we obtain after straightforward computations that
\begin{multline}\label{tet1a}
  \frac{1}{\ttet_1(\kappa r a,q) }  = \frac{1}{(\ee^{-\ii \kappa ra}-\ee^{\ii \kappa ra})\prod_{m=1}^\infty \big(1-q^{2m}\ee^{2\ii\kappa r a}\big)\big(1-q^{2m}\ee^{-2\ii\kappa ra}\big)} 
  \\ = \frac1{G^2}\frac{\ii}{2\kappa ra}\big(1+\half c_0a^2 + O(a^4)\big) 
\end{multline} 
with $c_0$ in \eqref{constants} and $G$ in \eqref{G}.
To expand the numerator on the right-hand side in  \eqref{psirpsir} we compute, recalling from \eqref{rhor} that $
\rho_r (x) =  2\kappa Q_r  + r K_r'(x)$,  
$$
-2\kappa Q_r ra +K_r(x-a)-K_r(x)
= -ra\rho_r(x) +r\frac{a^2}{2}\rho_r'(x) -r\frac{a^3}{6}\rho_r''(x) + O(a^4).
$$
Using that $b_{r,0}=Q_r/\nu_0 = Q_r + (1/\nu_0-1)Q_r$ and introducing the notation 
\begin{equation}\label{checkrhor} 
\check{\rho}_r(x)\coloneqq \rho_r(x) + 2\kappa (1/\nu_0-1)Q_r, 
\end{equation} 
this yields 
\begin{align*} 
\xxa \ee^{\ii (-2\kappa b_{r,0} ra +K_r(x-a)-K_r(x))}\xxe  \;
=&\; I - \ii ra \check{\rho}_r(x) - \frac{a^2}{2}\xxa (\check{\rho}_r(x)^2 - \ii r\check{\rho}'_r(x))\xxe
	\\
 & + \ii r\frac{a^3}{6}\xxa (\check{\rho}_r(x)^3 - 3\ii r \check{\rho}_r(x)\check{\rho}'_r(x)-\check{\rho}''_r(x) )\xxe +O(a^4) .
\end{align*} 
Combining this with \eqref{tet1a} we obtain the following expansion of the right-hand side in \eqref{psirpsir},  
\begin{multline*} 
 \frac{\xxa \ee^{\ii (-2\kappa b_{r,0} ra +K_r(x-a)-K_r(x))}\xxe-I}{2\ell\ttet_1(\kappa ra,q)} 
 = \frac{1}{2\pi G^2}\xxa \Big(  \check{\rho}_r(x) - \frac{\ii ra}{2}(\check{\rho}_r(x)^2 -\ii r\check{\rho}'_r(x)) \\
 - \frac{a^2}{6}(\check{\rho}_r(x)^3 
 - 3\ii r \check{\rho}_r(x)\check{\rho}'_r(x)-\check{\rho}''_r(x)-3c_0\check{\rho}_r(x)) \Big)\xxe  + O(a^3). 
\end{multline*} 
By equating the coefficients of $a^k$ here and in \eqref{LHS} for $k=0,1,2$, we obtain the identities
\begin{equation}\label{ak0}
\nna \psi_r^\dag(x)\psi_r(x) \nne =  \frac{1}{2\pi G^2} \check{\rho}_r(x)  \quad (r=\pm) , 
\end{equation} 
\begin{equation}\label{ak1}  
 \nna \psi_r^\dag(x)(-r\ii \partial_x)\psi_r(x) \nne = \frac{1}{4\pi G^2}\xxa \check{\rho}_r(x)^2 -\ii r\check{\rho}'_r(x)\xxe  \quad (r=\pm) , 
 \end{equation} 
\begin{equation}\label{ak2}
 \nna \psi_r^\dag(x)(-\partial_x^2)\psi_r(x) \nne =  \frac{1}{6\pi G^2}\xxa \check{\rho}_r(x)^3 - 3\ii r\check{\rho}_r(x)\check{\rho}'_r(x)-\check{\rho}''_r(x)-3c_0\check{\rho}_r(x)\xxe   \quad (r=\pm) . 
\end{equation} 
We recall from \eqref{Wkdefinition} and \eqref{WFkr} that
\begin{equation} 
W_{k,r} = \frac1{2\pi k}\int_{-\ell}^\ell \xxa \rho_r(x)^k\xxe\dd{x},\quad W_{k,r}^{\F} = \int_{-\ell}^\ell  \nna \psi_r^\dag(x)(-r\ii \partial_x)^{k-1}\psi_r(x) \nne  \dd{x} \quad (k=1,2,3).
\end{equation} 
In particular, $W_{1,r}=Q_r$. Thus, integrating both sides of \eqref{ak0}, we obtain 
\begin{multline}\label{WF1r1}
W^{\F}_{1,r} = \frac1{2\pi G^2}\int_{-\ell}^\ell \check{\rho}_r(x) \dd{x} 
=  \frac1{2\pi G^2}\int_{-\ell}^\ell \left(  \rho_r(x) + 2\kappa (1/\nu_0-1)Q_r\right)\dd{x} \\
=  \frac1{G^2} \left( W_{1,r} + (1/\nu_0-1)Q_r \right) = \frac1{G^2\nu_0} Q_r
\end{multline} 
using $4\kappa\ell=2\pi$, which gives \eqref{W1rW1rF}. 
Similarly, integrating \eqref{ak1} and dropping the total derivative term, we obtain
\begin{multline}\label{WF2r1} 
W^{\F}_{2,r} =  \frac1{4\pi G^2}\int_{-\ell}^\ell \xxa \check{\rho}_r(x)^2\xxe  \dd{x} =  \frac1{4\pi G^2}\int_{-\ell}^\ell \xxa  \rho_r(x)^2 + 4\kappa(1/\nu_0-1)Q_r\rho_r(x) \\ + (2\kappa (1/\nu_0-1)Q_r)^2 \xxe  \dd{x} 
=  \frac1{G^2}W_{2,r} +  2\kappa(1/\nu_0-1)Q_r \frac1{G^2}W_{1,r} \\ + \frac1{G^2}\kappa(1/\nu_0-1)^2Q_r^2 
=  \frac1{G^2}W_{2,r}+ \frac1{G^2}\kappa(1/\nu_0^2-1)Q_r^2
\end{multline} 
since  $2\ell(2\kappa)^2 =4\pi\kappa$, which gives  \eqref{W2rW2rF}. 
Finally, by integrating \eqref{ak2}, 
\begin{multline} 
W^{\F}_{3,r} =  \frac1{6\pi G^2}\int_{-\ell}^\ell\left(  \xxa \check{\rho}_r(x)^3\xxe -3c_0\check{\rho}_r(x) \right)\dd{x} =  
\frac1{6\pi G^2}\int_{-\ell}^\ell \xxa  \rho_r(x)^3 + 6\kappa (1/\nu_0-1)Q_r  \rho_r(x)^2 \\ 
+ 3(2\kappa (1/\nu_0-1)Q_r)^2 \rho_r(x) + (2\kappa (1/\nu_0-1)Q_r)^3 \xxe  \dd{x}  - c_0W^{\F}_{1,r} \\
=  \frac1{G^2}W_{3,r} +  4\kappa(1/\nu_0-1)Q_r \frac1{G^2}W_{2,r} +(2\kappa)^2 \nu_0(1/\nu_0-1)^2Q_r^2 W^{\F}_{1,r} \\
+ \frac{1}{3G^2}(2\kappa)^2(1/\nu_0-1)^3Q_r^3 - c_0W^{\F}_{1,r}, 
\end{multline} 
using again $4\kappa\ell=2\pi$; inserting \eqref{W1rW1rF}--\eqref{W2rW2rF} and \eqref{WF1r1}, we obtain \eqref{W3rW3rF} by straightforward computations.
To justify \eqref{cCF}, note that  \eqref{checkrhor} and \eqref{ak0} suggest that 
\begin{equation}\label{rhorJr}
\rho_{r}(x;\eps) = 2\pi G^2 J_r(x;\eps)-2\kappa(1/\nu_0-1)Q_r + O(\eps)
\end{equation} 
with $J_r(x;\eps)$ given in \eqref{Jreps}. Inserting this into \eqref{cC2} we get  
\begin{equation} 
\label{cC4} 
\cC = -\lim_{\eps\to 0^+} G^4 \pi\int_{-\ell}^{\ell} \sum_{r=\pm} \xxa J_r(x;\eps)(TJ_{r,x})(x;\eps) +  J_{-r}(x;\eps)(\tilde{T}J_{r,x})(x;\eps)  \xxe \dd{x} 
\end{equation} 
(the term proportional to $Q_r$ on the right-hand side of \eqref{rhorJr} drops out because 
$$\int_{-\ell}^\ell (TJ_{r,x})(x;\eps)\,\dd{x} =0$$ 
as a consequence of Lemma \ref{lem:TT}, and similarly for $\tilde{T}$). Inserting the definitions \eqref{TT} of $T$, $\tilde{T}$, we get
\begin{multline*} 
\cC = -\lim_{\eps\to 0^+} G^4 \int_{-\ell}^{\ell} \fpint{-\ell}{\ell}  \sum_{r=\pm} \xxa J_r(x;\eps)\zeta_1(x'-x)\partial_{x'}J_{r}(x';\eps)\xxe \dd{x'}\,\dd{x}   \\
-  \lim_{\eps\to 0^+} G^4 \int_{-\ell}^{\ell}\int_{-\ell}^{\ell} \sum_{r=\pm} \xxa J_{-r}(x;\eps)\zeta_1(x'-x+\ii\delta)\partial_{x'}J_{r}(x';\eps)  \xxe \dd{x'}\,\dd{x}.
\end{multline*} 
By partial integration, using the $2\ell$-periodicity of the integrand and $\partial_z\zeta_1(z)=-\wp_1(z)=-\wp_{1}(-z)$, we obtain \eqref{cCF}. 
\end{proof}

\bibliographystyle{unsrt}	
\bibliography{CFTeCS.bib}

\end{document}